\PassOptionsToPackage{unicode}{hyperref}
\PassOptionsToPackage{hyphens}{url}
\PassOptionsToPackage{dvipsnames,svgnames,x11names}{xcolor}
\documentclass[12pt]{article}

\usepackage{amsmath,amssymb,amsthm}
\usepackage{placeins}
\usepackage{iftex}
\ifPDFTeX
  \usepackage[T1]{fontenc}
  \usepackage[utf8]{inputenc}
  \usepackage{textcomp} 
\else 
  \usepackage{unicode-math}
  \defaultfontfeatures{Scale=MatchLowercase}
  \defaultfontfeatures[\rmfamily]{Ligatures=TeX,Scale=1}
\fi
\usepackage{lmodern}
\ifPDFTeX\else  
\fi
\IfFileExists{upquote.sty}{\usepackage{upquote}}{}
\IfFileExists{microtype.sty}{
  \usepackage[]{microtype}
  \UseMicrotypeSet[protrusion]{basicmath} 
}{}
\makeatletter
\@ifundefined{KOMAClassName}{
  \IfFileExists{parskip.sty}{%
    \usepackage{parskip}
  }{
    \setlength{\parindent}{0pt}
    \setlength{\parskip}{6pt plus 2pt minus 1pt}}
}{
  \KOMAoptions{parskip=half}}
\makeatother
\usepackage{xcolor}
\makeatletter
\ifx\paragraph\undefined\else
  \let\oldparagraph\paragraph
  \renewcommand{\paragraph}{
    \@ifstar
      \xxxParagraphStar
      \xxxParagraphNoStar
  }
  \newcommand{\xxxParagraphStar}[1]{\oldparagraph*{#1}\mbox{}}
  \newcommand{\xxxParagraphNoStar}[1]{\oldparagraph{#1}\mbox{}}
\fi
\ifx\subparagraph\undefined\else
  \let\oldsubparagraph\subparagraph
  \renewcommand{\subparagraph}{
    \@ifstar
      \xxxSubParagraphStar
      \xxxSubParagraphNoStar
  }
  \newcommand{\xxxSubParagraphStar}[1]{\oldsubparagraph*{#1}\mbox{}}
  \newcommand{\xxxSubParagraphNoStar}[1]{\oldsubparagraph{#1}\mbox{}}
\fi
\makeatother
\usepackage{bm}

\usepackage{longtable,booktabs,array}
\usepackage{calc} 
\usepackage{etoolbox}
\makeatletter
\patchcmd\longtable{\par}{\if@noskipsec\mbox{}\fi\par}{}{}
\makeatother
\IfFileExists{footnotehyper.sty}{\usepackage{footnotehyper}}{\usepackage{footnote}}
\makesavenoteenv{longtable}
\usepackage{graphicx}
\makeatletter
\def\maxwidth{\ifdim\Gin@nat@width>\linewidth\linewidth\else\Gin@nat@width\fi}
\def\maxheight{\ifdim\Gin@nat@height>\textheight\textheight\else\Gin@nat@height\fi}
\makeatother
\setkeys{Gin}{width=\maxwidth,height=\maxheight,keepaspectratio}
\makeatletter
\def\fps@figure{htbp}
\makeatother

\makeatletter
\@ifpackageloaded{caption}{}{\usepackage{caption}}
\AtBeginDocument{%
\ifdefined\contentsname
  \renewcommand*\contentsname{Table of contents}
\else
  \newcommand\contentsname{Table of contents}
\fi
\ifdefined\listfigurename
  \renewcommand*\listfigurename{List of Figures}
\else
  \newcommand\listfigurename{List of Figures}
\fi
\ifdefined\listtablename
  \renewcommand*\listtablename{List of Tables}
\else
  \newcommand\listtablename{List of Tables}
\fi
\ifdefined\figurename
  \renewcommand*\figurename{Figure}
\else
  \newcommand\figurename{Figure}
\fi
\ifdefined\tablename
  \renewcommand*\tablename{Table}
\else
  \newcommand\tablename{Table}
\fi
}
\@ifpackageloaded{float}{}{\usepackage{float}}
\floatstyle{ruled}
\@ifundefined{c@chapter}{\newfloat{codelisting}{h}{lop}}{\newfloat{codelisting}{h}{lop}[chapter]}
\floatname{codelisting}{Listing}

\makeatother
\makeatletter
\@ifpackageloaded{caption}{}{\usepackage{caption}}
\@ifpackageloaded{subcaption}{}{\usepackage{subcaption}}
\makeatother

\ifLuaTeX
  \usepackage{selnolig}  
\fi
\usepackage[]{natbib}
\usepackage{bookmark}

\IfFileExists{xurl.sty}{\usepackage{xurl}}{} 
\hypersetup{
  pdftitle={Title},
  pdfauthor={Author 1; Author 2},
  pdfkeywords={3 to 6 keywords, that do not appear in the title},
  colorlinks=true,
  linkcolor={blue},
  filecolor={Maroon},
  citecolor={Blue},
  urlcolor={Blue},
  pdfcreator={LaTeX via pandoc}}
\usepackage{multirow}
\usepackage{booktabs}
\usepackage[ruled,vlined,linesnumbered]{algorithm2e}
\usepackage{enumerate}
\usepackage{hyperref}

\newcommand{\anon}{1}

\theoremstyle{plain}
\newtheorem{proposition}{Proposition}
\newtheorem{theorem}{Theorem}

\theoremstyle{definition}
\newtheorem{remark}{Remark}

\begin{document}

\def\spacingset#1{\renewcommand{\baselinestretch}%
{#1}\small\normalsize} \spacingset{1}


\if1\anon
{
  \title{\bf Robust Bayesian Quantile Regression: Targeted and Calibrated Posterior Inference}
  \author{Dongu Han\thanks{Corresponding author, hdongu94@gmail.com}\hspace{.2cm}\\
    School of Commerce, Meiji University\\
    and \\
    Genya Kobayashi \\
    School of Commerce, Meiji University\\
    and \\
    Shonosuke Sugasawa \\
    Faculty of Economics, Keio University}
  \maketitle
} \fi

\if0\anon
{
  \bigskip
  \bigskip
  \bigskip
  \begin{center}
    {\LARGE\bf Title}
\end{center}
  \medskip
} \fi

\bigskip

\begin{abstract}
Robustness to extreme observations does not guarantee correct quantile targeting or valid uncertainty quantification. We investigate these distinct objectives using a mixture of an asymmetric Laplace density and its log-Pareto scale mixture (AL--LPAL), combining two-sided log-regularly varying tails with a logarithmically unbounded center. We establish posterior propriety for every finite sample and finite coefficient posterior moments of all positive orders. Although the working density preserves its nominal quantile, misspecification can displace the posterior intercept target. Under jointly increasing sample size and contamination magnitude, allowing a nonvanishing contamination fraction and an unknown mixture weight, a draw-wise residual-quantile correction restores the uncontaminated target. We derive the local covariance structure induced by the logarithmic singularity and establish corrected-posterior contraction at rate $n^{-1/3}(\log n)^{-2/3}$. At this scale, the corrected posterior collapses around a random center with a Gaussian-process argmax sampling limit, yielding asymptotically zero coverage of equal-tailed credible intervals for contrasts whose limiting fluctuations have no atom at zero. We develop an interval-calibration procedure combining finite-scale Gaussian-process simulation with Bayesian-bootstrap residual quantiles and establish consistency of the plug-in Gaussian-process reference under specified conditions. Simulations and an airline-delay application demonstrate the practical benefits of targeting correction and uncertainty calibration under extreme contamination.
\end{abstract}

\noindent%
\vfill

\newpage
\spacingset{1.8} 

\section{Introduction}
\label{sec:introduction}

Bayesian quantile regression under extreme response contamination poses three distinct inferential questions: whether a working likelihood rejects the influence of increasingly extreme observations, whether its posterior targets the intended conditional quantile, and whether the resulting uncertainty statements have valid frequentist coverage. These properties need not coincide. In particular, posterior concentration around the correct regression coefficients does not ensure that posterior credible intervals reproduce the sampling uncertainty of an estimator. We investigate their relationship for a quantile-specific working likelihood with two-sided super-heavy tails and a logarithmically unbounded center.

The asymmetric Laplace (AL) likelihood connects Bayesian quantile regression to the check-loss formulation of \citet{Koenker1978}: its likelihood maximizer is the classical quantile-regression estimator \citep{Yu2001}, and its normal--exponential mixture representation facilitates posterior simulation \citep{Kozumi2011}. The AL distribution's exponential tails and restricted shape have motivated semiparametric extensions based on Dirichlet process mixtures \citep{Kottas2009} and quantile-constrained mixtures \citep{Reich2010}, as well as parametric alternatives involving generalized skewed distributions \citep{Wichitaksorn2014,GalarzaMorales2017}, skew exponential power errors \citep{Bernardi2018}, generalized asymmetric Huberised-type errors \citep{Hu2024}, and structured error mixtures \citep{Yan2026}. These specifications broaden distributional flexibility, but flexibility alone neither guarantees rejection of increasingly extreme observations nor ensures correct quantile targeting under misspecification.

A separate literature develops posterior robustness against diverging outliers, building on outlier-mixture modeling \citep{Box1968} and heavy-tailed robustness arguments \citep{OHagan1979}. Log-regularly varying errors have been studied in location--scale models \citep{Desgagne2015}, and whole robustness has been established for super-heavy-tailed regression \citep{Gagnon2020}. Particularly relevant is \citet{Hamura2022}, who combined a conventional regression component with a log-regularly varying scale mixture of normals and analyzed posterior robustness as fixed-sample outliers diverge. Other approaches, such as the synthetic posterior of \citet{Hashimoto2020}, modify Bayesian updating itself. Our focus is instead on joint sample-size and contamination asymptotics for a quantile-specific working likelihood whose unbounded center requires separate control of posterior integrability and local likelihood fluctuations.

Misspecified Bayesian inference raises a further distinction between the parameter targeted by a working likelihood and the uncertainty attached to that target. For conventional AL-based quantile regression, \citet{Sriram2013} established posterior consistency and convergence rates under misspecification, while \citet{Sriram2015} and \citet{Yang2016} developed adjustments for frequentist uncertainty assessment. More generally, misspecified posteriors may concentrate around Kullback--Leibler minimizers with conditional dispersion different from the sampling variation of their centers \citep{Kleijn2006,Kleijn2012,Muller2013}. Related perspectives include quasi-posterior estimation \citep{Chernozhukov2003} and generalized Bayesian updating \citep{Bissiri2016}. Adjustment of posterior dispersion, however, cannot repair an incorrectly targeted quantile.

Nonregular statistical inference has been developed through likelihood-ratio asymptotics \citep{Ibragimov1981}, cube-root estimation theory \citep{Kim1990}, and Laplace-type estimation with rate-adaptive inference \citep{Jun2015}. Nonregular convergence rates and uncertainty calibration also arise in Gibbs posteriors based on discontinuous loss functions \citep{Syring2017}. A related literature studies inference in singular statistical models, including cusp-type models under correct specification and misspecification \citep{Dachian2018,Chernoyarov2018} and location models with centrally unbounded densities \citep{Podgorski2015}. Our setting differs in that a logarithmic singularity is introduced through a potentially misspecified working likelihood while Bayesian updating retains the original unbounded likelihood. This singularity changes the local covariance scaling underlying conventional cube-root asymptotics.

We study these questions using the AL--LPAL working likelihood, which mixes an AL density with a log-Pareto scale mixture of AL densities. Mixing the AL scale directly preserves the specified quantile while generating log-regularly varying tails on both sides, including at asymmetric quantile levels. The tails attenuate parameter-dependent contributions from increasingly extreme residuals, whereas the logarithmic central singularity determines the local stochastic behavior of the likelihood. Despite its unbounded center, we establish posterior propriety for every finite sample and the existence of posterior moments of every positive order for the regression coefficients.

Our contributions are threefold. First, we show that the AL--LPAL working likelihood can induce intercept distortion under misspecification despite preserving its nominal quantile. A draw-wise residual-quantile correction restores the uncontaminated regression target, and we establish corrected-posterior contraction at rate $n^{-1/3}(\log n)^{-2/3}$ under jointly increasing sample size and contamination magnitude, allowing a potentially nonvanishing contamination fraction and an unknown mixture weight. We verify the population identification and curvature conditions over a connected family of generalized Gaussian error distributions and smooth perturbations of a quantile-adjusted Gaussian reference. Second, we characterize the nonregular posterior geometry induced by the logarithmic central singularity. By deriving the exact local covariance structure and controlling the integrated likelihood near its singularities, we establish that the corrected posterior collapses at the scale of its random center's sampling fluctuations, which converge after normalization to a Gaussian-process argmax law. Consequently, equal-tailed credible intervals have asymptotically zero frequentist coverage for scalar contrasts whose limiting fluctuations have no atom at zero, a failure more severe than conventional covariance mismatch. Third, we translate this nonregular theory into an interval-calibration procedure that combines finite-scale Gaussian-process argmax simulation, tailored to the local covariance and efficient curvature of the AL--LPAL likelihood, with Bayesian-bootstrap residual quantiles \citep{Rubin1981,Lo1987}. We establish consistency of the plug-in Gaussian-process reference and demonstrate the practical benefits of targeting correction and uncertainty calibration through simulations and a New York City airline-delay application.

Section~\ref{sec:al_lpal} introduces the AL--LPAL working likelihood and its distributional properties. Section~\ref{sec:theory} develops the posterior theory, Section~\ref{sec:computation} presents posterior computation and interval calibration, Section~\ref{sec:Empirical} reports the empirical studies, and Section~\ref{sec:discussion} discusses the implications and limitations of our findings.

\section{The AL--LPAL Working Likelihood}
\label{sec:al_lpal}
\label{subsec:al_lpal_regression}

For a fixed quantile level $\tau\in(0,1)$, let $\rho_\tau(\epsilon)=\epsilon\{\tau-\mathbf{1}(\epsilon<0)\}$ denote the check loss. The asymmetric Laplace (AL) density with location zero and scale $u>0$ is $f_{\mathrm{AL},\tau}(\epsilon;0,u)=\tau(1-\tau)u^{-1}\exp\{-\rho_\tau(\epsilon)/u\}$. Following the log-Pareto mixing construction of \citet{Hamura2022}, we use the scale-mixing density
\[
H(u;\gamma)=\frac{\gamma}{(1+u)\{1+\log(1+u)\}^{1+\gamma}},\qquad u>0,\quad\gamma>0.
\]
The log-Pareto mixture of asymmetric Laplace distributions (LPAL) and the AL--LPAL density are defined by
\begin{equation}\label{eq:al_lpal_density}
\begin{aligned}
f_{\mathrm{LPAL},\tau}(\epsilon;\gamma)&=\int_0^\infty f_{\mathrm{AL},\tau}(\epsilon;0,u)H(u;\gamma)\,du,\\
f_\tau(\epsilon;s,\gamma)&=(1-s)f_{\mathrm{AL},\tau}(\epsilon;0,1)+s f_{\mathrm{LPAL},\tau}(\epsilon;\gamma),
\end{aligned}
\end{equation}
where $s\in(0,1)$ is the mixture weight and $\gamma$ controls tail decay. The AL component provides a baseline error density, while the LPAL component permits scale variation over $(0,\infty)$.

\begin{proposition}[Properties of the AL--LPAL distribution]\label{prop:al_lpal_properties}
For any $\tau,s\in(0,1)$ and $\gamma>0$, the density in \eqref{eq:al_lpal_density} satisfies
\[
\int_{-\infty}^0 f_\tau(\epsilon;s,\gamma)\,d\epsilon=\tau,\qquad \lim_{\epsilon\to0}\frac{f_\tau(\epsilon;s,\gamma)}{\log(1/|\epsilon|)}=s\gamma\tau(1-\tau),
\]
and
\[
\begin{aligned}
\lim_{\epsilon\to+\infty}\epsilon(\log\epsilon)^{1+\gamma}f_\tau(\epsilon;s,\gamma)&=s\gamma(1-\tau),\\
\lim_{\epsilon\to-\infty}|\epsilon|(\log|\epsilon|)^{1+\gamma}f_\tau(\epsilon;s,\gamma)&=s\gamma\tau.
\end{aligned}
\]
\end{proposition}

Every AL component has zero as its $\tau$th quantile regardless of scale, so both LPAL and AL--LPAL preserve that quantile. Proposition~\ref{prop:al_lpal_properties} also identifies two features with distinct inferential implications. The integrable logarithmic singularity gives a local log-likelihood shape of order $\log\log(1/|\epsilon|)$ near zero, rather than a locally quadratic form. The two tails decay at order $\{|\epsilon|(\log|\epsilon|)^{1+\gamma}\}^{-1}$, with relative amplitudes determined by $\tau$. Consequently, the working error distribution has no finite positive-order absolute moments, and the leading contribution of an increasingly extreme response becomes independent of fixed regression coefficients and scale while retaining its dependence on $s$.

\begin{figure}[!htp]
\centering
\includegraphics[width=\textwidth]{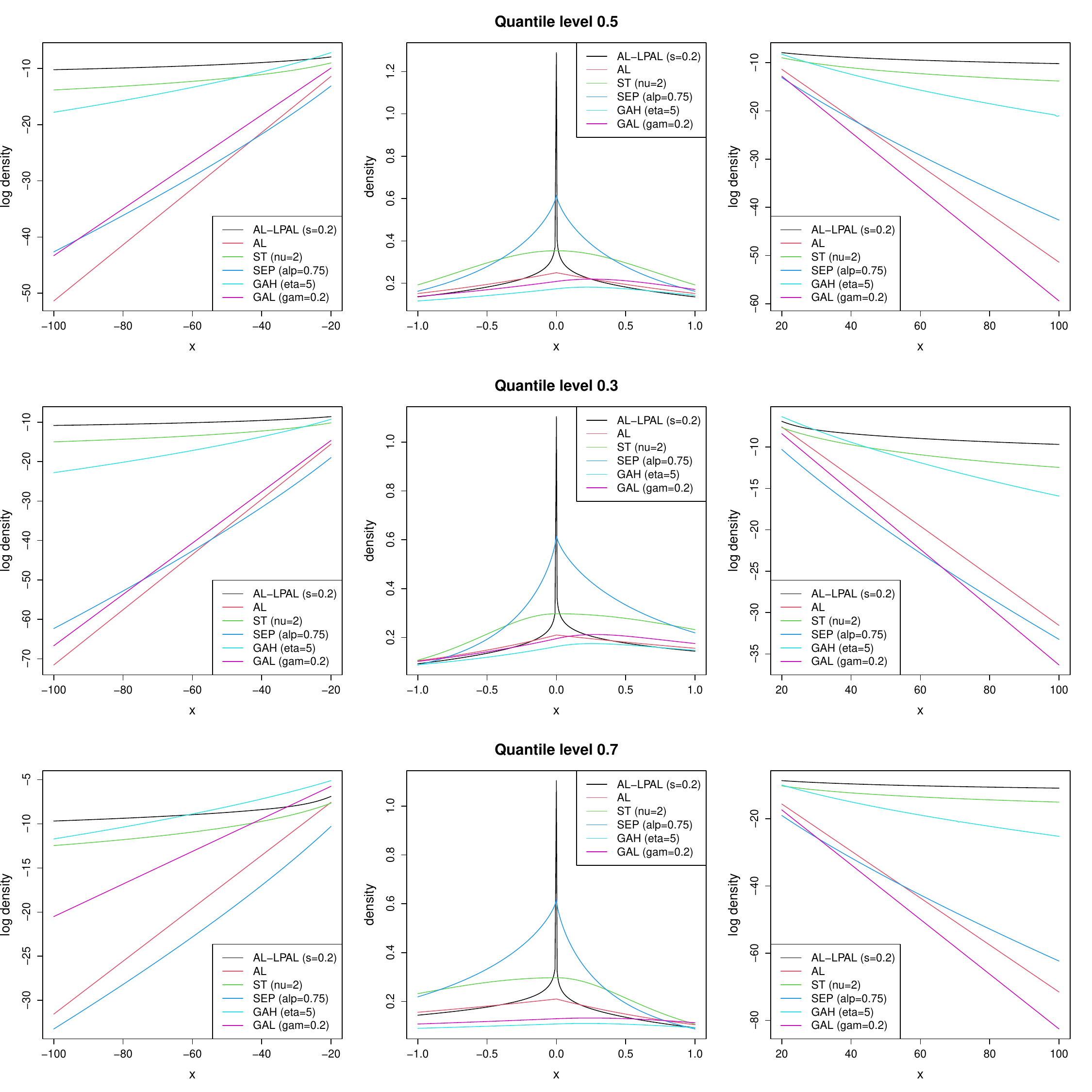}
\caption{Comparison of AL--LPAL with quantile-regression error densities at $\tau=0.5$ (top), $0.3$ (middle), and $0.7$ (bottom). The center panels display densities near zero; the left and right panels display log densities in the lower and upper tails. AL--LPAL uses $s=0.2$; competitors are AL, skew-$t$ (ST, $\nu=2$), skew exponential power (SEP, $\alpha=0.75$), generalized asymmetric Huberised-type (GAH, $\eta=5$, $\gamma^*=0.2$), and generalized asymmetric Laplace (GAL, $\gamma^*=0.2$). All densities have location zero and scale one.}
\label{fig:al_lpal_density_comparison}
\end{figure}

Figure~\ref{fig:al_lpal_density_comparison} illustrates the simultaneous central peak and slower tail decay of AL--LPAL relative to competing error densities, both at the median and at asymmetric quantile levels.

We use AL--LPAL as a working likelihood for linear quantile regression. Let $\bm{x}_i=(1,\bm{z}_i^\top)^\top\in\mathbb{R}^p$ and $\bm\beta=(\alpha,\bm b^\top)^\top$, where $\alpha$ is the intercept and $\bm b\in\mathbb{R}^{p-1}$ contains the slope coefficients. The observations are treated as conditionally independent under the working conditional density
\[
q_{\bm\theta}(y_i\mid\bm{x}_i)=\frac{1}{\sigma}f_\tau\!\left(\frac{y_i-\bm{x}_i^\top\bm\beta}{\sigma};s,\gamma\right),\qquad \bm\theta=(\bm\beta^\top,\sigma,s)^\top,\quad\sigma>0.
\]
Equivalently, the working regression specification is $Y_i=\bm{x}_i^\top\bm\beta+\sigma\epsilon_i$ with $\epsilon_i\overset{\mathrm{iid}}{\sim}f_\tau(\cdot;s,\gamma)$. Its model-implied conditional $\tau$th quantile is $\bm{x}_i^\top\bm\beta$, although the working density need not coincide with the true conditional response density.

We assign mutually independent proper priors
\[
\bm\beta\sim N_p(\bm\mu_\beta,\bm\Sigma_\beta),\qquad \sigma\sim\mathrm{IG}(a_\sigma,b_\sigma),\qquad s\sim\mathrm{Beta}(a_s,b_s),
\]
where $\bm\Sigma_\beta\succ0$ and all scalar hyperparameters are fixed and strictly positive. The mixture weight $s$ is estimated, while $\gamma$ is fixed at $1$ unless stated otherwise.

\section{Posterior Targeting and Asymptotic Theory}
\label{sec:theory}

We establish finite-sample posterior propriety, examine population targeting under misspecification, and derive posterior asymptotics under jointly increasing sample size and contamination magnitude.

\subsection{Posterior propriety and moments}
\label{subsec:posterior_propriety}

Let $D_n=\{(Y_i,\bm x_i):i=1,\ldots,n\}$ and denote the working posterior by $\Pi_n(\cdot\mid D_n)$. Posterior propriety and moment existence have been studied for AL-based quantile regression \citep{Yu2001} and log-regularly varying regression models \citep{Hamura2022}, respectively. The logarithmically unbounded AL--LPAL density requires a separate integrability argument.

\begin{proposition}[Posterior propriety and moments]\label{prop:posterior_propriety}
Under the model and prior specification of Section~\ref{sec:al_lpal}, the posterior is proper for every finite sample with finite observed responses. Moreover, for every $r>0$, $E_{\Pi_n}(|\beta_j|^r\mid D_n)<\infty$ for $j=1,\ldots,p$.
\end{proposition}

Despite the absence of finite positive-order moments for the AL--LPAL working errors, Proposition~\ref{prop:posterior_propriety} guarantees all positive-order posterior moments of the regression coefficients without requiring a full-column-rank design.

\subsection{Inferential distortions under misspecification}
\label{subsec:inferential_distortions}

Under likelihood misspecification, the population working criterion determines the pseudo-true parameter \citep{Kleijn2006}. Consider a clean response $Y=\bm x^\top\bm\beta_0+\epsilon$, where the error density $g_0$ satisfies $G_0(0)=\tau$. For fixed $\sigma>0$ and $s\in(0,1)$, define the expected working log likelihood for a location displacement $a$ by $F_0(a,\sigma,s)=E_{g_0}[\log\{\sigma^{-1}f_\tau((\epsilon-a)/\sigma;s,\gamma)\}]$. The conventional AL likelihood targets the true $\tau$th quantile through expected check-loss minimization \citep{Sriram2013}, but the AL--LPAL likelihood need not retain this property despite preserving its model-implied quantile.

\begin{proposition}[Possible intercept distortion]\label{prop:intercept_distortion}
For every fixed $\tau\in(0,1)$, $\gamma>0$, $\sigma>0$, and $s\in(0,1)$, there exists a bounded smooth error density $g_0$ with $G_0(0)=\tau$, $g_0(0)>0$, and finite moments of every positive order, for which $F_0(a,\sigma,s)>F_0(0,\sigma,s)$ at some $a\ne0$.
\end{proposition}

Proposition~\ref{prop:intercept_distortion} demonstrates that quantile preservation does not guarantee correct population targeting, even under smooth clean errors with moments of all orders. Identifying the joint pseudo-true regression parameter and its posterior behavior requires the population conditions introduced next.

\subsection{Posterior asymptotics under increasing contamination}
\label{subsec:corrected_posterior}

Let $K_n$ and $L_n$ partition $\{1,\ldots,n\}$ into clean and contaminated indices, with $k=|K_n|$, $m=|L_n|$, and $k+m=n$. We study the joint growth of $n$ and the contamination magnitude $\omega_n$ under
\[
Y_i=\begin{cases}
\bm x_i^\top\bm\beta_0+\epsilon_i,&i\in K_n,\\
a_i+c_i\omega_n,&i\in L_n,
\end{cases}
\qquad \bm\beta_0=(\alpha_0,\bm b_0^\top)^\top,
\]
where the clean errors have density $g_0$ and $\{a_i,c_i\}_{i\ge1}$ are deterministic. Let $P_0$ denote probability under this experiment.

Write $\bm\xi=(\log\sigma,\operatorname{logit}s)^\top$ and $\rho=\eta/(1-\eta)$, where $\eta=\lim_{n\to\infty}m/n$. Define
\[
F_\rho(a,\bm\xi)=F_0(a,\sigma,s)+\rho\log s.
\]
The term $\rho\log s$ captures the leading contribution of extreme observations when the mixture weight is unknown. When the maximizer exists, denote it by $(a^\dagger,\bm\xi^\dagger)$ and put $\bm\beta^\dagger=(\alpha_0+a^\dagger,\bm b_0^\top)^\top$.

Fix $q\in(1,2)$ and define the theoretical screening set $\widehat K_n=\{i:|Y_i|\le n^{1/q}\}$. Let $T_n(\bm b)=\widehat Q_\tau\{Y_i-\bm z_i^\top\bm b:i\in\widehat K_n\}$, where $\widehat Q_\tau$ denotes the lower empirical quantile and $T_n(\bm b)=0$ if $\widehat K_n$ is empty. For an original posterior draw $(\alpha,\bm b)$, define $\Delta_n=\alpha-T_n(\bm b)$ and $\bm\beta^{\mathrm{cor}}=(T_n(\bm b),\bm b^\top)^\top$. The corrected posterior is the pushforward of the slope marginal posterior,
\begin{equation}\label{eq:corrected_posterior}
\Pi_n^{\mathrm{cor}}(B\mid D_n)=\int\mathbf1\!\left\{(T_n(\bm b),\bm b^\top)^\top\in B\right\}\Pi_n^{\bm b}(d\bm b\mid D_n),\qquad B\in\mathcal B(\mathbb R^p).
\end{equation}

We retain the AL--LPAL working likelihood and priors from Section~\ref{sec:al_lpal} and impose the following assumptions. Write $\bm x(\bm z)=(1,\bm z^\top)^\top$.

\begin{itemize}
\item[(A1)] \emph{(Clean error distribution)} The clean errors $\epsilon_i$ are i.i.d. with a bounded density $g_0$, and zero is the unique $\tau$th quantile of $G_0$. For the fixed $q\in(1,2)$, $E_{g_0}|\epsilon|^q<\infty$. The density $g_0$ is continuously differentiable near zero, with bounded derivative and $g_0(0)>0$.

\item[(A2)] \emph{(Deterministic fixed-dimensional design)} The covariates form a deterministic triangular array with fixed dimension $p$, $\sup_{n,i}\|\bm x_i\|<\infty$, and
\[
k^{-1}\sum_{i\in K_n}\delta_{\bm z_i}\Rightarrow\nu,\qquad \lambda_{\min}(\bm Q_\nu)>0,
\]
where $\bm Q_\nu=\int\bm x(\bm z)\bm x(\bm z)^\top\nu(d\bm z)$.

\item[(A3)] \emph{(Increasing contamination)} As $n\to\infty$, $m/n\to\eta\in[0,1)$. The coefficients $\{a_i,c_i\}_{i\ge1}$ do not depend on $n$ and satisfy, for fixed $C_a,c_-,c_+>0$,
\[
|a_i|\le C_a,\qquad 0<c_-\le|c_i|\le c_+<\infty,\qquad i\ge1.
\]
Moreover, $\omega_n\ge n^{B_*}$ eventually, where $B_*>2$ is a sufficiently large fixed constant depending only on the fixed model, prior, and contamination-fraction parameters.

\item[(A4)] \emph{(Population identification)} The criterion $F_\rho(a,\bm\xi)$ has a unique maximizer $(a^\dagger,\bm\xi^\dagger)\in\mathbb R^3$ satisfying, for every $\delta>0$,
\[
\sup_{\|(a,\bm\xi)-(a^\dagger,\bm\xi^\dagger)\|\ge\delta}F_\rho(a,\bm\xi)<F_\rho(a^\dagger,\bm\xi^\dagger).
\]
The density $g_0$ is continuous and strictly positive near $a^\dagger$.

\item[(A5)] \emph{(Local population curvature)} The criterion $F_\rho(a,\bm\xi)$ is twice continuously differentiable near $(a^\dagger,\bm\xi^\dagger)$, with uniformly continuous Hessian there, and $-\nabla^2F_\rho(a^\dagger,\bm\xi^\dagger)$ is positive definite.
\end{itemize}

Assumptions (A1)--(A3) govern the clean errors, design, and increasing-contamination regime, under which $P_0(\widehat K_n=K_n)\to1$ and $k\asymp n$. Assumptions (A4)--(A5) impose global population identification and local curvature without restricting the clean errors to a particular family. The following proposition verifies these conditions for a connected class of generalized Gaussian distributions and their smooth perturbations, including examples with nonzero pseudo-true intercept displacement.

\begin{proposition}[Population identification and curvature for a connected class of clean distributions]
\label{prop:nonvacuity}
Fix $\tau\in(0,1)$ and $\gamma>0$. Let
\[
h_\kappa(z)=\frac{\exp(-|z|^\kappa)}{2\Gamma(1+1/\kappa)},\qquad \kappa\ge2,
\]
and denote its $\tau$th quantile by $q_{\kappa,\tau}$. Define the quantile-adapted and shifted generalized Gaussian densities by
\[
g^Q_{\tau,\kappa}(u)=\frac{1}{2\Gamma(1+1/\kappa)}
\begin{cases}
\exp\left\{-\left(\dfrac{|u|}{2\tau}\right)^\kappa\right\},&u<0,\\[6pt]
\exp\left\{-\left(\dfrac{u}{2(1-\tau)}\right)^\kappa\right\},&u\ge0,
\end{cases}
\]
and $g^S_{\tau,\kappa}(u)=h_\kappa(u+q_{\kappa,\tau})$. For $\vartheta\in[0,1]$ and $d>0$, consider the clean error density
\[
g_{\tau,\kappa,\vartheta,d}(u)=\frac1d\left\{(1-\vartheta)g^Q_{\tau,\kappa}(u/d)+\vartheta g^S_{\tau,\kappa}(u/d)\right\}.
\]
There exist constants $K(\tau,\gamma)\ge2$, $\varepsilon(\tau,\gamma)>0$, and $\eta_*(\tau,\gamma)>0$ such that assumptions (A1), (A4), and (A5) hold whenever $d>0$, $\eta\in(0,\eta_*(\tau,\gamma))$, and
\[
(\kappa,\vartheta)\in\mathcal R_{\tau,\gamma}:=\bigl([2,K(\tau,\gamma)]\times[0,\varepsilon(\tau,\gamma)]\bigr)\cup\bigl([K(\tau,\gamma),\infty)\times[0,1]\bigr),
\]
where $\varepsilon(\tau,\gamma)$ can be chosen in $(0,1]$. In particular, $F_\rho$ has a unique, globally separated maximizer $(a^\dagger,\bm\xi^\dagger)$ with $s^\dagger\in(0,1)$ and $-\nabla^2F_\rho(a^\dagger,\bm\xi^\dagger)\succ0$. The threshold $\eta_*(\tau,\gamma)$ is uniform over $(\kappa,\vartheta)\in\mathcal R_{\tau,\gamma}$ and $d>0$. Moreover, for any compact interval $I_\eta\subset(0,\eta_*(\tau,\gamma))$, these conditions remain valid uniformly over $I_\eta$ in a sufficiently small smooth density neighborhood of the standardized quantile-adapted Gaussian reference $g^Q_{\tau,2}$. This neighborhood contains an infinite-dimensional family of smooth, strictly log-concave densities, including examples with $a^\dagger\ne0$.
\end{proposition}

\begin{remark}
Proposition~\ref{prop:nonvacuity} covers quantile-adapted generalized Gaussian errors for every $\kappa\ge2$, shifted generalized Gaussian errors for sufficiently large $\kappa$, and their admissible mixtures. At $\tau=1/2$, both families coincide and include the ordinary Gaussian distribution. These are sufficient rather than necessary conditions: the subsequent posterior theory also applies to distributions outside this class, including potentially heavier-tailed errors, whenever (A1)--(A5) hold.
\end{remark}

Define the contraction rate $r_n=n^{-1/3}(\log n)^{-2/3}$.

\begin{theorem}\label{thm:corrected_contraction}
Suppose that (A1)--(A5) hold. Then, for any deterministic sequence $M_n\to\infty$,
\[
\Pi_n^{\mathrm{cor}}\left\{\|\bm\beta^{\mathrm{cor}}-\bm\beta_0\|>M_n r_n\mid D_n\right\}\xrightarrow{P_0}0.
\]
\end{theorem}

Theorem~\ref{thm:corrected_contraction} establishes correct posterior targeting at the logarithmically modified cube-root rate despite possible displacement of the original working-posterior intercept.

\begin{remark}[Origin of the logarithmically modified cube-root rate]
\label{rem:nonregular_rate}
The present contraction rate is determined by the local stochastic behavior of the working likelihood rather than by a departure from quadratic population curvature. Let $r>0$ denote the magnitude of a local location perturbation. In regular misspecified models, square-integrable score expansions yield local empirical fluctuations of order $r/\sqrt{k}$, which balance quadratic population drift at the $k^{-1/2}$ scale \citep{Kleijn2012}. Conventional cube-root asymptotics instead arise from local increment variances of order $r$ \citep{Kim1990,Jun2015}. For the AL--LPAL working density, the location score diverges at order $\{|\epsilon|\log(1/|\epsilon|)\}^{-1}$ as $\epsilon\to0$, precluding the usual square-integrable score expansion. The local log-likelihood increment variance is instead of order $r/\log^2(1/r)$. Balancing the resulting stochastic fluctuations against quadratic population drift gives
\[
r^2\asymp\left\{\frac{r}{k\log^2(1/r)}\right\}^{1/2},
\qquad
r\asymp k^{-1/3}(\log k)^{-2/3}.
\]
Since $k\asymp n$, this yields the rate in Theorem~\ref{thm:corrected_contraction}. The logarithmic factor is thus intrinsic to the local likelihood geometry rather than a technical artifact of the contraction analysis.
\end{remark}

To characterize the random center and its sampling limit, let $\nu$ be the limiting clean-design measure from (A2) and define the regression population criterion
\[
\mathcal M_\rho(\bm\beta,\bm\xi)=\int F_\rho\{\alpha-\alpha_0+\bm z^\top(\bm b-\bm b_0),\bm\xi\}\nu(d\bm z).
\]
Let $\bm H=-\nabla^2\mathcal M_\rho(\bm\beta^\dagger,\bm\xi^\dagger)$, partitioned according to $(\bm\beta,\bm\xi)$, and define the efficient curvature by
\[
\bm H_{\mathrm{eff}}=\bm H_{\beta\beta}-\bm H_{\beta\xi}\bm H_{\xi\xi}^{-1}\bm H_{\xi\beta}.
\]
For sufficiently large $k$, let $r_k\in(0,e^{-2/3})$ be the unique solution of $kr_k^3\{\log(1/r_k)\}^2=1$; since $k\asymp n$, we have $r_k\asymp r_n$. The logarithmic modification of the cube-root rate follows by balancing quadratic population drift against averaged log-likelihood increments of order $\{r/[k\log^2(1/r)]\}^{1/2}$, yielding $kr^3\log^2(1/r)\asymp1$. Write $\bm\mu_z=\int\bm z\,\nu(d\bm z)$ and set
\[
\bm A=\begin{pmatrix}
0&-\bm\mu_z^\top\\
\bm0_{(p-1)\times1}&\bm I_{p-1}
\end{pmatrix}.
\]
Let $\mathbb G$ be a centered Gaussian process on $\mathbb R^p$ with $\mathbb G(\bm0)=0$ and covariance
\begin{equation}\label{eq:nonregular_gp_covariance}
\begin{aligned}
\operatorname{Cov}\{\mathbb G(\bm t),\mathbb G(\bm u)\}&=\frac{g_0(a^\dagger)}{2}\left[\pi^2+\left\{\log\frac{\tau}{1-\tau}\right\}^2\right]\\
&\quad\times\int\left\{|\bm x(\bm z)^\top\bm t|+|\bm x(\bm z)^\top\bm u|-|\bm x(\bm z)^\top(\bm t-\bm u)|\right\}\nu(d\bm z).
\end{aligned}
\end{equation}
Define $\bm T^\star=\operatorname*{argmax}_{\bm t\in\mathbb R^p}\{\mathbb G(\bm t)-\bm t^\top\bm H_{\mathrm{eff}}\bm t/2\}$.

\begin{theorem}\label{thm:corrected_random_center}
Suppose that (A1)--(A5) hold. There exist measurable random vectors $\bm T_{n,\omega}\in\mathbb R^p$ such that $\bm T_{n,\omega}\Rightarrow\bm T^\star$ and, for every $\varepsilon>0$,
\[
\Pi_n^{\mathrm{cor}}\left\{\left\|r_k^{-1}(\bm\beta^{\mathrm{cor}}-\bm\beta_0)-\bm A\bm T_{n,\omega}\right\|>\varepsilon\mid D_n\right\}\xrightarrow{P_0}0.
\]
\end{theorem}

Theorem~\ref{thm:corrected_random_center} separates the corrected posterior's $o_{P_0}(r_k)$ conditional spread from its random center's $r_k$-order sampling fluctuations. The matrix $\bm A$ propagates slope fluctuations to the corrected intercept; when $\bm\mu_z=\bm0$, its leading $r_k$-scale variation vanishes. This separation yields the following coverage failure.

\begin{proposition}[Undercoverage of unadjusted corrected-posterior intervals]\label{prop:coverage_distortion}
Suppose that (A1)--(A5) hold, and let $\bm v\in\mathbb R^p$ be fixed with $P(\bm v^\top\bm A\bm T^\star=0)=0$. For any fixed $\zeta\in(0,1)$, let $I_{n,1-\zeta}^{\mathrm{cor}}(\bm v)$ be an equal-tailed $(1-\zeta)$ credible interval for $\bm v^\top\bm\beta^{\mathrm{cor}}$ under $\Pi_n^{\mathrm{cor}}(\cdot\mid D_n)$. Then
\[
P_0\{\bm v^\top\bm\beta_0\in I_{n,1-\zeta}^{\mathrm{cor}}(\bm v)\}\longrightarrow0.
\]
\end{proposition}

Proposition~\ref{prop:coverage_distortion} shows that correcting posterior targeting does not restore valid frequentist uncertainty: equal-tailed credible intervals have asymptotically zero coverage for contrasts with nondegenerate leading fluctuations.

\section{Posterior Computation and Interval Calibration}
\label{sec:computation}

\subsection{Posterior sampling and quantile correction}
\label{subsec:posterior_sampling_correction}

We sample from the original AL--LPAL working posterior using a partially collapsed Gibbs sampler and subsequently apply the residual-quantile correction in \eqref{eq:corrected_posterior}. The sampler combines the normal--exponential representation of the AL distribution \citep{Kozumi2011} with the log-Pareto scale mixture \citep{Hamura2022}. The partially collapsed updates integrate out a latent variable before updating its associated block and then refresh that variable in the prescribed order \citep{vanDykPark2008}. All MCMC updates use the full, unscreened likelihood; quantile correction is applied only to retained posterior draws.

For each observation, introduce $z_i\mid s\sim\mathrm{Bernoulli}(s)$ and $u_i\sim H(\cdot;\gamma)$ independently a priori, and set $\lambda_i=u_i^{z_i}$. With $\kappa_1=(1-2\tau)/\{\tau(1-\tau)\}$ and $\kappa_2^2=2/\{\tau(1-\tau)\}$, the likelihood admits the augmentation
\[
\theta_i\mid\sigma,\lambda_i\sim\mathrm{Exp}(1/\sigma\lambda_i),\qquad
Y_i\mid\theta_i,\lambda_i,\bm\beta,\sigma\sim N\!\left(\bm x_i^\top\bm\beta+\kappa_1\theta_i,\;\kappa_2^2\sigma\lambda_i\theta_i\right).
\]
Marginalizing $\theta_i$ gives an AL density with scale $\sigma\lambda_i$. Write $r_i=Y_i-\bm x_i^\top\bm\beta$ and $d_i=\rho_\tau(r_i)/\sigma$. Conditional on $u_i$, the collapsed allocation probability is
\[
P(z_i=1\mid u_i,\bm\beta,\sigma,s,Y_i)=\frac{s u_i^{-1}\exp(-d_i/u_i)}{(1-s)\exp(-d_i)+s u_i^{-1}\exp(-d_i/u_i)}.
\]
After updating $z_i$, we draw $u_i$ from $H(\cdot;\gamma)$ if $z_i=0$; otherwise, its conditional density is proportional to $u_i^{-1}H(u_i;\gamma)\exp(-d_i/u_i)$. We update $h_i=\log u_i$ by slice sampling \citep{Neal2003}, using the log density $\log H(e^{h_i};\gamma)-d_i e^{-h_i}$ up to a constant. Given the updated $(z_i,u_i)$, we draw $\theta_i\mid-\sim\mathrm{GIG}(1/2,\chi_i,\psi_i)$, where $\chi_i=r_i^2/(\kappa_2^2\sigma\lambda_i)$ and $\psi_i=(2+\kappa_1^2/\kappa_2^2)/(\sigma\lambda_i)$. Here $\mathrm{GIG}(\lambda,\chi,\psi)$ has density proportional to $t^{\lambda-1}\exp\{-(\chi/t+\psi t)/2\}$. The remaining updates are inverse-gamma for $\sigma$, Gaussian for $\bm\beta$, and beta for $s$, with $s\mid-\sim\mathrm{Beta}(a_s+\sum_i z_i,b_s+n-\sum_i z_i)$. The update order is $(z_i,u_i)_{i=1}^n$, $(\theta_i)_{i=1}^n$, $\sigma$, $\bm\beta$, and $s$.

For each retained draw, we implement the transformation in \eqref{eq:corrected_posterior} using one common screened set $\widehat K_n$; sampling always precedes screening, and no likelihood refitting is required. The retained slope draw is unchanged, while its corresponding intercept is replaced by the empirical residual quantile computed on that set. In addition to the theoretical rule $\widehat K_n=\{i:|Y_i|\le n^{1/q}\}$, numerical implementation uses a residual-based rule. Let $\widehat{\bm b}$ be the coordinatewise median of the original slope draws and set $R_i(\widehat{\bm b})=Y_i-\bm z_i^\top\widehat{\bm b}$. We retain observations satisfying
\[
\left|R_i(\widehat{\bm b})-\operatorname{med}_{j}R_j(\widehat{\bm b})\right|\le (\log n)\operatorname{MAD}\{R_1(\widehat{\bm b}),\ldots,R_n(\widehat{\bm b})\},
\]
where MAD denotes the conventionally scaled median absolute deviation \citep{Rousseeuw1993}. The data-adaptive rule is distinct from the theoretical screening rule and is not covered by its exact-recovery statement.

\begin{algorithm}[!ht]
\small
\caption{Posterior sampling and residual-quantile correction}
\label{alg:posterior_sampling_correction}
\KwIn{$\bm y$, $\bm X$, quantile level $\tau$, fixed $\gamma$, prior hyperparameters, iterations $T$, burn-in $B$, and screening rule.}
\KwOut{Original posterior draws, correction draws, and corrected coefficient draws.}
Initialize $(\bm\beta,\sigma,s)$ and $(z_i,u_i,\theta_i)_{i=1}^n$.

\For{$t=1,\ldots,T$}{
  \For{$i=1,\ldots,n$}{
    Update $z_i$ with $\theta_i$ marginalized out, then update $u_i$ using its conditional distribution.
  }
  Refresh $(\theta_i)_{i=1}^n$ from their GIG conditionals.

  Update $\sigma$, $\bm\beta$, and $s$, in this order.

  \If{$t>B$}{Store $(\alpha^{(t)},\bm b^{(t)},\sigma^{(t)},s^{(t)})$.}
}
Construct one screened set $\widehat K_n$ using the specified rule.

\For{$t=B+1,\ldots,T$}{
  Compute $\Delta_n^{(t)}=\alpha^{(t)}-T_n(\bm b^{(t)})$.

  Set $\bm\beta^{\mathrm{cor},(t)}=(T_n(\bm b^{(t)}),(\bm b^{(t)})^\top)^\top$.
}
\end{algorithm}

\subsection{Gaussian-process-based interval calibration}
\label{subsec:nonregular_calibration}

Theorem~\ref{thm:corrected_random_center} shows that the quantile-corrected posterior fails to reproduce the leading sampling fluctuations of its random center. Simulation-based inference using Gaussian-process approximations has been developed for nonregular Laplace-type estimators \citep{Jun2015}, and alternative bootstrap methods have been proposed for cube-root problems \citep{Cattaneo2020}. We construct a finite-scale Gaussian-process reference tailored to the logarithmically modified cube-root asymptotics of the AL--LPAL likelihood and combine it with Bayesian-bootstrap residual quantiles for intercept inference. The resulting confidence-type intervals are based on simulated sampling references rather than posterior probabilities.

Let $(\widehat{\bm\beta},\widehat\sigma,\widehat s)$ denote the coordinatewise medians of the original posterior draws, put $\widehat{\bm\xi}=(\log\widehat\sigma,\operatorname{logit}\widehat s)^\top$, $k=|\widehat K_n|$, $\widehat\rho=(n-k)/k$, and set $r=r_k$. We evaluate the local likelihood geometry at the original posterior pilot rather than the corrected intercept, since the covariance and efficient curvature in Theorem~\ref{thm:corrected_random_center} are determined at the pseudo-true center of the working posterior. For $i\in\widehat K_n$, let $R_i=Y_i-\bm x_i^\top\widehat{\bm\beta}$, and denote by $\widehat G_h$ the empirical distribution of these residuals convolved with $N(0,h^2)$, where $h$ is the smoothing bandwidth. Define $\ell_{\widehat{\bm\xi}}(v)=\log\{\widehat\sigma^{-1}f_\tau(v/\widehat\sigma;\widehat s,\gamma)\}$ and $D_r(v,t)=\ell_{\widehat{\bm\xi}}(v-rt)-\ell_{\widehat{\bm\xi}}(v)$. The finite-scale covariance is
\[
\widehat C_{r,h}(t,u)=\operatorname{Cov}_{R\sim\widehat G_h}\{D_r(R,t),D_r(R,u)\}.
\]
We evaluate this covariance on a grid of projected parameter shifts and generate centered Gaussian paths from the resulting covariance matrix. Rather than directly estimating the limiting covariance kernel in \eqref{eq:nonregular_gp_covariance}, our construction retains the dependence structure of the AL--LPAL log-likelihood increments at the nonregular scale $r_k$.

To approximate the quadratic drift, let $\widehat h_{aa}$, $\widehat{\bm h}_{a\xi}$, and $\widehat{\bm h}_{\xi\xi}$ be the blocks of the negative Hessian of the smoothed clean scalar criterion, evaluated at the original posterior pilot. With $\widehat{\bm Q}=k^{-1}\sum_{i\in\widehat K_n}\bm x_i\bm x_i^\top$, $\overline{\bm x}=k^{-1}\sum_{i\in\widehat K_n}\bm x_i$, and $\widehat{\bm J}_{\xi\xi}=\widehat{\bm h}_{\xi\xi}+\widehat\rho\operatorname{diag}\{0,\widehat s(1-\widehat s)\}$, define
\[
\widehat{\bm H}_{\mathrm{eff}}=\widehat h_{aa}\widehat{\bm Q}-\left(\widehat{\bm h}_{a\xi}\widehat{\bm J}_{\xi\xi}^{-1}\widehat{\bm h}_{\xi a}\right)\overline{\bm x}\,\overline{\bm x}^\top.
\]
The Schur-complement form follows the profiling principle in regular likelihood theory \citep{Murphy2000}, but here determines the quadratic drift of the nonregular Gaussian-process limit. While smoothing has been used for inference with nonsmooth quantile-regression criteria \citep{Horowitz1998}, Gaussian convolution here permits curvature estimation without differentiating directly across the logarithmic singularity. Numerical implementation allows data-adaptive bandwidth selection, covariance regularization, and curvature stabilization.

For each simulation replicate $j=1,\ldots,J$, generate independent centered Gaussian paths $\mathbb V_i^{(j)}$, $i\in\widehat K_n$, with covariance $\widehat C_{r,h}$, and construct
\[
\widehat{\mathbb G}_r^{(j)}(\bm t)=\frac{1}{kr^2}\sum_{i\in\widehat K_n}\mathbb V_i^{(j)}(\bm x_i^\top\bm t).
\]
We approximate the Gaussian-process argmax distribution numerically \citep{Groeneboom2001} by computing
\[
\widehat{\bm T}^{*(j)}=\operatorname*{argmax}_{\bm t\in\mathcal T}\left\{\widehat{\mathbb G}_r^{(j)}(\bm t)-\frac12\bm t^\top\widehat{\bm H}_{\mathrm{eff}}\bm t\right\},
\]
where $\mathcal T$ is a bounded numerical search region. The slope reference draws are $\bm b_{\mathrm{cal}}^{(j)}=\widehat{\bm b}-r[\widehat{\bm T}^{*(j)}]_{\mathrm{slope}}$, where $[\cdot]_{\mathrm{slope}}$ selects the last $p-1$ coordinates. The negative sign converts simulated estimation errors into parameter-side reference draws, whose empirical quantiles yield marginal slope intervals.

The following proposition establishes consistency of the plug-in Gaussian-process reference.

\begin{proposition}[Consistency of the GP calibration reference]
\label{prop:gp_calibration}
Suppose (A1)--(A5) hold. Construct the plug-in Gaussian process and efficient curvature using the theoretical screening rule and exact original-posterior pilot medians, with bandwidth $h_k=k^{-1/8}$ and without numerical regularization. Let $\widehat{\bm T}_k^*$ maximize a single simulated penalized Gaussian-process path over the grid
\[
\mathcal T_k=\left\{(\log k)^{-1}\bm j:\bm j\in\mathbb Z^p,\ \|\bm j\|\le(\log k)^2\right\}.
\]
Then
\[
d_{\mathrm{BL}}\left\{\mathcal L(\widehat{\bm T}_k^*\mid D_n),\mathcal L(\bm T^\star)\right\}\xrightarrow{P_0}0,
\]
where $d_{\mathrm{BL}}$ denotes the bounded-Lipschitz distance and $\bm T^\star$ is the Gaussian-process argmax defined in Theorem~\ref{thm:corrected_random_center}.
\end{proposition}

For intercept calibration, we retain the nonlinear residual-quantile transformation to account for both slope-induced fluctuations and empirical-quantile variability, as motivated by classical sample-quantile linearization \citep{Bahadur1966}. The latter is particularly relevant under centered designs, where the leading $r_k$-scale slope-induced intercept fluctuation vanishes. For each replicate $l=1,\ldots,L$, select $J_l$ uniformly from $\{1,\ldots,J\}$ and independently generate $E_i^{(l)}\sim\mathrm{Exp}(1)$ for $i\in\widehat K_n$. Following the Bayesian bootstrap \citep{Rubin1981,Lo1987}, set $W_i^{(l)}=E_i^{(l)}/\sum_{a\in\widehat K_n}E_a^{(l)}$ and compute
\[
\alpha_{\mathrm{cal}}^{(l)}=\widehat Q_{\tau,W^{(l)}}\left\{Y_i-\bm z_i^\top\bm b_{\mathrm{cal}}^{(J_l)}:i\in\widehat K_n\right\},
\]
where $\widehat Q_{\tau,W^{(l)}}$ denotes the lower weighted empirical quantile. We report $\widehat{\bm b}$ and $T_n(\widehat{\bm b})$ as point estimates and obtain marginal intervals from the simulated slope and intercept references. Algorithm~\ref{alg:gp_bb_calibration} summarizes the procedure.

\begin{algorithm}[!ht]
\small
\caption{Gaussian-process-based interval calibration with Bayesian-bootstrap residual quantiles}
\label{alg:gp_bb_calibration}
\KwIn{$\bm y$, $\bm X$, retained original posterior draws, $\widehat K_n$, $\tau$, fixed $\gamma$, GP repetitions $J$, bootstrap repetitions $L$, and numerical controls.}
\KwOut{Corrected point estimates and calibrated marginal coefficient intervals.}
Compute pilot medians $(\widehat{\bm\beta},\widehat\sigma,\widehat s)$ and set $k=|\widehat K_n|$ and $r=r_k$.

Estimate $\widehat C_{r,h}$ and $\widehat{\bm H}_{\mathrm{eff}}$ from the smoothed pilot residual distribution.

\For{$j=1,\ldots,J$}{
  Simulate $\widehat{\mathbb G}_r^{(j)}$ and compute its penalized argmax $\widehat{\bm T}^{*(j)}$.

  Set $\bm b_{\mathrm{cal}}^{(j)}=\widehat{\bm b}-r[\widehat{\bm T}^{*(j)}]_{\mathrm{slope}}$.
}
\For{$l=1,\ldots,L$}{
  Draw $J_l$ uniformly from $\{1,\ldots,J\}$ and normalized exponential weights $W_i^{(l)}$, $i\in\widehat K_n$.

  Compute $\alpha_{\mathrm{cal}}^{(l)}=\widehat Q_{\tau,W^{(l)}}\{Y_i-\bm z_i^\top\bm b_{\mathrm{cal}}^{(J_l)}:i\in\widehat K_n\}$.
}
Report $\widehat{\bm b}$, $T_n(\widehat{\bm b})$, and marginal intervals from the simulated slope and intercept references.
\end{algorithm}

Proposition~\ref{prop:gp_calibration} establishes consistency of the Gaussian-process reference under the specified theoretical regime. The implemented procedure additionally incorporates data-adaptive screening and smoothing, numerical regularization, and Bayesian-bootstrap intercept calibration. The frequentist coverage and interval performance of this complete procedure are assessed empirically in Section~\ref{sec:Empirical}, rather than established by Proposition~\ref{prop:gp_calibration}.

\section{Empirical Studies}\label{sec:Empirical}

\subsection{Simulation studies}\label{sec:simulation}

We examine three aspects of the proposed procedure through Monte Carlo experiments: stability under extreme contamination, restoration of the clean quantile target, and interval calibration. The clean responses are generated from $Y_i=\bm z_i^\top\bm b^*+\varepsilon_i$, where $\bm b^*=(1,-2,1)^\top$, $\bm z_i\overset{\mathrm{iid}}{\sim}N_3(\bm0,\Sigma)$ with $\Sigma_{jk}=0.5^{|j-k|}$, and $\varepsilon_i\sim\mathcal L(0.5T_3)$ for a Student-$t$ variable $T_3$ with three degrees of freedom. The inferential target is the clean conditional quantile $q_\tau^*(\bm z)=\alpha_\tau^*+\bm z^\top\bm b^*$, where $\alpha_\tau^*=0.5F_{t_3}^{-1}(\tau)$. Thus, the same clean error distribution is used across quantile levels; its $\tau$th quantile determines the target intercept.

For one-sided contamination, we replace exactly $n\varphi$ errors by independent $N(\omega,1)$ draws. For two-sided contamination, we replace $0.6n\varphi$ errors by $N(\omega,1)$ draws and another $0.6n\varphi$ by $N(-1.5\omega,1)$ draws, yielding a total contamination fraction of $1.2\varphi$. The contaminated indices are selected randomly, with fixed contamination counts in each setting. We consider $\tau\in\{0.1,0.5,0.9\}$, $\varphi\in\{0.05,0.10,0.15,0.20\}$, $\omega\in\{5,10,20,50,100,200\}$, and an uncontaminated baseline. Each configuration is replicated 500 times. These finite-sample experiments intentionally go beyond the sufficient conditions of the asymptotic theory, using Gaussian covariates, random contamination values, and moderate as well as extreme perturbations instead of the bounded fixed design and deterministic diverging contamination of (A2)--(A3).

We compare AL--LPAL with asymmetric Laplace (AL; \citealp{Yu2001}), generalized asymmetric Laplace (GAL; \citealp{Yan2026}), generalized asymmetric Huberised-type (GAH; \citealp{Hu2024}), skew-$t$ (\citealp{GalarzaMorales2017}), and contaminated generalized asymmetric Laplace (cGAL; \citealp{Burger2025}) methods. The main contamination-grid experiment uses $n=100$, with 10,000 MCMC iterations, a burn-in of 5,000, and three chains. For AL--LPAL, the numerical priors are $\bm\beta\sim N_p(\bm0,1000\bm I_p)$, $\sigma\sim\mathrm{IG}(0.1,0.1)$, and $s\sim\mathrm{Beta}(1,1)$, with $\gamma=1$. Interval calibration uses 199 Gaussian-process simulations and 500 Bayesian-bootstrap replications. 
We assess estimation relative to the clean target and evaluate 95\% intervals by empirical coverage probability (CP), average interval length (AvL), and interval score (IS). For a scalar coefficient with clean target $\theta^*$, let $[L_m,U_m]$ denote its 95\% interval in the $m$th Monte Carlo replication, $m=1,\ldots,M$. We define
\[
\mathrm{AvL}=\frac{1}{M}\sum_{m=1}^M(U_m-L_m).
\]
Following \citet{Gneiting2007}, we compute the average interval score as
\[
\mathrm{IS}=\frac{1}{M}\sum_{m=1}^M
\left\{
U_m-L_m
+\frac{2}{0.05}(L_m-\theta^*)\mathbf 1(\theta^*<L_m)
+\frac{2}{0.05}(\theta^*-U_m)\mathbf 1(\theta^*>U_m)
\right\}.
\]
The interval score penalizes both interval width and noncoverage, with smaller values indicating better performance. Average marginal measures are obtained by taking the arithmetic mean of the corresponding coefficient-specific measures.
 Additional experiments vary $n$ to examine intercept correction and calibration separately. In the direct AL--LPAL interval comparison, ``Original'' denotes equal-tailed credible intervals from the uncorrected working posterior, whereas ``Calibrated'' denotes confidence-type intervals constructed by Algorithm~\ref{alg:gp_bb_calibration}; the latter incorporate residual-quantile correction for intercept inference while retaining the original slope point estimates.

\begin{figure}[!ht]
\centering
\includegraphics[width=\textwidth]{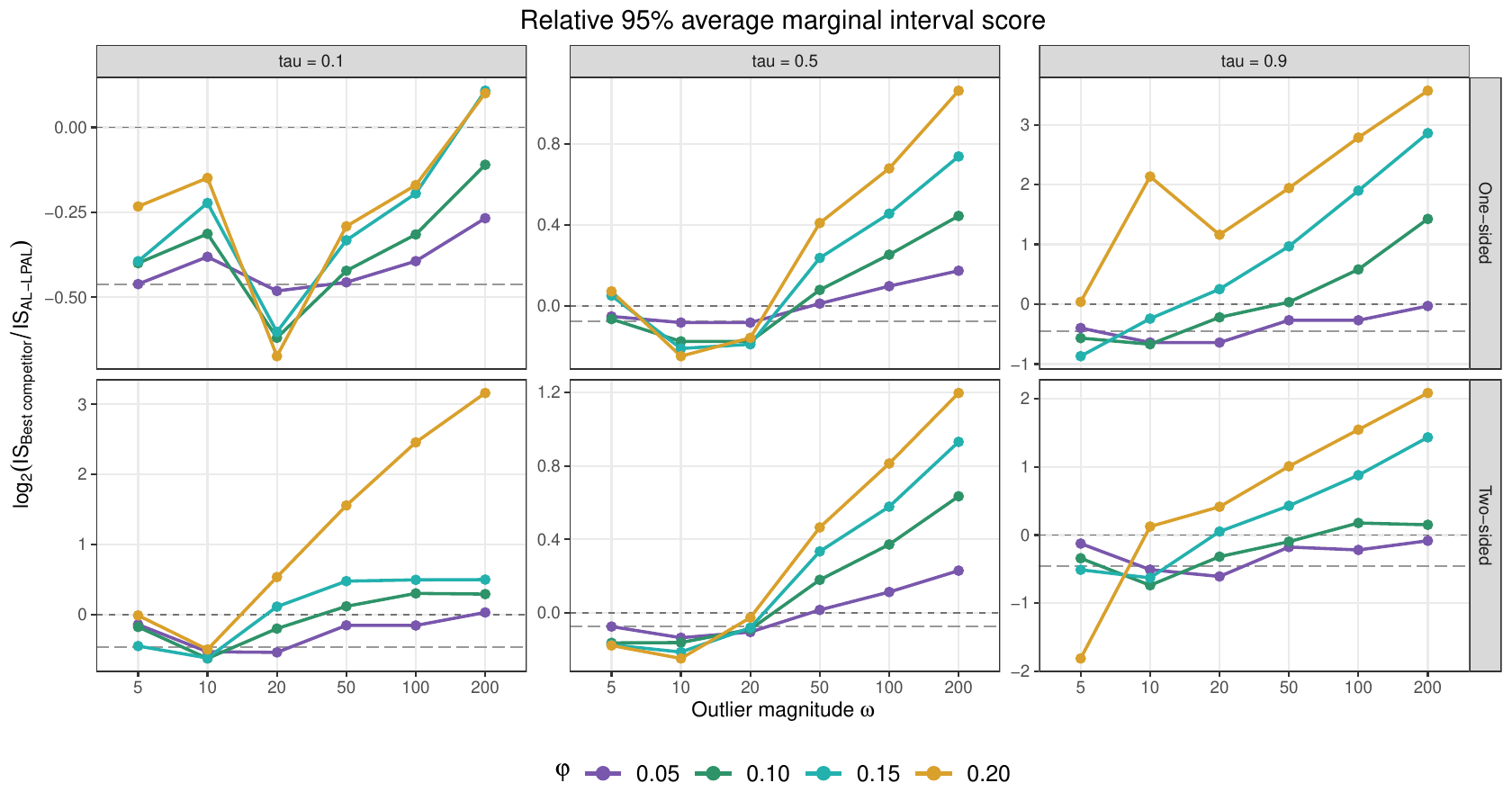}
\caption{Relative 95\% average marginal interval scores under one-sided (top) and two-sided (bottom) contamination at $\tau=0.1$, $0.5$, and $0.9$. The plotted quantity is $D_{\mathrm{IS}}=\log_2(\mathrm{IS}_{\mathrm{best}}/\mathrm{IS}_{\mathrm{AL\text{-}LPAL}})$, where $\mathrm{IS}_{\mathrm{best}}$ is the lowest average interval score among the five competing methods; positive values favor AL--LPAL. Curves correspond to nominal $\varphi$. The dotted horizontal line at $D_{\mathrm{IS}}=0$ indicates equal interval scores, while the long-dashed horizontal line indicates the corresponding $D_{\mathrm{IS}}$ under uncontaminated sampling at each quantile level.}
\label{fig:relative_is}
\end{figure}

Figure~\ref{fig:relative_is} shows that the relative interval-score performance of AL--LPAL generally improves as contamination becomes more severe. The advantage is particularly pronounced at $\tau=0.9$ under positive one-sided contamination and at large $\omega$ under two-sided contamination. In contrast, improvements at $\tau=0.1$ under positive one-sided contamination are less consistent. Under clean sampling, GAL attains the lowest interval score at each quantile level, illustrating an efficiency cost of the proposed robustification in the absence of extreme observations. Because AL--LPAL intervals in this comparison are calibrated, these results reflect the combined effects of its working likelihood and interval construction rather than isolating either component.

\begin{figure}[!ht]
\centering
\includegraphics[width=\textwidth]{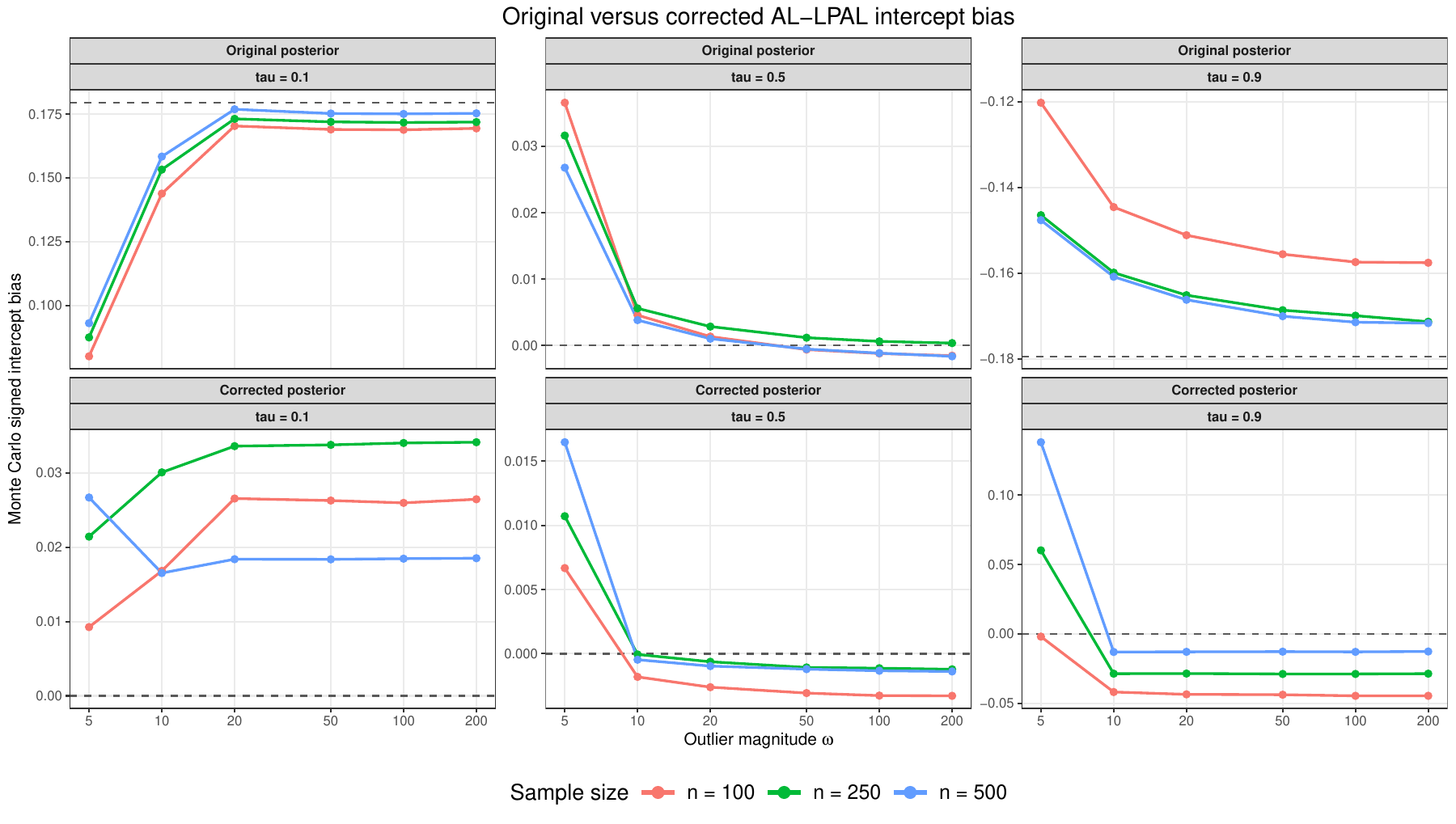}
\caption{Monte Carlo signed bias of the original and residual-quantile-corrected intercept estimates under 10\% one-sided contamination, for $n=100$, $250$, and $500$ and $\tau=0.1$, $0.5$, and $0.9$. Bias is measured relative to the clean target $\alpha_\tau^*$. The dashed horizontal lines indicate the numerically evaluated population intercept displacement $a^\dagger$ under the limiting contamination criterion with $\rho=1/9$ for the original posterior (top), and zero bias for the corrected posterior (bottom).}
\label{fig:intercept_bias}
\end{figure}

Figure~\ref{fig:intercept_bias} isolates intercept targeting under 10\% one-sided contamination for $n\in\{100,250,500\}$. As the contamination magnitude increases, the original posterior median retains signed biases of approximately $0.17$ and $-0.17$ at $\tau=0.1$ and $0.9$, respectively, while its median-regression bias remains comparatively small. The opposite signs at the lower and upper quantiles illustrate the dependence of intercept distortion on $\tau$. Residual-quantile correction substantially reduces these biases under sufficiently extreme contamination. For example, at $n=500$, $\tau=0.9$, and $\omega=200$, the bias decreases in magnitude from approximately $0.172$ to $0.013$. The correction is not uniformly beneficial at smaller contamination magnitudes; at $\tau=0.9$ and $\omega=5$, it can produce a relatively large positive bias. These findings support the targeting correction in the extreme-contamination settings considered, without implying finite-sample improvement in every configuration.

\begin{table}[!ht]
\centering
\footnotesize
\setlength{\tabcolsep}{4pt}
\begin{tabular}{cccccccc}
\toprule
& & \multicolumn{3}{c}{\textbf{Panel A: Intercept}}
  & \multicolumn{3}{c}{\textbf{Panel B: Average of three slopes}} \\
\cmidrule(lr){3-5}\cmidrule(lr){6-8}
Method & $n$ & CP & AvL & IS & CP & AvL & IS \\
\midrule
\multirow{3}{*}{Original}
 & 100 & 0.892 (0.014) & 0.252 (0.002) & 0.383 (0.024)
       & 0.913 (0.008) & 0.309 (0.003) & 0.444 (0.017) \\
 & 250 & 0.884 (0.014) & 0.152 (0.001) & 0.278 (0.023)
       & 0.884 (0.009) & 0.184 (0.001) & 0.308 (0.015) \\
 & 500 & 0.854 (0.016) & 0.106 (0.001) & 0.201 (0.014)
       & 0.848 (0.010) & 0.128 (0.001) & 0.251 (0.013) \\
\addlinespace[4pt]
\multirow{3}{*}{Calibrated}
 & 100 & 0.974 (0.007) & 0.304 (0.002) & 0.340 (0.013)
       & 0.952 (0.006) & 0.367 (0.002) & 0.443 (0.013) \\
 & 250 & 0.964 (0.008) & 0.184 (0.001) & 0.219 (0.010)
       & 0.961 (0.006) & 0.254 (0.001) & 0.289 (0.008) \\
 & 500 & 0.966 (0.008) & 0.128 (0.001) & 0.143 (0.005)
       & 0.955 (0.006) & 0.190 (0.001) & 0.220 (0.006) \\
\bottomrule
\end{tabular}
\caption{Empirical coverage probability (CP), average interval length (AvL), and interval score (IS) for nominal 95\% intervals under 10\% one-sided contamination at $\tau=0.5$ and $\omega=100$. Panel B reports the arithmetic average over the three slopes. Monte Carlo standard errors, based on 500 paired replications, are given in parentheses.}
\label{tab:gpbb_performance}
\end{table}

Table~\ref{tab:gpbb_performance} examines interval calibration at $\tau=0.5$ under 10\% one-sided contamination with $\omega=100$. Although intercept distortion is relatively small at the median, the original posterior intervals exhibit substantial undercoverage. At $n=500$, empirical coverage is $0.854$ for the intercept and $0.848$ on average for the slopes, compared with $0.966$ and $0.955$ after calibration. The calibrated intervals are longer but achieve lower interval scores. Since residual-quantile correction leaves the slope posterior unchanged, the slope results isolate the effect of interval calibration. The intercept comparison evaluates the complete corrected-and-calibrated procedure, for which targeting correction has a relatively limited role at the median.

\subsection{Airline delay data}
\label{sec:airline}

We illustrate the method using the 2013 New York City airline data from the \texttt{nycflights13} package \citep{Wickham2022}, with arrival delay in minutes as the response. After excluding missing arrival delays, we select a fixed training sample of $2{,}000$ flights stratified by month and departure airport and a disjoint holdout sample of $10{,}000$ flights. The predictors are indicators for Newark (EWR) and LaGuardia (LGA), with John F. Kennedy (JFK) as the reference airport; standardized log flight distance; standardized scheduled departure hour; and sine and cosine terms for annual seasonality. Together with the intercept, these give seven regression coefficients. The response is centered and scaled by the training-sample median and median absolute deviation, with reported results transformed back to minutes. We fit the AL--LPAL model at $\tau\in\{0.1,0.25,0.5,0.75,0.9\}$.

Figure~\ref{fig:airline_inference} compares original, corrected, and calibrated intercept inference, together with calibrated-to-original interval-width ratios across coefficients.

\begin{figure}[!ht]
\centering
\includegraphics[width=\textwidth]{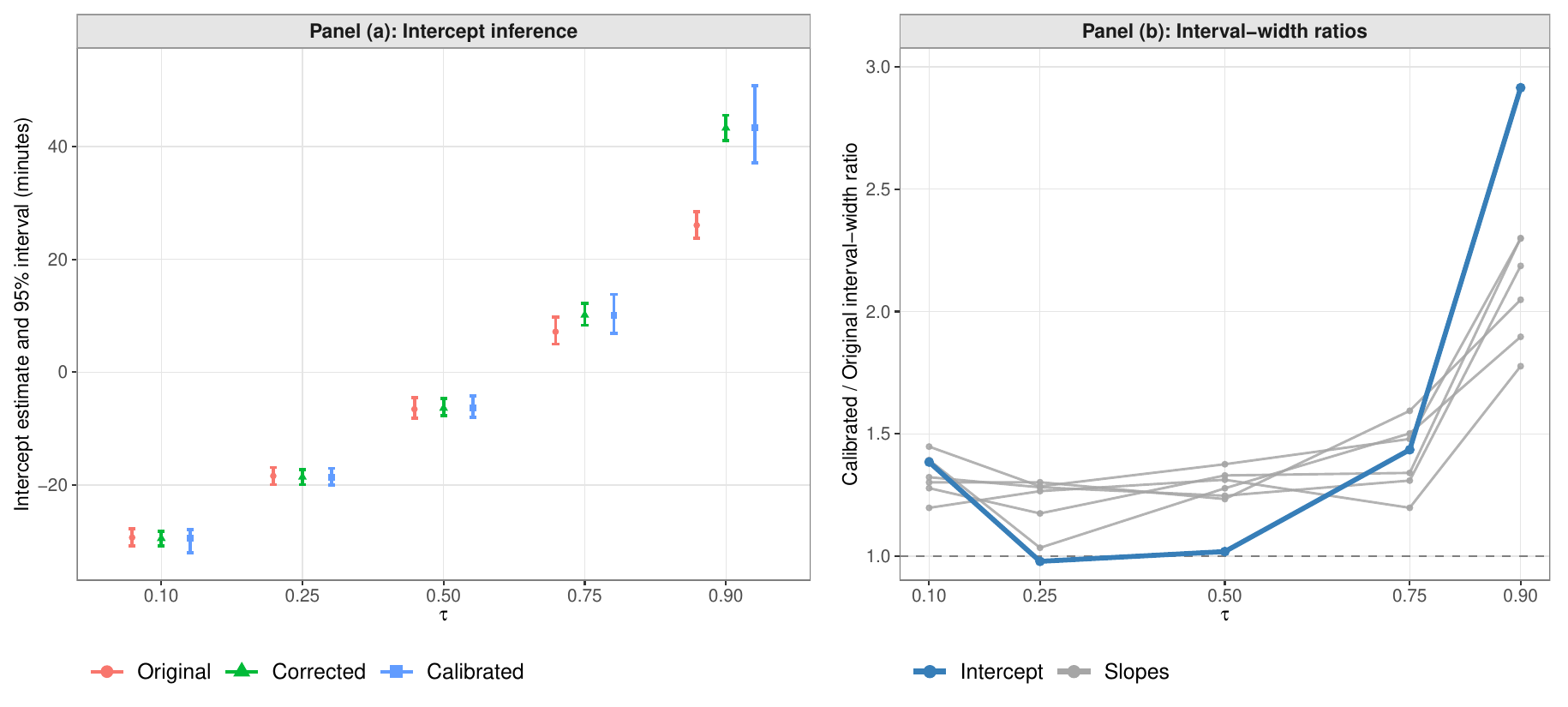}
\caption{Intercept estimates and 95\% intervals before and after residual-quantile correction and interval calibration (left), and calibrated-to-original interval-width ratios for the intercept and six slopes (right), across five quantile levels.}
\label{fig:airline_inference}
\end{figure}

The intercept adjustment in Figure~\ref{fig:airline_inference}(a) is strongly quantile-dependent. The original and corrected estimates are nearly identical at $\tau\le0.5$, whereas at $\tau=0.9$ the estimate increases from approximately $26.1$ to $43.3$ minutes, with nonoverlapping 95\% intervals. This difference illustrates the practical importance of distinguishing the working-posterior intercept from its residual-quantile-corrected counterpart in the original, unperturbed dataset. Calibration retains the corrected point estimate while modifying its interval. At $\tau=0.9$, the calibrated intercept interval is approximately $2.9$ times as wide as the original interval, with corresponding slope ratios of approximately $1.8$--$2.3$; by contrast, the intercept ratios are near unity at $\tau=0.25$ and $0.5$. These comparisons measure the size of the uncertainty adjustment, not frequentist coverage in the observed dataset.

For prediction, we evaluate average check loss on the fixed holdout sample. Across the five quantile levels, corrected AL--LPAL is within approximately $2.4\%$ of the best-performing method. At $\tau=0.9$, its average check loss is $10.19$ minutes, compared with $9.95$ minutes for the best competitor and $10.80$ minutes for original AL--LPAL. Thus, the substantial upper-quantile intercept adjustment is accompanied by improved holdout check loss relative to the uncorrected fit, although a competitor still performs slightly better.

We assess contamination sensitivity by selecting 20 training observations (1\%) and perturbing their arrival delays by $+\omega$ or $-\omega$ minutes for $\omega\in\{50,200,1000\}$. The selected indices are held fixed across magnitudes, directions, and methods. We refit classical quantile regression (QR), AL, GAL, GAH, cGAL, and AL--LPAL, retaining both original and corrected AL--LPAL estimates, at $\tau\in\{0.1,0.5,0.9\}$. For each method, the root mean squared (RMS) shift is computed from the difference between fitted conditional quantiles under perturbed and unperturbed training responses, evaluated at the same holdout covariates. The resulting measure quantifies sensitivity relative to each method's own unperturbed fit, not predictive error relative to the observed holdout responses.

\begin{figure}[!ht]
\centering
\includegraphics[width=\textwidth]{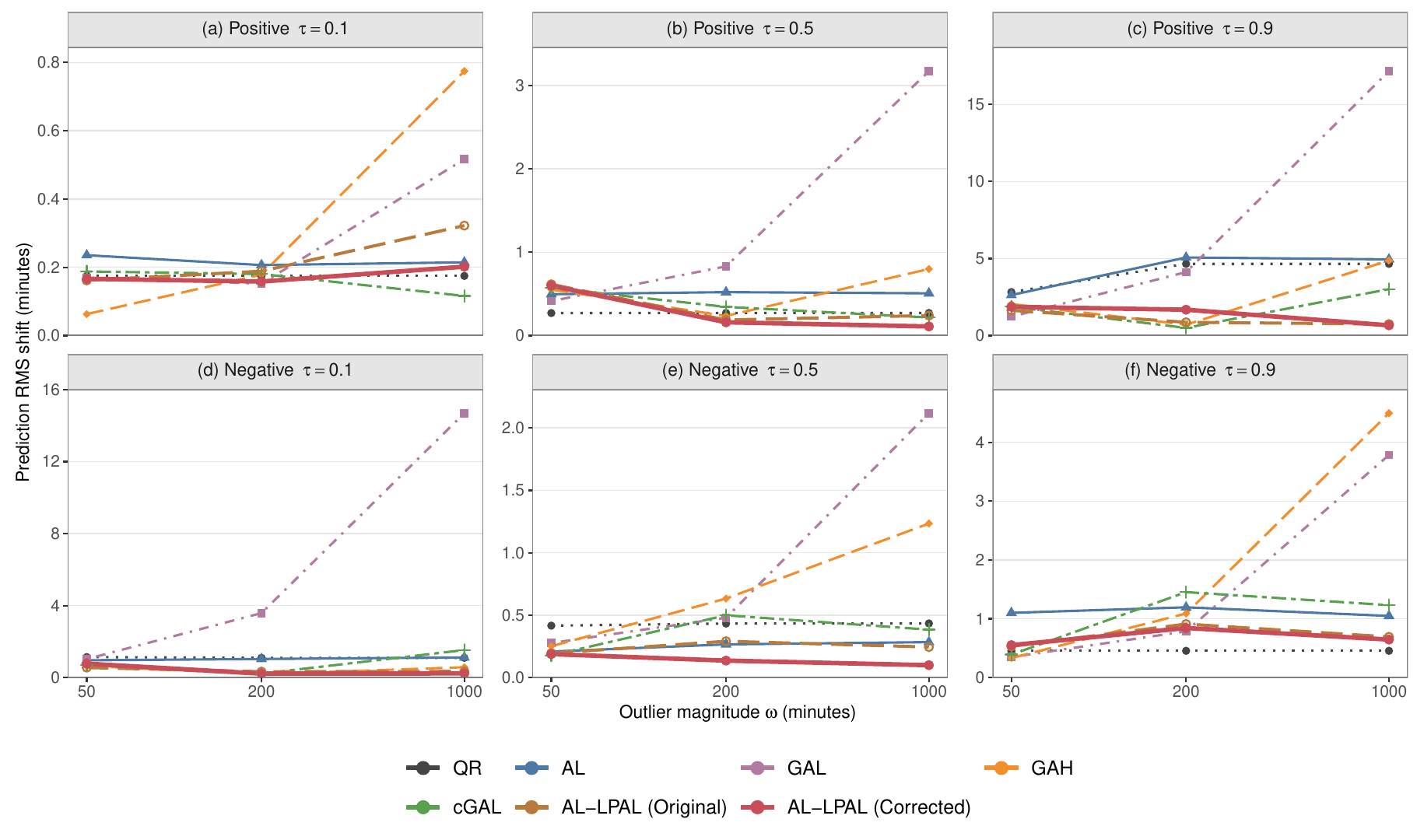}
\caption{RMS shifts of fitted conditional quantiles after perturbing 1\% of training responses by $+\omega$ or $-\omega$ minutes, for $\omega\in\{50,200,1000\}$ and $\tau\in\{0.1,0.5,0.9\}$. Shifts are measured against each method's fit to the unperturbed training sample; vertical scales differ across panels.}
\label{fig:airline_robustness}
\end{figure}

Figure~\ref{fig:airline_robustness} shows that contamination sensitivity depends on both direction and quantile level. The clearest contrasts occur under positive contamination at $\tau=0.9$ and negative contamination at $\tau=0.1$. For positive contamination with $\omega=1000$ and $\tau=0.9$, the RMS shifts are approximately $0.64$ minutes for original AL--LPAL and $0.50$ minutes after correction, compared with approximately $4.7$--$5.0$ minutes for QR and AL. Under negative contamination at $\tau=0.1$, corrected AL--LPAL has an RMS shift of approximately $0.21$ minutes, compared with about $1.1$ minutes for QR and AL; GAL is considerably more sensitive in both cases. The differences are less uniform when the contamination direction is opposite to the target tail: QR or cGAL can then produce comparable or smaller shifts. As the same set of 20 observations is perturbed throughout, this experiment provides a controlled sensitivity comparison rather than a general ranking across all contamination mechanisms.

\section{Discussion}
\label{sec:discussion}

Our results demonstrate a nonregular failure of posterior uncertainty quantification that cannot be explained by a conventional covariance mismatch. Residual-quantile correction restores posterior concentration around the uncontaminated regression coefficients, yet the corrected posterior collapses around a random center whose sampling fluctuations follow a Gaussian-process argmax limit. Consequently, equal-tailed corrected-posterior credible intervals have asymptotically zero frequentist coverage for contrasts with nondegenerate leading fluctuations. This separation between posterior concentration and sampling variability is not captured by regular misspecified Bernstein--von Mises approximations.

The analysis also identifies distinct roles for the components of our procedure. The log-regular tails attenuate the influence of increasingly extreme responses, residual-quantile correction restores the regression target, and Gaussian-process-based calibration approximates the sampling uncertainty of the slope estimator. Bayesian-bootstrap residual quantiles then propagate slope fluctuations and empirical-quantile variability to intercept inference. These confidence-type intervals are constructed from simulated sampling references, not by rescaling credible intervals from the corrected posterior. The simulations show improvements in severe-contamination regimes, alongside an efficiency cost under clean sampling and occasional deterioration from correction at moderate contamination magnitudes.

The theoretical results concern fixed-dimensional regression under joint growth of the sample size and contamination magnitude. Proposition~\ref{prop:gp_calibration} establishes consistency of the plug-in Gaussian-process reference under its specified theoretical regime, rather than nominal frequentist coverage for the fully implemented GP--Bayesian-bootstrap intervals. Establishing that coverage would require joint control of data-adaptive screening, smoothing, covariance and curvature estimation, numerical argmax approximation, and weighted residual quantiles. Directions with degenerate leading argmax fluctuations, including the intercept under centered designs, also call for a finer-scale analysis beyond the present $r_k$ normalization.

\bibliography{bibliography.bib}

@article{Bernardi2018,
  author  = {Bernardi, Mauro and Bottone, Marco and Petrella, Lea},
  title   = {Bayesian Quantile Regression Using the Skew Exponential Power Distribution},
  journal = {Computational Statistics \& Data Analysis},
  volume  = {126},
  pages   = {92--111},
  year    = {2018}
}

@article{Desgagne2015,
  author  = {Desgagn{\'e}, Alain},
  title   = {Robustness to Outliers in Location--Scale Parameter Model Using Log-Regularly Varying Distributions},
  journal = {The Annals of Statistics},
  volume  = {43},
  number  = {4},
  pages   = {1568--1595},
  year    = {2015}
}

@article{GalarzaMorales2017,
  author  = {Galarza Morales, Christian and Lachos Davila, Victor and Barbosa Cabral, Celso and Castro Cepero, Luis},
  title   = {Robust Quantile Regression Using a Generalized Class of Skewed Distributions},
  journal = {Stat},
  volume  = {6},
  number  = {1},
  pages   = {113--130},
  year    = {2017}
}

@article{Hamura2022,
  author  = {Hamura, Yasuyuki and Irie, Kaoru and Sugasawa, Shonosuke},
  title   = {Log-Regularly Varying Scale Mixture of Normals for Robust Regression},
  journal = {Computational Statistics \& Data Analysis},
  volume  = {173},
  pages   = {107517},
  year    = {2022}
}

@article{Hu2024,
  author  = {Hu, Weitao and Zhang, Weiping},
  title   = {Flexible {B}ayesian Quantile Regression Based on the Generalized Asymmetric Huberised-Type Distribution},
  journal = {Statistics and Computing},
  volume  = {34},
  number  = {4},
  pages   = {144},
  year    = {2024}
}

@article{Kim1990,
  author  = {Kim, Jeankyung and Pollard, David},
  title   = {Cube Root Asymptotics},
  journal = {The Annals of Statistics},
  volume  = {18},
  number  = {1},
  pages   = {191--219},
  year    = {1990}
}

@article{Kleijn2012,
  author  = {Kleijn, B. J. K. and van der Vaart, A. W.},
  title   = {The {Bernstein--von Mises} Theorem under Misspecification},
  journal = {Electronic Journal of Statistics},
  volume  = {6},
  pages   = {354--381},
  year    = {2012}
}

@article{Koenker1978,
  author  = {Koenker, Roger and Bassett, Gilbert},
  title   = {Regression Quantiles},
  journal = {Econometrica},
  volume  = {46},
  number  = {1},
  pages   = {33--50},
  year    = {1978}
}

@article{Kozumi2011,
  author  = {Kozumi, Hideo and Kobayashi, Genya},
  title   = {Gibbs Sampling Methods for {B}ayesian Quantile Regression},
  journal = {Journal of Statistical Computation and Simulation},
  volume  = {81},
  number  = {11},
  pages   = {1565--1578},
  year    = {2011}
}

@article{Sriram2013,
  author  = {Sriram, Karthik and Ramamoorthi, R. V. and Ghosh, Pulak},
  title   = {Posterior Consistency of {B}ayesian Quantile Regression Based on the Misspecified Asymmetric Laplace Density},
  journal = {Bayesian Analysis},
  volume  = {8},
  number  = {2},
  pages   = {479--504},
  year    = {2013}
}

@article{Sriram2015,
  author  = {Sriram, Karthik},
  title   = {A Sandwich Likelihood Correction for {B}ayesian Quantile Regression Based on the Misspecified Asymmetric Laplace Density},
  journal = {Statistics \& Probability Letters},
  volume  = {107},
  pages   = {18--26},
  year    = {2015}
}

@article{Yang2016,
  author  = {Yang, Yunwen and Wang, Huixia Judy and He, Xuming},
  title   = {Posterior Inference in {B}ayesian Quantile Regression with Asymmetric Laplace Likelihood},
  journal = {International Statistical Review},
  volume  = {84},
  number  = {3},
  pages   = {327--344},
  year    = {2016}
}

@article{Yan2026,
  author  = {Yan, Yifei and Zheng, Xiaotian and Kottas, Athanasios},
  title   = {A New Family of Error Distributions for {B}ayesian Quantile Regression},
  journal = {Bayesian Analysis},
  year    = {2026},
  volume  = {21},
  number  = {3},
  pages   = {1029-1057}
}

@article{Yu2001,
  author  = {Yu, Keming and Moyeed, Rana A.},
  title   = {{B}ayesian Quantile Regression},
  journal = {Statistics \& Probability Letters},
  volume  = {54},
  number  = {4},
  pages   = {437--447},
  year    = {2001}
}

@article{Bissiri2016,
  author  = {Bissiri, P. G. and Holmes, C. C. and Walker, S. G.},
  title   = {A General Framework for Updating Belief Distributions},
  journal = {Journal of the Royal Statistical Society: Series B},
  volume  = {78},
  number  = {5},
  pages   = {1103--1130},
  year    = {2016}
}

@article{Box1968,
  author  = {Box, G. E. P. and Tiao, G. C.},
  title   = {A {B}ayesian Approach to Some Outlier Problems},
  journal = {Biometrika},
  volume  = {55},
  number  = {1},
  pages   = {119--129},
  year    = {1968}
}

@article{Cattaneo2020,
  author  = {Cattaneo, Matias D. and Jansson, Michael and Nagasawa, Kenichi},
  title   = {Bootstrap-Based Inference for Cube Root Asymptotics},
  journal = {Econometrica},
  volume  = {88},
  number  = {5},
  pages   = {2203--2219},
  year    = {2020}
}

@article{Chernozhukov2003,
  author  = {Chernozhukov, Victor and Hong, Han},
  title   = {An {MCMC} Approach to Classical Estimation},
  journal = {Journal of Econometrics},
  volume  = {115},
  number  = {2},
  pages   = {293--346},
  year    = {2003}
}

@article{Gagnon2020,
  author  = {Gagnon, Philippe and Desgagn{\'e}, Alain and B{\'e}dard, Myl{\`e}ne},
  title   = {A New {B}ayesian Approach to Robustness Against Outliers in Linear Regression},
  journal = {Bayesian Analysis},
  volume  = {15},
  number  = {2},
  pages   = {389--414},
  year    = {2020}
}

@article{Hashimoto2020,
  author  = {Hashimoto, Shintaro and Sugasawa, Shonosuke},
  title   = {Robust {B}ayesian Regression with Synthetic Posterior Distributions},
  journal = {Entropy},
  volume  = {22},
  number  = {6},
  pages   = {661},
  year    = {2020}
}

@article{Kottas2009,
  author  = {Kottas, Athanasios and Krnjaji{\'c}, Milovan},
  title   = {{B}ayesian Semiparametric Modelling in Quantile Regression},
  journal = {Scandinavian Journal of Statistics},
  volume  = {36},
  number  = {2},
  pages   = {297--319},
  year    = {2009}
}

@article{Muller2013,
  author  = {M{\"u}ller, Ulrich K.},
  title   = {Risk of {B}ayesian Inference in Misspecified Models, and the Sandwich Covariance Matrix},
  journal = {Econometrica},
  volume  = {81},
  number  = {5},
  pages   = {1805--1849},
  year    = {2013}
}

@article{OHagan1979,
  author  = {O'Hagan, A.},
  title   = {On Outlier Rejection Phenomena in {B}ayes Inference},
  journal = {Journal of the Royal Statistical Society: Series B},
  volume  = {41},
  number  = {3},
  pages   = {358--367},
  year    = {1979}
}

@article{Reich2010,
  author  = {Reich, Brian J. and Bondell, Howard D. and Wang, Huixia J.},
  title   = {Flexible {B}ayesian Quantile Regression for Independent and Clustered Data},
  journal = {Biostatistics},
  volume  = {11},
  number  = {2},
  pages   = {337--352},
  year    = {2010}
}

@article{Rubin1981,
  author  = {Rubin, Donald B.},
  title   = {The {B}ayesian Bootstrap},
  journal = {The Annals of Statistics},
  volume  = {9},
  number  = {1},
  pages   = {130--134},
  year    = {1981}
}

@article{Wichitaksorn2014,
  author  = {Wichitaksorn, Nuttanan and Choy, S. T. Boris and Gerlach, Richard},
  title   = {A Generalized Class of Skew Distributions and Associated Robust Quantile Regression Models},
  journal = {Canadian Journal of Statistics},
  volume  = {42},
  number  = {4},
  pages   = {579--596},
  year    = {2014}
}

@article{vanDykPark2008,
  author  = {van Dyk, David A. and Park, Taeyoung},
  title   = {Partially Collapsed {G}ibbs Samplers: Theory and Methods},
  journal = {Journal of the American Statistical Association},
  year    = {2008},
  volume  = {103},
  number  = {482},
  pages   = {790--796}
}

@article{Neal2003,
  author  = {Neal, Radford M.},
  title   = {Slice Sampling},
  journal = {The Annals of Statistics},
  year    = {2003},
  volume  = {31},
  number  = {3},
  pages   = {705--767}
}

@book{Ibragimov1981,
  author    = {Ibragimov, I. A. and Has'minskii, R. Z.},
  title     = {Statistical Estimation: Asymptotic Theory},
  publisher = {Springer-Verlag},
  address   = {New York},
  year      = {1981},
  series    = {Applications of Mathematics},
  volume    = {16}
}

@article{Kleijn2006,
  author  = {Kleijn, B. J. K. and van der Vaart, A. W.},
  title   = {Misspecification in Infinite-Dimensional {B}ayesian Statistics},
  journal = {The Annals of Statistics},
  year    = {2006},
  volume  = {34},
  number  = {2},
  pages   = {837--877}
}

@article{Lo1987,
  author  = {Lo, Albert Y.},
  title   = {A Large Sample Study of the {B}ayesian Bootstrap},
  journal = {The Annals of Statistics},
  year    = {1987},
  volume  = {15},
  number  = {1},
  pages   = {360--375}
}

@article{Bahadur1966,
  author  = {Bahadur, R. R.},
  title   = {A Note on Quantiles in Large Samples},
  journal = {The Annals of Mathematical Statistics},
  year    = {1966},
  volume  = {37},
  number  = {3},
  pages   = {577--580}
}

@article{Rousseeuw1993,
  author  = {Rousseeuw, Peter J. and Croux, Christophe},
  title   = {Alternatives to the Median Absolute Deviation},
  journal = {Journal of the American Statistical Association},
  year    = {1993},
  volume  = {88},
  number  = {424},
  pages   = {1273--1283}
}

@book{Bingham1987,
  author    = {Bingham, N. H. and Goldie, C. M. and Teugels, J. L.},
  title     = {Regular Variation},
  publisher = {Cambridge University Press},
  series    = {Encyclopedia of Mathematics and its Applications},
  volume    = {27},
  year      = {1987}
}

@book{vanDerVaart1996,
  author    = {van der Vaart, Aad W. and Wellner, Jon A.},
  title     = {Weak Convergence and Empirical Processes: With Applications to Statistics},
  series    = {Springer Series in Statistics},
  publisher = {Springer},
  address   = {New York},
  year      = {1996}
}

@article{Vehtari2021,
  author  = {Vehtari, Aki and Gelman, Andrew and Simpson, Daniel and Carpenter, Bob and B{\"u}rkner, Paul-Christian},
  title   = {Rank-Normalization, Folding, and Localization: An Improved {$\widehat{R}$} for Assessing Convergence of {MCMC}},
  journal = {Bayesian Analysis},
  year    = {2021},
  volume  = {16},
  number  = {2},
  pages   = {667--718}
}

@book{Adler2007,
  author    = {Adler, Robert J. and Taylor, Jonathan E.},
  title     = {Random Fields and Geometry},
  series    = {Springer Monographs in Mathematics},
  publisher = {Springer},
  address   = {New York},
  year      = {2007}
}

@article{Groeneboom2001,
  author  = {Groeneboom, Piet and Wellner, Jon A.},
  title   = {Computing {Chernoff}'s Distribution},
  journal = {Journal of Computational and Graphical Statistics},
  year    = {2001},
  volume  = {10},
  number  = {2},
  pages   = {388--400}
}

@article{Murphy2000,
  author  = {Murphy, Susan A. and van der Vaart, Aad W.},
  title   = {On Profile Likelihood},
  journal = {Journal of the American Statistical Association},
  year    = {2000},
  volume  = {95},
  number  = {450},
  pages   = {449--465}
}

@article{Horowitz1998,
  author  = {Horowitz, Joel L.},
  title   = {Bootstrap Methods for Median Regression Models},
  journal = {Econometrica},
  year    = {1998},
  volume  = {66},
  number  = {6},
  pages   = {1327--1351}
}

@article{Burger2025,
  author  = {Burger, Divan A. and van der Merwe, Sean and Lesaffre, Emmanuel},
  title   = {A Robust Mixed-Effects Quantile Regression Model Using Generalized Laplace Mixtures to Handle Outliers and Skewness},
  journal = {Computational Statistics},
  year    = {2025},
  volume  = {40},
  number  = {8},
  pages   = {4775--4798}
}

@Manual{Wickham2022,
  title  = {nycflights13: Flights that Departed NYC in 2013},
  author = {Hadley Wickham},
  year   = {2022},
  note   = {R package version 1.0.2},
  url    = {https://github.com/hadley/nycflights13}
}

@article{Jun2015,
  author  = {Jun, Sung Jae and Pinkse, Joris and Wan, Yuanyuan},
  title   = {Classical {Laplace} Estimation for {$\sqrt[3]{n}$}-Consistent Estimators: Improved Convergence Rates and Rate-Adaptive Inference},
  journal = {Journal of Econometrics},
  year    = {2015},
  volume  = {187},
  number  = {1},
  pages   = {201--216}
}

@article{Podgorski2015,
  author  = {Podg{\'o}rski, Krzysztof and Wallin, Jonas},
  title   = {Maximizing Leave-One-Out Likelihood for the Location Parameter of Unbounded Densities},
  journal = {Annals of the Institute of Statistical Mathematics},
  year    = {2015},
  volume  = {67},
  number  = {1},
  pages   = {19--38}
}

@article{Syring2017,
  author  = {Syring, Nicholas and Martin, Ryan},
  title   = {Gibbs Posterior Inference on the Minimum Clinically Important Difference},
  journal = {Journal of Statistical Planning and Inference},
  year    = {2017},
  volume  = {187},
  pages   = {67--77}
}

@article{Dachian2018,
  author  = {Dachian, S. and Kordzakhia, N. and Kutoyants, Y. A. and Novikov, A.},
  title   = {Estimation of Cusp Location of Stochastic Processes: A Survey},
  journal = {Statistical Inference for Stochastic Processes},
  year    = {2018},
  volume  = {21},
  number  = {2},
  pages   = {345--362}
}

@article{Chernoyarov2018,
  author  = {Chernoyarov, O. V. and Dachian, S. and Kutoyants, Y. A.},
  title   = {On Parameter Estimation for Cusp-Type Signals},
  journal = {Annals of the Institute of Statistical Mathematics},
  year    = {2018},
  volume  = {70},
  number  = {1},
  pages   = {39--62}
}

@article{Gneiting2007,
  author  = {Gneiting, Tilmann and Raftery, Adrian E.},
  title   = {Strictly Proper Scoring Rules, Prediction, and Estimation},
  journal = {Journal of the American Statistical Association},
  year    = {2007},
  volume  = {102},
  number  = {477},
  pages   = {359--378}
}

\end{document}


\renewcommand{\thesection}{S\arabic{section}}
\renewcommand{\thesubsection}{S\arabic{section}.\arabic{subsection}}

\renewcommand{\theequation}{S\arabic{equation}}
\renewcommand{\thefigure}{S\arabic{figure}}
\renewcommand{\thetable}{S\arabic{table}}

\setcounter{section}{0}
\setcounter{equation}{0}
\setcounter{figure}{0}
\setcounter{table}{0}

\maketitle

\section{Proofs of Propositions 1--4}
\label{sec:supp_propositions}

Following the log-Pareto scale-mixing construction of \citet{Hamura2022}, we use the density $H(v;\gamma)=\gamma/[(1+v)\{1+\log(1+v)\}^{1+\gamma}]$ for $v>0$. Throughout this section, put $c_\tau=\tau(1-\tau)$, $\rho_\tau(u)=u\{\tau-\mathbf{1}(u<0)\}$, $f_0(u)=c_\tau\exp\{-\rho_\tau(u)\}$, $f_1(u)=\int_0^\infty c_\tau v^{-1}\exp\{-\rho_\tau(u)/v\}H(v;\gamma)\,dv$, and $f_s=(1-s)f_0+sf_1=f_\tau(\cdot;s,\gamma)$. Write $q_{\bm\theta}(y\mid\bm x)=\sigma^{-1}f_s\{(y-\bm x^\top\bm\beta)/\sigma\}$, with $\bm x=(1,\bm z^\top)^\top$ and $\bm\beta=(\alpha,\bm b^\top)^\top$. The priors for $\bm\beta$, $\sigma$, and $s$ are independent Gaussian, inverse-gamma, and beta distributions with the fixed positive parameters specified in the manuscript. The constants below may depend on $\tau\in(0,1)$ and $\gamma>0$, but not on $u$ or $s\in[0,1]$.

\begin{proof}[Proof of Proposition 1]
The substitution $t=1+\log(1+v)$ gives
\[
\int_0^\infty H(v;\gamma)\,dv=\int_1^\infty\gamma t^{-1-\gamma}\,dt=1.
\]
For each $v>0$, the AL density with scale $v$ assigns probability $\tau$ to $(-\infty,0]$. Tonelli's theorem therefore yields
\[
\int_{-\infty}^0 f_1(u)\,du=\int_0^\infty\left\{\int_{-\infty}^0c_\tau v^{-1}e^{-\rho_\tau(u)/v}\,du\right\}H(v;\gamma)\,dv=\tau.
\]
The same identity holds for $f_0$, and hence for $f_s$.

For $a>0$, define
\[
K_\gamma(a)=\int_0^\infty v^{-1}e^{-a/v}H(v;\gamma)\,dv,
\qquad f_1(u)=c_\tau K_\gamma\{\rho_\tau(u)\}.
\]
Since $H(v;\gamma)=\gamma+O(v)$ as $v\downarrow0$ and $\int_1^\infty v^{-1}H(v;\gamma)\,dv<\infty$, we have, as $a\downarrow0$,
\begin{align*}
K_\gamma(a)&=\gamma\int_0^1v^{-1}e^{-a/v}\,dv+O(1)\\
&=\gamma\int_a^\infty w^{-1}e^{-w}\,dw+O(1)
=\gamma\log(1/a)+O(1).
\end{align*}
For $u\ne0$, $\rho_\tau(u)=\tau u$ when $u>0$ and $\rho_\tau(u)=(1-\tau)|u|$ when $u<0$. Thus $\log\{1/\rho_\tau(u)\}=\log(1/|u|)+O(1)$ as $u\to0$, and $f_1(u)=\gamma c_\tau\log(1/|u|)+O(1)$. Because $f_0$ is bounded at zero, $f_s(u)/\log(1/|u|)\to s\gamma\tau(1-\tau)$.

The mixing density is regularly varying at infinity with index $-1$:
\[
H(v;\gamma)\sim\frac{\gamma}{v(\log v)^{1+\gamma}},\qquad v\to\infty.
\]
The change of variables $w=a/v$ gives
\begin{equation}
\label{eq:s1_K_change}
K_\gamma(a)=\int_0^\infty e^{-w}\frac{H(a/w;\gamma)}{w}\,dw.
\end{equation}
For each fixed $w>0$, $H(a/w;\gamma)/H(a;\gamma)\to w$ as $a\to\infty$. To justify integration, fix $\delta\in(0,1)$ and a sufficiently large $v_0$. By Potter's bound \citep[Theorem~1.5.6]{Bingham1987}, on $0<w\le a/v_0$ the regular-variation bounds give $H(a/w;\gamma)/H(a;\gamma)\le C\max\{w^{1-\delta},w^{1+\delta}\}$. The resulting integrand is bounded by $Ce^{-w}\max\{w^{-\delta},w^\delta\}$, which is integrable. On $w>a/v_0$, $H(a/w;\gamma)\le\gamma$, while the exponential factor makes the integral $o\{H(a;\gamma)\}$. Dominated convergence in \eqref{eq:s1_K_change} therefore gives
\[
\frac{K_\gamma(a)}{H(a;\gamma)}\longrightarrow\int_0^\infty e^{-w}\,dw=1,
\qquad
K_\gamma(a)\sim\frac{\gamma}{a(\log a)^{1+\gamma}}.
\]
Consequently,
\[
f_1(u)\sim
\begin{cases}
\displaystyle\frac{\gamma(1-\tau)}{u(\log u)^{1+\gamma}},&u\to+\infty,\\[7pt]
\displaystyle\frac{\gamma\tau}{|u|(\log|u|)^{1+\gamma}},&u\to-\infty.
\end{cases}
\]
The AL component $f_0$ has exponentially decaying tails and is negligible relative to $f_1$ in both directions. Multiplication by $s$ yields the two claimed limits for $f_s$.
\end{proof}

\begin{lemma}[Analytic bounds for the AL--LPAL density]
\label{lem:s1_analytic_bounds}
For fixed $\tau\in(0,1)$ and $\gamma>0$, the following statements hold.
\begin{enumerate}
\item[(i)] There exist constants $c,C>0$ such that, for all $u\ne0$ and $s\in[0,1]$,
\begin{equation}
\label{eq:s1_density_envelope}
ce^{-|u|}\le f_s(u)\le C\{1+\log_+(1/|u|)\}.
\end{equation}
\item[(ii)] For $u\ne0$, the component densities and their mixtures are differentiable and satisfy
\begin{equation}
\label{eq:s1_log_derivative}
|\partial_u\log f_s(u)|\le C(1+|u|^{-1}),\qquad s\in[0,1].
\end{equation}
Moreover, for some $u_0>1$ and $c,C>0$,
\begin{equation}
\label{eq:s1_tail_derivatives}
|u\partial_u\log f_1(u)|\le C,\qquad
\frac{f_0(u)}{f_1(u)}\le C(1+|u|)^{2+\gamma}e^{-c|u|},\qquad |u|\ge u_0,
\end{equation}
and $\sup_{u\ne0}f_0(u)/f_1(u)<\infty$.
\item[(iii)] If $U$ has a bounded density $g$ and $E_g|U|<\infty$, then, for every compact $D\subset\mathbb R\times(0,\infty)\times(0,1)$,
\[
\sup_{(a,\sigma,s)\in D}E_g\left|\log\left\{\sigma^{-1}f_s\left(\frac{U-a}{\sigma}\right)\right\}\right|<\infty.
\]
The derivatives of this log density with respect to $\log\sigma$ and $\operatorname{logit}s$ also have uniformly bounded absolute expectations over $D$.
\end{enumerate}
\end{lemma}

\begin{proof}
Since $H(v;\gamma)=\gamma+O(v)$ near zero, the integral defining $f_1$ over $0<v<1$ is bounded above by $C\{1+\log_+(1/\rho_\tau(u))\}$; the integral over $v\ge1$ is bounded by $c_\tau\int_1^\infty v^{-1}H(v;\gamma)\,dv$. Together with boundedness of $f_0$ and $\rho_\tau(u)\asymp|u|$, this proves the upper bound in \eqref{eq:s1_density_envelope}. For the lower bound, $f_0(u)=c_\tau e^{-\rho_\tau(u)}\ge c_\tau e^{-|u|}$, while restricting the mixing integral to $1\le v\le2$ gives
\[
f_1(u)\ge c_\tau e^{-\rho_\tau(u)}\int_1^2v^{-1}H(v;\gamma)\,dv\ge ce^{-|u|}.
\]
The same lower bound, with a possibly smaller constant, holds for every mixture $f_s$.

Differentiation under the defining integral gives, for $a>0$,
\[
K_\gamma'(a)=-\int_0^\infty v^{-2}e^{-a/v}H(v;\gamma)\,dv,
\qquad
aK_\gamma'(a)=-\int_0^\infty e^{-w}H(a/w;\gamma)\,dw.
\]
Because $H(v;\gamma)\le\gamma$ and $H(v;\gamma)\to\gamma$ as $v\downarrow0$, dominated convergence implies $aK_\gamma'(a)\to-\gamma$ as $a\downarrow0$. The central expansion in the proof of Proposition 1 then gives $aK_\gamma'(a)/K_\gamma(a)\to0$. At infinity, the same regular-variation domination used in \eqref{eq:s1_K_change} yields
\[
\frac{aK_\gamma'(a)}{H(a;\gamma)}\longrightarrow-\int_0^\infty we^{-w}\,dw=-1,
\qquad
\frac{aK_\gamma'(a)}{K_\gamma(a)}\longrightarrow-1.
\]
Continuity and positivity of $K_\gamma$ therefore imply $\sup_{a>0}|aK_\gamma'(a)/K_\gamma(a)|<\infty$. Since $a=\rho_\tau(u)$ is proportional to $|u|$ on either side of zero, $|u\partial_u\log f_1(u)|\le C$ for every $u\ne0$. On each half-line, $|\partial_u\log f_0(u)|\le1$. The logarithmic derivative of $f_s$ is a convex combination of the two component logarithmic derivatives, which proves \eqref{eq:s1_log_derivative} and the first bound in \eqref{eq:s1_tail_derivatives}. The two-sided tail equivalents from Proposition 1 and the exponential tails of $f_0$ give the second bound in \eqref{eq:s1_tail_derivatives}. Finally, $f_0/f_1\to0$ both as $u\to0$ and as $|u|\to\infty$, and the ratio is continuous away from zero, proving its boundedness.

From \eqref{eq:s1_density_envelope},
\[
-C(1+|u|)\le\log f_s(u)\le C+\log\{1+\log_+(1/|u|)\}.
\]
The positive part is locally integrable against any bounded density because $\int_{-1}^1\log\{1+\log_+(1/|u|)\}\,du<\infty$; the negative part is integrable under the stated first-moment condition. These bounds remain uniform for locations in bounded sets and scales bounded above and away from zero, proving the asserted uniform expectation bound. For $v=(U-a)/\sigma\ne0$, the two nuisance derivatives are
\[
\partial_{\log\sigma}\log\{\sigma^{-1}f_s(v)\}=-1-v\partial_v\log f_s(v),
\qquad
\partial_{\operatorname{logit}s}\log f_s(v)=\frac{s(1-s)\{f_1(v)-f_0(v)\}}{f_s(v)}.
\]
The second derivative expression has absolute value at most one. On a compact subset of $s\in(0,1)$, the ratio $f_0/f_1$ is bounded, and the exponentially small tail ratio in \eqref{eq:s1_tail_derivatives} implies $\sup_{v\ne0}|v\partial_v\log f_s(v)|<\infty$ uniformly in $s$. The first nuisance derivative is therefore uniformly bounded as well.
\end{proof}

\begin{proof}[Proof of Proposition 2]
Fix an arbitrary dataset $\{(y_i,\bm x_i):i=1,\ldots,n\}$ with $n\ge1$, finite responses, and $\bm x_i=(1,\bm z_i^\top)^\top$. Let $\bm\beta=(\alpha,\bm b^\top)^\top$, and write $\pi_\beta$, $\pi_\sigma$, and $\pi_s$ for the independent Gaussian, inverse-gamma, and beta priors. By \eqref{eq:s1_density_envelope}, for every $\sigma>0$, $s\in(0,1)$, and all parameter values outside a Gaussian-prior-null set,
\begin{equation}
\label{eq:s1_likelihood_majorant}
q_{\bm\theta}(y_i\mid\bm x_i)
\le\frac{C}{\sigma}(1+\log_+\sigma)
\left\{1+\log_+\frac1{|y_i-\alpha-\bm z_i^\top\bm b|}\right\}.
\end{equation}
Conditionally on $\bm b$, the Gaussian prior distribution of $\alpha$ has a fixed strictly positive variance $v_\alpha$; for $p=1$, this is simply the marginal variance of $\alpha$. Set $d_i(\bm b)=y_i-\bm z_i^\top\bm b$ and $W_i=1+\log_+\{1/|\alpha-d_i(\bm b)|\}$. Since the conditional Gaussian density of $\alpha$ is uniformly bounded by $(2\pi v_\alpha)^{-1/2}$, for every integer $j\ge1$,
\begin{equation}
\label{eq:s1_log_moment}
\sup_{\bm b}E_{\pi_\beta}(W_i^j\mid\bm b)
\le1+\frac{2}{\sqrt{2\pi v_\alpha}}\int_0^1\{1+\log(1/t)\}^j\,dt<\infty.
\end{equation}
The right-hand side does not depend on the observed $y_i$ or $\bm z_i$. Conditional H\"older inequality, with exponent $n$ for each factor, now gives
\[
E_{\pi_\beta}\left\{\left(\prod_{i=1}^nW_i\right)^2\middle|\bm b\right\}
\le\prod_{i=1}^n\left\{E_{\pi_\beta}(W_i^{2n}\mid\bm b)\right\}^{1/n}\le C_n<\infty.
\]
After integration in $\bm b$, the same bound holds unconditionally. Therefore, for any $r>0$ and each coordinate $\beta_j$, Cauchy--Schwarz and the existence of all Gaussian prior moments yield
\begin{equation}
\label{eq:s1_weighted_likelihood_bound}
E_{\pi_\beta}\left\{|\beta_j|^r\prod_{i=1}^nW_i\right\}
\le\{E_{\pi_\beta}|\beta_j|^{2r}\}^{1/2}
\left[E_{\pi_\beta}\left\{\left(\prod_{i=1}^nW_i\right)^2\right\}\right]^{1/2}<\infty.
\end{equation}
The same argument without $|\beta_j|^r$ controls the marginal likelihood. By \eqref{eq:s1_likelihood_majorant}--\eqref{eq:s1_weighted_likelihood_bound}, and because the bound is uniform in $s$, it remains only to integrate
\[
\int_0^\infty\sigma^{-n}(1+\log_+\sigma)^n\pi_\sigma(d\sigma)
=\frac{b_\sigma^{a_\sigma}}{\Gamma(a_\sigma)}\int_0^\infty\sigma^{-n-a_\sigma-1}(1+\log_+\sigma)^ne^{-b_\sigma/\sigma}\,d\sigma<\infty.
\]
Finiteness follows from the inverse-gamma exponential factor at zero and the polynomial decay at infinity. Thus the posterior normalizing constant is finite, and the same holds for the likelihood integral weighted by $|\beta_j|^r$. The normalizing constant is strictly positive because the likelihood is positive almost everywhere under the proper prior. This proves posterior propriety and $E_{\Pi_n}(|\beta_j|^r\mid D_n)<\infty$ for every $r>0$. No rank condition on the design matrix is used.
\end{proof}

\begin{proof}[Proof of Proposition 3]
Fix arbitrary $\tau\in(0,1)$, $\gamma>0$, $\sigma>0$, and $s\in(0,1)$. We construct a clean error density with $G_0(0)=\tau$ for which a positive displacement improves the expected working log likelihood. Let $a=2$, choose a nonnegative $C^\infty$ density $\psi_-$ supported in $(-3,-2)$, and choose a nonnegative $C^\infty$ density $\psi$ supported in $(-1,1)$. Let $\phi$ and $\Phi$ be the standard normal density and distribution function, and put $h_\tau(x)=\phi\{x+\Phi^{-1}(\tau)\}$. Then $\int_{-\infty}^0h_\tau(x)\,dx=\tau$ and $h_\tau(0)>0$. For fixed $\delta\in(0,1)$ and $0<h<1/2$, define
\begin{equation}
\label{eq:s1_counterexample_density}
g_h(x)=(1-\delta)\left\{\tau\psi_-(x)+(1-\tau)h^{-1}\psi\left(\frac{x-a}{h}\right)\right\}+\delta h_\tau(x).
\end{equation}
For every fixed $h$, $g_h$ is a bounded $C^\infty$ probability density with finite absolute moments of all positive orders. The negative bump lies entirely below zero and the positive bump entirely above zero, so $G_h(0)=(1-\delta)\tau+\delta\tau=\tau$, while $g_h(0)=\delta h_\tau(0)>0$.

Write $\ell_c(x)=\log\{\sigma^{-1}f_s((x-c)/\sigma)\}$ and $F_h(c)=\int\ell_c(x)g_h(x)\,dx$. The expectations are finite by Lemma~\ref{lem:s1_analytic_bounds}(iii). By the central asymptotic in Proposition 1 and the fact that $s>0$, there exist $h_0>0$ and $c_0>0$, depending only on the fixed parameters, such that for $0<h<h_0$ and $|t|\le1$, $f_s(ht/\sigma)\ge c_0\log(\sigma/h)$. Consequently, the positive-bump contribution to $F_h(a)$ satisfies
\begin{equation}
\label{eq:s1_bump_divergence}
\int_{\mathbb R}\psi(t)\ell_a(a+ht)\,dt
=\int_{\mathbb R}\psi(t)\log\{\sigma^{-1}f_s(ht/\sigma)\}\,dt
\ge\log\{c_0\log(\sigma/h)/\sigma\}\longrightarrow\infty.
\end{equation}
By contrast, $a+ht$ remains in a compact interval bounded away from zero when $|t|\le1$ and $h<1/2$. Continuity and positivity of $f_s$ on that interval imply that $\int\psi(t)\ell_0(a+ht)\,dt$ is uniformly bounded in $h$. The contributions of $\psi_-$ and $h_\tau$ to $F_h(a)-F_h(0)$ are finite constants independent of $h$, again by Lemma~\ref{lem:s1_analytic_bounds}(iii). Since $(1-\delta)(1-\tau)>0$, \eqref{eq:s1_bump_divergence} implies
\[
F_h(a)-F_h(0)\longrightarrow\infty\qquad(h\downarrow0).
\]
Taking any sufficiently small fixed $h$ in \eqref{eq:s1_counterexample_density} and setting $g_0=g_h$ gives $F_0(a,\sigma,s)>F_0(0,\sigma,s)$ at the nonzero displacement $a=2$, as required.
\end{proof}

\subsection{Population identification and curvature for a connected class of clean distributions}
\label{subsec:s1_nonvacuity}

We prove Proposition~4 of the main manuscript. Throughout this subsection, $\rho=\eta/(1-\eta)$, $\lambda=\log\sigma$, $\bm\xi=(\lambda,\operatorname{logit}s)^\top$, and $c_\tau=\tau(1-\tau)$. Let
\[
D_s(w)=(1-s)e^{-w}+sK_\gamma(w),\qquad K_\gamma(w)=\int_0^\infty v^{-1}e^{-w/v}H(v;\gamma)\,dv,\qquad w>0.
\]
For a clean density $g$, write
\begin{equation}
\label{eq:s1_nv_F}
F_{\rho,g,\tau}(a,\lambda,s)=\log c_\tau-\lambda+E_g\log D_s\!\left\{e^{-\lambda}\rho_\tau(U-a)\right\}+\rho\log s.
\end{equation}
The extension to $s=0$ is used only when $\rho=0$. To distinguish the interpolation weight from the working mixture weight $s$, write $\vartheta\in[0,1]$ for the former. Let $h_\kappa(z)=\exp(-|z|^\kappa)/\{2\Gamma(1+1/\kappa)\}$, let $H_\kappa$ denote its distribution function, and set $q_{\kappa,\tau}=H_\kappa^{-1}(\tau)$. Define
\begin{equation}
\label{eq:s1_nv_family}
\begin{aligned}
g^Q_{\tau,\kappa}(u)&=\frac{1}{2\Gamma(1+1/\kappa)}\begin{cases}\exp\{-\left(|u|/(2\tau)\right)^\kappa\},&u<0,\\\exp\{-\left(u/\{2(1-\tau)\}\right)^\kappa\},&u\ge0,\end{cases}\\
g^S_{\tau,\kappa}(u)&=h_\kappa(u+q_{\kappa,\tau}),\\
g_{\tau,\kappa,\vartheta,d}(u)&=d^{-1}\left\{(1-\vartheta)g^Q_{\tau,\kappa}(u/d)+\vartheta g^S_{\tau,\kappa}(u/d)\right\},\qquad \kappa\ge2,\quad d>0.
\end{aligned}
\end{equation}
All constants in the proofs below may depend on fixed $\tau$ and $\gamma$, but not on the shape or interpolation parameters in the specified admissible region. We first take $d=1$ and restore arbitrary $d$ at the end. For the clean criterion, put $\mathcal F_g=F_{0,g,\tau}$ and $m_g=E_g\rho_\tau(U)$, and use $u=(a,\lambda)^\top$ for the location--log-scale parameter. The symbol $\mathcal R_{K,\varepsilon}$ denotes $([2,K]\times[0,\varepsilon])\cup([K,\infty)\times[0,1])$.

The quantile-centering identities follow by changing variables separately on the negative and positive half-lines for $g^Q$ and by translating $h_\kappa$ by its $\tau$th quantile for $g^S$. Thus all distributions in \eqref{eq:s1_nv_family} are probability densities, have zero as their unique $\tau$th quantile, are bounded and $C^1$, and possess moments of every positive order. At $\kappa=2$, $g^Q_{\tau,2,d}$ agrees with the original reference density $g^\star_{\tau,D}(u)=4c_\tau\exp\{-2\rho_\tau(u)^2/D^2\}/(D\sqrt{2\pi})$ after setting $D=2\sqrt2\,dc_\tau$. For $\tau=1/2$, $g^Q_{\tau,\kappa}=g^S_{\tau,\kappa}$ for every $\kappa\ge2$. As $\kappa\to\infty$, $h_\kappa\to\mathbf1_{[-1,1]}/2$ in $L^1$ and $q_{\kappa,\tau}\to2\tau-1$. Both families therefore converge to the same quantile-centered uniform density $g_{\tau,\infty}(u)=\frac12\mathbf1_{[-2\tau,\,2(1-\tau)]}(u)$. For each fixed $b>0$, the generalized Gaussian tails admit an integrable common envelope after multiplication by $e^{b|u|}$ for all $\kappa\ge2$, since the shape exponents are at least two and the relevant centering and scale constants are uniformly bounded. Dominated convergence and translation continuity thus give
\begin{equation}
\label{eq:s1_nv_weighted_limit}
\sup_{\vartheta\in[0,1]}\int_{\mathbb R}e^{b|u|}\left|g_{\tau,\kappa,\vartheta,1}(u)-g_{\tau,\infty}(u)\right|\,du\longrightarrow0.
\end{equation}

\begin{lemma}[Quantile-adapted global identification and curvature]
\label{lem:s1_nv_quantile}
Fix $\tau\in(0,1)$ and $\gamma>0$. For each $\kappa\ge2$, put $m_\kappa=\Gamma(2/\kappa)/\Gamma(1/\kappa)$, $m^Q_{\tau,\kappa}=2c_\tau m_\kappa$, and $u^Q_\kappa=(0,\log m^Q_{\tau,\kappa})^\top$. The clean criterion $\mathcal F_{g^Q_{\tau,\kappa}}$ has the unique, globally separated maximizer $(u^Q_\kappa,0)$ on $\mathbb R^2\times[0,1]$. For $X_\kappa=|Z_\kappa|$, $Z_\kappa\sim h_\kappa$, define $J_\kappa(v)=v^{-1}E\exp\{(1-v^{-1})X_\kappa/m_\kappa\}$ and $c_{\kappa,\gamma}=1-\int_0^\infty J_\kappa(v)H(v;\gamma)\,dv$. Then
\begin{equation}
\label{eq:s1_nv_boundary_margin}
J_\kappa(v)<1\ (v\ne1),\qquad \inf_{\kappa\ge2}c_{\kappa,\gamma}>0,\qquad \partial_s^+\mathcal F_{g^Q_{\tau,\kappa}}(u^Q_\kappa,0)=-c_{\kappa,\gamma}.
\end{equation}
Moreover,
\begin{equation}
\label{eq:s1_nv_quantile_curvature}
-\nabla^2_{uu}\mathcal F_{g^Q_{\tau,\kappa}}(u^Q_\kappa,0)=\operatorname{diag}\!\left\{\frac{1}{2c_\tau\Gamma(1+2/\kappa)},1\right\}\succ0.
\end{equation}
For every $\rho\ge0$, fixed $\lambda\in\mathbb R$, and $s\in[0,1]$ when $\rho=0$ or $s\in(0,1]$ when $\rho>0$, the location map $a\mapsto F_{\rho,g^Q_{\tau,\kappa},\tau}(a,\lambda,s)$ is uniquely maximized at $a=0$.
\end{lemma}

\begin{proof}
Changing variables on the two half-lines gives $\rho_\tau(U)\overset{\mathcal D}=2c_\tau X_\kappa$ for $U\sim g^Q_{\tau,\kappa}$. The variable $X_\kappa$ has survival function $S_\kappa$ with a strictly increasing hazard rate, because the reciprocal hazard is proportional to $\int_0^\infty\exp\{-(x+t)^\kappa+x^\kappa\}\,dt$, which is strictly decreasing in $x\ge0$. Hence $S_\kappa(x+t)<S_\kappa(x)S_\kappa(t)$ for $x,t>0$. Setting $R_\kappa(x)=\int_x^\infty S_\kappa(y)\,dy$, we have $R_\kappa(x)<m_\kappa S_\kappa(x)$ and $R_\kappa'(x)<-R_\kappa(x)/m_\kappa$, whence $R_\kappa(x)<m_\kappa e^{-x/m_\kappa}$. The stop-loss characterization proves strict convex-order domination of $X_\kappa/m_\kappa$ by $E\sim\operatorname{Exp}(1)$. Applying this domination to $x\mapsto\exp\{(1-v^{-1})x\}$, for $v\ne1$, yields
\begin{equation}
\label{eq:s1_nv_laplace_ratio}
J_\kappa(v)<v^{-1}E\exp\{(1-v^{-1})E\}=1,\qquad v\ne1.
\end{equation}
For $p_b(z)=c_\tau b^{-1}\exp\{-\rho_\tau(z)/b\}$, $m=m^Q_{\tau,\kappa}$, and any probability measure $\mu$ on $(0,\infty)$, let $p_\mu(z)=\int p_{mv}(z)\,\mu(dv)$. Since $\rho_\tau(U)/m\overset{\mathcal D}=X_\kappa/m_\kappa$, Tonelli's theorem and $\log x\le x-1$ yield
\begin{equation}
\label{eq:s1_nv_projection}
E\log\frac{p_\mu(U)}{p_m(U)}\le E\left\{\frac{p_\mu(U)}{p_m(U)}-1\right\}=\int J_\kappa(v)\,\mu(dv)-1\le0,
\end{equation}
with equality only for $\mu=\delta_1$. The AL--LPAL likelihood corresponds to a single scale when $s=0$, and its nondegenerate LPAL component excludes equality when $s>0$. Its $s$-boundary score at the AL optimizer is $-c_{\kappa,\gamma}$, by Tonelli's theorem. Continuity in $\kappa$ and convergence of $X_\kappa/m_\kappa$ to $\operatorname{Unif}(0,2)$ imply continuity of $c_{\kappa,\gamma}$ on $[2,\infty]$, where \eqref{eq:s1_nv_laplace_ratio} remains strict away from $v=1$. Dominated convergence applies because $0\le J_\kappa(v)\le1$. Thus the continuous strictly positive function $c_{\kappa,\gamma}$ has a positive minimum on this compactified interval, proving \eqref{eq:s1_nv_boundary_margin}.

For location identification, put $J_a(t)=P\{\rho_\tau(U-a)\le t\}=G(a+t/\tau)-G\{a-t/(1-\tau)\}$. The two-piece density $g^Q_{\tau,\kappa}(u)$ is a strictly decreasing function of $\rho_\tau(u)$. Comparing the check losses of the two endpoints gives $\partial_aJ_a(t)>0$ for $a<0$ and $\partial_aJ_a(t)<0$ for $a>0$. Thus $J_a(t)<J_0(t)$ for every $a\ne0$ and $t>0$. Since $w\mapsto\log D_s(w)$ is strictly decreasing, the expected log likelihood, with or without the additive $\rho\log s$ term, is uniquely maximized in $a$ at zero for every fixed $(\lambda,s)$. Combined with \eqref{eq:s1_nv_projection}, this establishes the claimed unique clean maximizer $(0,\log m,0)$. To verify global separation, choose two fixed disjoint bounded intervals with positive clean probability. If $(a,\lambda)$ leaves every compact subset, then along each such parameter sequence at least one of these intervals has fitted probability tending to zero, uniformly in $s\in[0,1]$: this follows from scale concentration when $\sigma\downarrow0$, vanishing interval probability when $\sigma\uparrow\infty$, and translation tightness when $|a|\uparrow\infty$ with $\sigma$ bounded. The binary Kullback--Leibler inequality, together with the finite clean entropy, forces the criterion to $-\infty$. Compactness and strict uniqueness give global separation. Finally, at $s=0$, $\mathcal F_g(a,\lambda,0)=\log c_\tau-\lambda-e^{-\lambda}E_g\rho_\tau(U-a)$. Its location--scale Hessian at $(a,\lambda)=(0,\log m)$ is $-\operatorname{diag}\{g(0)/m,1\}$. The identities $g^Q_{\tau,\kappa}(0)=\{2\Gamma(1+1/\kappa)\}^{-1}$ and $m=2c_\tau\Gamma(2/\kappa)/\Gamma(1/\kappa)$ give \eqref{eq:s1_nv_quantile_curvature}.
\end{proof}

\begin{lemma}[Large-shape interpolation and its uniform limit]
\label{lem:s1_nv_large}
For fixed $\tau\in(0,1)$ and $\gamma>0$, there exist $K<\infty$ and constants $r,s_0,c,h,L>0$ such that, for every $\kappa\ge K$ and $\vartheta\in[0,1]$, the clean criterion $\mathcal F_{g_{\tau,\kappa,\vartheta,1}}$ has the unique, globally separated maximizer $(u^0_{\kappa,\vartheta},0)$, where $u^0_{\kappa,\vartheta}=(0,\log m_{\kappa,\vartheta})^\top$ and $m_{\kappa,\vartheta}=E_{g_{\tau,\kappa,\vartheta,1}}\rho_\tau(U)$. On $\|u-u^0_{\kappa,\vartheta}\|\le r$ and $0\le s\le s_0$, the first two derivatives of the clean criterion exist, with one-sided derivatives in $s$ at zero, and
\begin{equation}
\label{eq:s1_nv_uniform_local}
\partial_s^+\mathcal F_g(u,0)\le-c,\qquad -\nabla^2_{uu}\mathcal F_g(u,s)\succeq h\bm I_2,\qquad \|\nabla^2_{us}\mathcal F_g(u,s)\|+|\partial^2_{ss}\mathcal F_g(u,s)|\le L.
\end{equation}
For each fixed $\delta>0$, the global objective gap outside the $\delta$-neighborhood of $(u^0_{\kappa,\vartheta},0)$ is positive uniformly over $\kappa\ge K$ and $\vartheta\in[0,1]$.
\end{lemma}

\begin{proof}
First consider $U_\infty\sim g_{\tau,\infty}$. Then $E\rho_\tau(U_\infty)=c_\tau$ and $\rho_\tau(U_\infty)/c_\tau\sim\operatorname{Unif}(0,2)$. For $W\sim\operatorname{Unif}(0,2)$ and $E\sim\operatorname{Exp}(1)$, $E(W-t)_+=(2-t)^2/4<e^{-t}=E(E-t)_+$ for $0<t<2$; the inequality is also strict for $t\ge2$. Hence $W<_{\mathrm{cx}}E$ and the scale-mixture comparison in \eqref{eq:s1_nv_projection} proves unique scale identification and a negative boundary score. For location, $P\{\rho_\tau(U_\infty-a)\le t\}$ is one half the length of the intersection of $[-2\tau,2(1-\tau)]$ with $[a-t/(1-\tau),a+t/\tau]$. This length is no larger than for $a=0$ at every $t>0$, and is strictly smaller on a nonempty interval of $t$ whenever $a\ne0$. Strict monotonicity of the working log-kernel therefore gives unique location identification. The separated-interval KL argument of Lemma~\ref{lem:s1_nv_quantile} gives global separation. Finally, $g_{\tau,\infty}(0)=1/2$, so the AL location--scale Hessian is the claimed diagonal matrix.

We now transfer these properties to the large-shape interpolation family. By \eqref{eq:s1_nv_weighted_limit}, both components converge in exponentially weighted $L^1$ to $g_{\tau,\infty}$, uniformly in $\vartheta$. The densities are uniformly bounded and assign positive probability, bounded away from zero, to two fixed disjoint intervals in the interior of the limiting support. The separated-interval binary KL argument gives common coercivity as $|a|+|\lambda|\to\infty$ uniformly over $s\in[0,1]$. To establish convergence of the criterion on compact parameter sets, split the integral into $|U-a|\le\delta$ and its complement. Near zero the singularity $\log D_s\{e^{-\lambda}\rho_\tau(U-a)\}$ is at most of order $\log\{1+|\log|U-a||\}$ uniformly for $s\in[0,1]$, whose integral over the small interval vanishes uniformly by boundedness of the densities. On the complement, the kernels have a common exponential envelope and the weighted convergence \eqref{eq:s1_nv_weighted_limit} applies. Consequently the criterion converges to the uniform-limit criterion uniformly over $\vartheta$ and over any compact subset of $(a,\lambda,s)$. The preceding uniform-limit identification and uniform coercivity then localize all global maximizers to an arbitrarily small common neighborhood of $(0,\log c_\tau,0)$.

It remains to justify derivatives near $s=0$ without differentiating the central singularity directly in location. Put $A_\gamma(w)=e^wK_\gamma(w)-1$ and decompose $\mathcal F_g(u,s)=\log c_\tau-\lambda-e^{-\lambda}E_g\rho_\tau(U-a)+\Delta_g(u,s)$, where $\Delta_g(u,s)=E_g\log\{1+sA_\gamma(e^{-\lambda}\rho_\tau(U-a))\}$. The density $H(\cdot;\gamma)$ is bounded and positive near zero, and its logarithmic upper tail gives the following kernel bounds: for $0<w\le1$, the estimates $|A_\gamma(w)|\le C(1+|\log w|)$ and $|wA_\gamma'(w)|+|w^2A_\gamma''(w)|\le C$ hold. Consequently, for compact $\lambda$ sets and $0\le s\le s_0<1$, the derivatives in $(\lambda,s)$ of the correction kernel through order two are bounded near $v=0$ by $C\{1+|\log|v||^2\}$. Away from zero, the corresponding kernel and its first two $v$ derivatives admit $C\exp(C|v|)$ envelopes, using $K_\gamma^{(j)}(w)=O(w^{-j-1})$ for large $w$ and $1+sA_\gamma(w)\ge1-s_0$.

Choose $r>0$ small enough that $[-4r,4r]$ lies strictly inside the limiting support, and take a smooth cutoff $\chi$ equal to one on $[-r,r]$ and zero outside $[-2r,2r]$. With $v=U-a$, decompose $\Delta_g=\int g(a+v)\chi(v)k_{\lambda,s}(v)\,dv+\int g(a+v)\{1-\chi(v)\}k_{\lambda,s}(v)\,dv$, where $k_{\lambda,s}(v)=\log\{1+sA_\gamma(e^{-\lambda}\rho_\tau(v))\}$. In the first integral, transfer each location derivative to $g'(a+v)$ or $g''(a+v)$. On the fixed interior interval, the functions $g^Q_{\tau,\kappa}$, $g^S_{\tau,\kappa}$ and their first two derivatives converge uniformly to $1/2,0,0$, respectively, because the rescaled arguments are uniformly separated from the limiting support endpoints. For the second integral, integrate by parts to move up to two location derivatives onto $\{1-\chi(v)\}k_{\lambda,s}(v)$; its derivative envelopes are integrable under the common exponential moment bounds. Dominated convergence and \eqref{eq:s1_nv_weighted_limit} therefore yield uniform convergence of all derivatives through total order two near the limiting maximizer, including one-sided derivatives at $s=0$. The AL part is handled directly using $\partial_{aa}E\rho_\tau(U-a)=g(a)$.

The strictly negative boundary score and positive location--scale curvature from the preceding uniform-limit identification persist uniformly, giving \eqref{eq:s1_nv_uniform_local}. Since the criterion is concave in $s$, its boundary-score inequality implies $\mathcal F_g(u,s)\le\mathcal F_g(u,0)-cs$ for the common location--scale neighborhood and every $s\in[0,1]$. Global localization therefore forces all maxima to $s=0$. The AL submodel has the unique maximizer $a=0,\sigma=m_{\kappa,\vartheta}$ because the clean $\tau$th quantile is uniquely zero. The same compact-uniform convergence, including the compactified shape point $\kappa=\infty$, and uniform coercivity establish a common strict global objective gap.
\end{proof}

\begin{lemma}[Connected clean-family identification]
\label{lem:s1_nv_bounded}
Let $K$ be chosen as in Lemma~\ref{lem:s1_nv_large}. There exists $\varepsilon\in(0,1]$ such that, for every $(\kappa,\vartheta)$ in the connected region $\mathcal R_{K,\varepsilon}$, the clean criterion has the unique, globally separated maximizer $(u^0_{\kappa,\vartheta},0)$, a common strictly negative mixture boundary-score margin, and all local derivative and curvature bounds in \eqref{eq:s1_nv_uniform_local}, after possibly reducing the common constants. The global objective gap outside each fixed neighborhood of $(u^0_{\kappa,\vartheta},0)$ can be chosen uniformly over $\mathcal R_{K,\varepsilon}$.
\end{lemma}

\begin{proof}
The exact identity $\mathcal F_{g_{\tau,\kappa,\vartheta,1}}=(1-\vartheta)\mathcal F_{g^Q_{\tau,\kappa}}+\vartheta\mathcal F_{g^S_{\tau,\kappa}}$ follows from linearity of expectation in the clean density. On the compact shape interval $[2,K]$, both density families and their weak second derivatives admit common Gaussian-weighted bounds. Moving location derivatives onto the clean density gives a $C^2$ perturbation of order $O(\vartheta)$ on a common neighborhood of the $Q$-family boundary maximizers. Lemma~\ref{lem:s1_nv_quantile}, compactness of $[2,K]$, and common location--scale coercivity give a uniform global objective gap, a strictly negative mixture boundary-score margin, and uniformly positive location--scale curvature for $\vartheta=0$. These conclusions persist for $0\le\vartheta\le\varepsilon$ with a common sufficiently small $\varepsilon>0$. Moreover, $m_{\kappa,\vartheta}=(1-\vartheta)m^Q_{\tau,\kappa}+\vartheta m^S_{\tau,\kappa}$, so $\sup_{\kappa\in[2,K]}|\log m_{\kappa,\vartheta}-\log m^Q_{\tau,\kappa}|=O(\vartheta)$. The negative boundary score and concavity in $s$ force any global maximum into $s=0$ after uniform localization. AL quantile identification then gives the unique maximizer $(0,\log m_{\kappa,\vartheta},0)$. Continuity on the compact indexing set gives the same uniform global objective gap and local derivative bounds as in \eqref{eq:s1_nv_uniform_local}.
\end{proof}

\begin{lemma}[Uniform interior identification under small contamination]
\label{lem:s1_nv_interior}
There exists $\rho_0>0$ such that, for all $(\kappa,\vartheta)\in\mathcal R_{K,\varepsilon}$ and $0<\rho<\rho_0$, the adjusted criterion \eqref{eq:s1_nv_F} has a unique, globally separated maximizer $(u_{\kappa,\vartheta,\rho}^\dagger,s_{\kappa,\vartheta,\rho}^\dagger)$ in $\mathbb R^2\times(0,1)$. Put $u^0_{\kappa,\vartheta}=(0,\log m_{\kappa,\vartheta})$, $c_{\kappa,\vartheta}=-\partial_s^+\mathcal F_{g_{\tau,\kappa,\vartheta,1}}(u^0_{\kappa,\vartheta},0)>0$, $H_{\kappa,\vartheta}=\nabla^2_{uu}\mathcal F_{g_{\tau,\kappa,\vartheta,1}}(u^0_{\kappa,\vartheta},0)$, and $b_{\kappa,\vartheta}=\nabla^2_{us}\mathcal F_{g_{\tau,\kappa,\vartheta,1}}(u^0_{\kappa,\vartheta},0)$. Then
\begin{equation}
\label{eq:s1_nv_branch}
\begin{aligned}
s_{\kappa,\vartheta,\rho}^\dagger&=\frac{\rho}{c_{\kappa,\vartheta}}+O(\rho^2),\\
u_{\kappa,\vartheta,\rho}^\dagger&=u^0_{\kappa,\vartheta}-H_{\kappa,\vartheta}^{-1}b_{\kappa,\vartheta}\frac{\rho}{c_{\kappa,\vartheta}}+o(\rho),
\end{aligned}
\end{equation}
where the remainder bounds are uniform in $(\kappa,\vartheta)\in\mathcal R_{K,\varepsilon}$. At the maximizer, the negative Hessian is positive definite in the original $(a,\lambda,\operatorname{logit}s)$ coordinates.
\end{lemma}

\begin{proof}
Lemmas~\ref{lem:s1_nv_large} and~\ref{lem:s1_nv_bounded} give common constants $r,s_0,c,h,L>0$, a global objective gap outside the $r$-neighborhood, and the derivative bounds in \eqref{eq:s1_nv_uniform_local}. At the comparison point $(u^0_{\kappa,\vartheta},s=\rho)$, the adjusted criterion exceeds the clean maximum by at least $-C\rho+\rho\log\rho$. Since $\rho\log s\le0$ for $0<s\le1$, the uniform global gap confines every adjusted maximizer to $\|u-u^0_{\kappa,\vartheta}\|<r$ and $s<s_0$ when $\rho$ is small. Uniform coercivity in $(a,\lambda)$ and $\rho\log s\to-\infty$ as $s\downarrow0$ give existence. Moreover, concavity in $s$ and the negative boundary-score bound imply $\mathcal F_g(u,s)\le\mathcal F_g(u^0,0)-cs$ on this neighborhood. Comparison with the candidate gives $\log(s/\rho)-c(s/\rho)\ge-C$, so there are common constants $0<c_1<c_2<\infty$ with $c_1\rho\le s\le c_2\rho$ for every maximizer. The location--scale first-order condition, strong concavity, and bounded mixed derivatives yield $\|u-u^0\|\le(L/h)s=O(\rho)$. The mixture first-order condition is $\partial_s\mathcal F_g(u,s)+\rho/s=0$, which gives $\rho/s=c_{\kappa,\vartheta}+O(\rho)$ and hence the first expansion in \eqref{eq:s1_nv_branch}. Taylor expansion of the location--scale score gives the second expansion; uniform continuity of the Hessians on the compactified indexing region makes its remainder uniform.

To establish uniqueness, consider the convex region $\|u-u^0\|\le r$ and $c_1\rho\le s\le c_2\rho$, which contains every global maximizer. There $-\nabla^2_{uu}\mathcal F_g\succeq h\bm I_2$, $\|\nabla^2_{us}\mathcal F_g\|\le L$, and $-\partial^2_{ss}F_{\rho,g,\tau}=-\partial^2_{ss}\mathcal F_g+\rho/s^2\ge1/(c_2^2\rho)$. Its negative Hessian has Schur complement at least $1/(c_2^2\rho)-L^2/h>0$ for sufficiently small $\rho$. Thus the criterion is strictly concave throughout this region and its global maximizer is unique. The same coercivity and compactness yield global separation. At a stationary point, the change of variable $s=\exp(\zeta)/(1+\exp\zeta)$ gives a Hessian congruence with Jacobian $\operatorname{diag}\{1,1,s(1-s)\}$, which preserves positive definiteness for each fixed $\rho>0$. We do not claim an eigenvalue lower bound uniform as $\rho\downarrow0$ in $\zeta$ coordinates.
\end{proof}

\begin{lemma}[Stability under smooth density perturbations]
\label{lem:s1_nv_stability}
Fix $\tau\in(0,1)$, $\gamma>0$, and a compact interval $I_\eta\subset(0,\eta_0)$, where $\eta_0=\rho_0/(1+\rho_0)$. Let $g^\star=g^Q_{\tau,2,1}$ and choose a sufficiently small $\omega>0$. For $h\in C^1(\mathbb R)$ with locally absolutely continuous first derivative and weak second derivative, set
\begin{equation}
\label{eq:s1_nv_norm}
\|h\|_{\mathcal X_\omega}=\sup_{u\in\mathbb R}e^{\omega u^2}\{|h(u)|+|h'(u)|\}+\left\{\int_{\mathbb R}e^{2\omega u^2}|h''(u)|^2\,du\right\}^{1/2}.
\end{equation}
There are $\varepsilon_*>0$ and $C<\infty$ such that, for every $\eta\in I_\eta$, any probability density $g$ with $G(0)=\tau$ and $\|g-g^\star\|_{\mathcal X_\omega}<\varepsilon_*$ satisfies (A1), (A4), and (A5). Its unique interior maximizer $\bm\vartheta_g^\dagger=(a_g^\dagger,\lambda_g^\dagger,\operatorname{logit}s_g^\dagger)^\top$ obeys
\begin{equation}
\label{eq:s1_nv_lipschitz}
\|\bm\vartheta_g^\dagger-\bm\vartheta_\star^\dagger\|\le C\|g-g^\star\|_{\mathcal X_\omega}.
\end{equation}
\end{lemma}

\begin{proof}
The reference density is bounded, positive on compact intervals, and has Gaussian tails; its weak second derivative belongs to the weighted $L^2$ space in \eqref{eq:s1_nv_norm} when $\omega>0$ is sufficiently small. For a perturbation $h=g-g^\star$, change variables to write $F_{\rho,g,\tau}(a,\lambda,s)=\int g(a+v)\ell_{\lambda,s}(v)\,dv+\rho\log s$, where $\ell_{\lambda,s}(v)=\log\{c_\tau e^{-\lambda}D_s(e^{-\lambda}\rho_\tau(v))\}$. Location derivatives act on $g$ rather than on the singular kernel. Near $v=0$, the kernel and its nuisance derivatives through order two have integrable log-log or logarithmic envelopes; in the tails they grow at most polynomially for each compact interior parameter set. By Cauchy--Schwarz and the weighted norm, for example,
\[
\left|\int h''(a+v)\ell_{\lambda,s}(v)\,dv\right|\le\|h\|_{\mathcal X_\omega}\left\{\int e^{-2\omega(a+v)^2}|\ell_{\lambda,s}(v)|^2\,dv\right\}^{1/2}.
\]
The same bounds for the other derivatives give uniform $C^2$ perturbation estimates on compact parameter sets. As in the earlier Gaussian reference argument, the clean mass on separated intervals and the binary KL inequality provide uniform coercivity for densities in a sufficiently small $\mathcal X_\omega$ neighborhood. Because $I_\eta$ is bounded away from zero, $\rho\log s$ also excludes $s\downarrow0$ uniformly. The reference optimizers form a compact interior family and possess a common global objective gap and a positive-definite negative Hessian by Lemma~\ref{lem:s1_nv_interior}. Small $C^2$ perturbations preserve the gap, strict local concavity, and interiority. The mean-value theorem applied to the population score equations yields \eqref{eq:s1_nv_lipschitz} and completes the proof.
\end{proof}

\begin{proof}[Proof of Proposition~4]
By Lemmas~\ref{lem:s1_nv_quantile}--\ref{lem:s1_nv_bounded}, the clean population criterion has a unique, globally separated boundary maximizer, a strictly negative mixture boundary score, and positive location--scale curvature uniformly over $\mathcal R_{K,\varepsilon}$. Lemma~\ref{lem:s1_nv_interior} supplies $\rho_0>0$ such that these conclusions extend to the adjusted criterion for every $0<\rho<\rho_0$. Put $\eta_0=\rho_0/(1+\rho_0)$. The normalization, quantile-centering, differentiability, positivity, and moment facts in the introductory paragraph verify (A1); global separation and positivity of the density near the pseudo-true location verify (A4); and the continuous, positive-definite negative Hessian verifies (A5). For $d>0$, changing variables gives
\[
F_{\rho,g_{\tau,\kappa,\vartheta,d},\tau}(a,\lambda,s)=-\log d+F_{\rho,g_{\tau,\kappa,\vartheta,1},\tau}(a/d,\lambda-\log d,s).
\]
The invertible change of coordinates preserves global separation, uniqueness, and positive definiteness at a stationary point, so the same $\eta_0$ works for every $d>0$. Lemma~\ref{lem:s1_nv_stability} gives an infinite-dimensional family of admissible perturbations around the included Gaussian reference point. It remains to exhibit smooth densities with $a^\dagger\ne0$.

Fix a positive contamination proportion $\eta\in(0,\eta_0)$ and set $g^\star=g^Q_{\tau,2,1}$. By Lemma~\ref{lem:s1_nv_quantile}, the reference pseudo-true location is zero for every nuisance pair. Let $(0,\lambda_\rho,s_\rho)$ be its adjusted maximizer, and define the working location score $\psi(u)=\partial_a\log\{c_\tau e^{-\lambda_\rho}D_{s_\rho}(e^{-\lambda_\rho}\rho_\tau(u-a))\}|_{a=0}$. For $u>0$,
\[
\psi(u)=-\tau e^{-\lambda_\rho}\frac{D'_{s_\rho}(\tau e^{-\lambda_\rho}u)}{D_{s_\rho}(\tau e^{-\lambda_\rho}u)}.
\]
Since $0<s_\rho<1$, $D_{s_\rho}(w)$ is the Laplace transform of a nondegenerate positive rate measure; although this measure need not have finite total mass, its exponential tilt is finite for each $w>0$. Its strictly positive tilted rate variance implies $(\log D_{s_\rho})''(w)>0$, so $\psi$ is nonconstant on every positive interval. Choose $h\in C_c^\infty((1,2))$ such that $\int g^\star h=0$ but $\int g^\star h\psi\ne0$; a suitable linear combination of two disjoint smooth bumps has these properties. For sufficiently small $t$, $g_t=g^\star(1+th)$ is a normalized, positive density with unique $\tau$th quantile zero and belongs to the stable neighborhood of $g^\star$. Because $g^\star$ is strongly log-concave and $h$ is compactly supported away from the origin, $g_t$ remains strictly log-concave for sufficiently small $|t|$.

The identity $F_{\rho,g^\star,\tau,a}(0,\lambda,s)=0$ for every interior nuisance pair implies that the location--nuisance Hessian blocks vanish at the reference maximizer. Moreover, write $D=2\sqrt2\,c_\tau$, so $g^\star=g^\star_{\tau,D}$ in the original Gaussian parameterization, and set $M_s(w)=-wD_s'(w)/D_s(w)$. The same differentiation of $J_a(t)$ as in the original Gaussian proof, followed by integration by parts, yields $-F_{\rho,g^\star,\tau,aa}(0,\lambda,s)=4c_\tau D^{-2}E\{M_s(\rho_\tau(U)e^{-\lambda})\}$. The log-scale score is $-1+E\{M_s(\rho_\tau(U)e^{-\lambda})\}$, which vanishes at the interior maximizer. Since the location score vanishes at $a=0$ for every nuisance pair, the cross derivatives are zero, and we obtain
\begin{equation}
\label{eq:s1_nv_curvature}
-F_{\rho,g^\star,\tau,aa}(0,\lambda_\rho,s_\rho)=\frac{4c_\tau}{D^2}>0,\qquad F_{\rho,g^\star,\tau,a\bm\xi}(0,\lambda_\rho,s_\rho)=\bm0.
\end{equation}
Differentiating the full population first-order conditions along $t\mapsto g_t$ therefore gives
\[
\left.\frac{d a_t^\dagger}{dt}\right|_{t=0}=-\frac{\int_{\mathbb R}g^\star(u)h(u)\psi(u)\,du}{F_{\rho,g^\star,\tau,aa}(0,\lambda_\rho,s_\rho)}\ne0.
\]
Thus $a_t^\dagger\ne0$ for some arbitrarily small $t\ne0$. To obtain a globally $C^\infty$ density, mollify the strongly convex negative log-density of $g_t$, normalize the result, and translate it by its own $\tau$th quantile. The mollified densities converge to $g_t$ in the weighted norm \eqref{eq:s1_nv_norm} because the negative log-density is quadratic outside a compact interval, its first derivative is continuous, and its weak second derivative is bounded. Lemma~\ref{lem:s1_nv_stability} then preserves $a^\dagger\ne0$ for sufficiently small smoothing bandwidth. Finally, the vector space of $h\in C_c^\infty((1,2))$ satisfying $\int g^\star h=0$ is infinite-dimensional, giving the asserted infinite-dimensional collection of smooth strictly log-concave examples and completing the proof.
\end{proof}
\section{Preliminary Results for Posterior Localization}
\label{sec:supp_global_localization}

We establish the global likelihood and posterior bounds needed for the nonregular contraction arguments. Throughout this section, $K_n$ and $L_n$ denote the clean and contaminated index sets, with $k=|K_n|$, $m=|L_n|$, $k+m=n$, $m/n\to\eta\in[0,1)$, and $\rho_n=m/k\to\rho=\eta/(1-\eta)$. The clean observations are $Y_i=\alpha_0+\bm z_i^\top\bm b_0+\epsilon_i$, $i\in K_n$, with independent errors $\epsilon_i\sim g_0$; the contaminated responses are $Y_i=a_i+c_i\omega_n$, $i\in L_n$. Write $\bm x_i=(1,\bm z_i^\top)^\top$, $\bm\beta=(\alpha,\bm b^\top)^\top$, $\bm\xi=(\log\sigma,\operatorname{logit}s)^\top$, and $\theta=(\bm\beta^\top,\sigma,s)^\top$; when studying population curvature, we identify $\theta$ with its one-to-one reparameterization $(\bm\beta^\top,\bm\xi^\top)^\top$. The working likelihood is $q_\theta(y\mid\bm x)=\sigma^{-1}f_s\{(y-\bm x^\top\bm\beta)/\sigma\}$, with $f_s=(1-s)f_0+sf_1$ as defined in Section S1. The regression coefficients, scale, and mixture weight have independent proper Gaussian, inverse-gamma, and beta priors. We write $\pi(d\theta)$ for their joint prior, $L_{K,n}(\theta)=\prod_{i\in K_n}q_\theta(Y_i\mid\bm x_i)$ for the clean likelihood, and $P_0$ for the sampling law of the clean observations. The constants below can depend on the fixed model and prior parameters but not on $n$ or the observed clean responses.

For $a\in\mathbb R$ and $\bm\xi=(\log\sigma,\operatorname{logit}s)^\top$, define
\[
F_0(a,\bm\xi)=E_{g_0}\log\left\{\sigma^{-1}f_s\left(\frac{\epsilon-a}{\sigma}\right)\right\},\qquad F_\rho(a,\bm\xi)=F_0(a,\bm\xi)+\rho\log s.
\]
Set $\bm\beta_0=(\alpha_0,\bm b_0^\top)^\top$, $\bm\beta^\dagger=(\alpha_0+a^\dagger,\bm b_0^\top)^\top$, and let $\theta^\dagger$ denote the parameter with regression component $\bm\beta^\dagger$ and nuisance component $\bm\xi^\dagger$. The regression population criterion is
\[
\mathcal M_\rho(\bm\beta,\bm\xi)=\int F_\rho\{\alpha-\alpha_0+\bm z^\top(\bm b-\bm b_0),\bm\xi\}\,\nu(d\bm z),\qquad M^\star=\mathcal M_\rho(\bm\theta^\dagger).
\]
All results below use Assumptions (A1)--(A5), as numbered in the main manuscript, without introducing additional conditions on the clean fraction or the rate of convergence of $\rho_n$. The proof also uses the density bounds established in Lemma~\ref{lem:s1_analytic_bounds}. The argument separates a $s^m$-tilted clean posterior from the full posterior without assuming that these two distributions converge in total variation, a distinction relevant to log-regularly varying regression likelihoods \citep{Desgagne2015,Hamura2022}.

\subsection{Prior-integrated likelihood bounds}

\begin{lemma}[Prior-integrated moments and sieve-tail control]
\label{lem:t1_integrated_envelope}

Let $r\in\{1,2\}$ and let $y_1,\ldots,y_j$ and $x_1,\ldots,x_j$ be arbitrary fixed responses and bounded covariates with intercept coordinate one. For the independent Gaussian, inverse-gamma, and beta priors, a constant $C_r$ independent of these observations satisfies
\begin{equation}
\label{eq:t1_integrated_envelope}
\int\prod_{i=1}^j q_\theta(y_i\mid x_i)^r
\,\pi(d\theta)
\le \exp\{C_rj\log(j+1)\}.
\end{equation}
For sufficiently large fixed $A>2$, there is $c_A>0$ such that the same integral with $r=1$ restricted to
\[
\mathcal T_n(A)^c,
\quad
\mathcal T_n(A)=\{\|\beta\|\le n^A,
\ n^{-A}\le\sigma\le n^A,
\ s\ge e^{-n^A}\},
\]
is at most $\exp\{-c_A n\log n\}$ when $j\le n$ and $j\asymp n$, after increasing $A$ if necessary.
\end{lemma}
\begin{proof}

By \eqref{eq:s1_density_envelope},
\[
q_\theta(y_i\mid x_i)
\le \frac C\sigma\left\{1+\log_+\sigma+
\log_+\frac1{|y_i-\alpha-z_i^\top b|}\right\}.
\]
Conditionally on $b$, the Gaussian conditional law of $\alpha$ has a fixed positive variance, and its density is uniformly bounded, independently of its conditional mean. Consequently, for every $v\ge1$ and $d\in\mathbb R$,
\[
E\left[\left\{1+\log_+\frac1{|\alpha-d|}\right\}^v
\middle|b\right]\le C^v\Gamma(v+1).
\]
Generalized H\"older inequality with exponent $j$, followed by integration in $b$, gives
\[
\int\prod_{i=1}^j q_\theta(y_i\mid x_i)^r
\,\pi_\beta(d\beta)
\le\left\{\frac{C(1+\log_+\sigma)}\sigma\right\}^{rj}
C^{rj}\Gamma(rj+1).
\]
Integrating the inverse-gamma prior, splitting at $\sigma=1$, and using Stirling's bound proves \eqref{eq:t1_integrated_envelope}. For the tails of $\beta$ and $s$, use Cauchy--Schwarz and the $r=2$ estimate. The relevant Gaussian and beta probabilities are bounded by $e^{-c n^{2A}}$ and $e^{-a_s n^A}$, respectively. On $\sigma<n^{-A}$, the factor $e^{-b_\sigma/\sigma}$ dominates the integrated polynomial inverse scale powers, giving $e^{-c n^A+C n\log n}$. On $\sigma>n^A$ the Gaussian-integrated likelihood is bounded by $C^n n!\sigma^{-n}(1+\log\sigma)^n$; monotonicity above $n^A$ yields $e^{-A n\log n+O(n\log\log n)+O(n\log n)}$. Choosing $A$ sufficiently large completes the claim.

\end{proof}

\subsection{Population identification and global likelihood comparison}

\begin{lemma}[Population identification, curvature, and fractional coercivity]
\label{lem:t1_coercivity}

\medskip\noindent\textbf{(a) Regression identification and curvature.} Under (A2), (A4), and (A5), $\mathcal M_\rho$ has the unique globally separated maximizer $\bm\theta^\dagger$. In particular, for every $\delta>0$,
\begin{equation}\label{eq:s2_reg_global_separation}
\sup_{\|\theta-\theta^\dagger\|\ge\delta}\mathcal M_\rho(\bm\theta)<M^\star.
\end{equation}
Its negative Hessian $\bm H=-\nabla^2_{(\bm\beta,\bm\xi)}\mathcal M_\rho(\theta^\dagger)$ in the $(\bm\beta,\bm\xi)$ coordinates is positive definite, as are $\bm H_{\xi\xi}$ and $\bm H_{\mathrm{eff}}=\bm H_{\beta\beta}-\bm H_{\beta\xi}\bm H_{\xi\xi}^{-1}\bm H_{\xi\beta}$.

\medskip\noindent\textbf{(b) Fractional-moment coercivity.} For every $B>0$ there exist $v\in(0,1)$ and a compact set $C_B$ in $\mathbb R^p\times(0,\infty)\times[0,1]$ (with $s$ bounded away from zero when $\eta>0$) such that, for all sufficiently large $n$,
\begin{equation}
\label{eq:t1_fractional_coercivity}
\sup_{\theta\notin C_B}
\left\{\frac1{vk}\sum_{i\in K_n}
\log E_0 q_\theta(Y_i\mid x_i)^v
+\rho_n\log s\right\}\le-B.
\end{equation}
For each fixed compact $C$ on which $\sigma$ is bounded away from zero, as $v\downarrow0$,
\begin{equation}
\label{eq:t1_fractional_continuity}
\sup_{\theta\in C}
\left|\frac1{vk}\sum_{i\in K_n}
\log E_0q_\theta(Y_i\mid x_i)^v
-\frac1k\sum_{i\in K_n}E_0\log q_\theta(Y_i\mid x_i)
\right|\longrightarrow0
\end{equation}
uniformly in $n$.

\medskip\noindent\textbf{(c) Compact convergence of the adjusted criterion.} Let $D$ be a fixed compact subset of the parameter compactification with $\sigma$ bounded away from zero. Write
\[
A_n(\theta)=k^{-1}\sum_{i\in K_n}E_0
\log q_\theta(Y_i\mid x_i),\quad
A(\theta)=\int E_{g_0}\log q_\theta
(\alpha_0+z^\top b_0+U\mid x(z))\,\nu(dz).
\]
Then $\sup_{\theta\in D}|A_n(\theta)-A(\theta)|\to0$. If $\rho>0$ and $s$ is bounded away from zero on $D$, this extends to uniform convergence of $A_n(\theta)+\rho_n\log s$ to $\mathcal M_\rho(\theta)$. If $\rho=0$, the required upper comparison is instead
\[
\sup_{\theta\in D\cap\{s>0\}}
\{A_n(\theta)+\rho_n\log s-\mathcal M_0(\theta)\}
\le o(1).
\]
\end{lemma}
\begin{proof}

\emph{Part (a).} Put $\bm x(\bm z)=(1,\bm z^\top)^\top$ and let $B_x=\sup_{\bm z\in\operatorname{supp}\nu}\|\bm x(\bm z)\|<\infty$. By (A4), the scalar criterion is bounded above by $M^\star$, and for every $e>0$ it is at most $M^\star-d(e)$ whenever $\|(a,\bm\xi)-(a^\dagger,\bm\xi^\dagger)\|\ge e$, for some $d(e)>0$. If $\|\bm\xi-\bm\xi^\dagger\|\ge e$, this immediately gives $\mathcal M_\rho\le M^\star-d(e)$. Otherwise put $\bm h=\bm\beta-\bm\beta^\dagger$. Since $\int\{\bm x(\bm z)^\top\bm h\}^2\nu(d\bm z)\ge\lambda_{\min}(\bm Q_\nu)\|\bm h\|^2$ and $|\bm x(\bm z)^\top\bm h|\le B_x\|\bm h\|$, there are constants $c_1,c_2>0$, independent of $\bm h\ne0$, such that $\nu\{|\bm x(\bm z)^\top\bm h|\ge c_1\|\bm h\|\}\ge c_2$. Thus if $\|\bm h\|\ge e$ the scalar criterion is uniformly separated on a set of design probability at least $c_2$. Taking $e$ smaller than half the prescribed separation radius and using the preceding two cases proves the strict gap in \eqref{eq:s2_reg_global_separation}. For curvature, define the linear map $B(\bm z)(\bm h,\bm u)^\top=(\bm x(\bm z)^\top\bm h,\bm u^\top)^\top$, where $\bm u\in\mathbb R^2$. If $\bm J=-\nabla^2F_\rho(a^\dagger,\bm\xi^\dagger)\succ0$, the chain rule and (A5) give
\[
\begin{aligned}
(\bm h^\top,\bm u^\top)\bm H(\bm h^\top,\bm u^\top)^\top
&=\int\{B(\bm z)(\bm h,\bm u)^\top\}^\top\bm J\{B(\bm z)(\bm h,\bm u)^\top\}\,\nu(d\bm z)\\
&\ge\lambda_{\min}(\bm J)\{\lambda_{\min}(\bm Q_\nu)\|\bm h\|^2+\|\bm u\|^2\}.
\end{aligned}
\]
The Hessian and its Schur complement are therefore positive definite.

\emph{Part (b).} Fix $v\in(0,1)$, and write $J_v(\mu,\sigma,s)=E_{g_0}\{q_{\mu,\sigma,s}(Y)^v\}$. The uniform density bound $\|g_0\|_\infty<\infty$ and Jensen's inequality yield
\[
 J_v(\mu,\sigma,s)
 \le\{E_{g_0}q_{\mu,\sigma,s}(Y)\}^v
 \le\|g_0\|_\infty^v.
\]
We verify separately escape through the scale, location and mixture-weight boundaries.

For $\sigma\downarrow0$, the following argument makes uniformity in the location $\mu$ and mixture weight $s\in[0,1]$ explicit. Put $d=\sqrt\sigma$ and $A_d=\{|Y-\mu|\le d\}$. The mixture is a normalized density; therefore, using H\"older's inequality with $0<v<1$ and $\|g_0\|_\infty=G_0$,
\[
 E_0\{q_{\mu,\sigma,s}(Y)^v\mathbf1_{A_d}\}
 \le G_0(2d)^{1-v}
 \left\{\int_{A_d}q_{\mu,\sigma,s}(y)\,dy\right\}^{v}
 \le G_0(2\sqrt\sigma)^{1-v}.
\]
On $A_d^c$ we have $|Y-\mu|/\sigma\ge\sigma^{-1/2}$. The two componentwise density tails give, uniformly in $s$,
\[
 q_{\mu,\sigma,s}(Y)
 \le C\sigma^{-1}e^{-c|Y-\mu|/\sigma}
 +\frac{C}{|Y-\mu|\{\log(|Y-\mu|/\sigma)\}^{1+\gamma}}.
\]
For $d<|Y-\mu|\le1$, integrate the $v$th power of the second term against $g_0\le G_0$ using $\log(|Y-\mu|/\sigma)\ge\frac12\log(1/\sigma)$ and $\int_0^1t^{-v}\,dt=(1-v)^{-1}$. For $|Y-\mu|>1$, instead use $|Y-\mu|^{-v}\le1$. The exponential component contributes at most $C\sigma^{-v}e^{-cv/\sqrt\sigma}$. Consequently, for $0<\sigma<\sigma_*$,
\begin{equation}
\label{eq:t1_explicit_smallscale_uniform}
 \sup_{\mu\in\mathbb R,\,s\in[0,1]}J_v(\mu,\sigma,s)
 \le C_v\left[\sigma^{(1-v)/2}
 +\{\log(1/\sigma)\}^{-v(1+\gamma)}
 +\sigma^{-v}e^{-cv/\sqrt\sigma}\right]=o(1).
\end{equation}
This bound uses no moment assumption on $g_0$ beyond boundedness and is uniform in every location and weight.

For $\sigma\uparrow\infty$, the global density envelope gives
\[
 q_{\mu,\sigma,s}(Y)
 \le\frac C\sigma\left\{1+\log_+\sigma
 +\log_+\frac1{|Y-\mu|}\right\}.
\]
Since $\sup_{\mu}E_0\{\log_+(1/|Y-\mu|)\}^v \le2G_0\Gamma(v+1)$, subadditivity of $u\mapsto u^v$ yields
\begin{equation}
\label{eq:t1_explicit_largescale_uniform}
 \sup_{\mu\in\mathbb R,\,s\in[0,1]}J_v(\mu,\sigma,s)
 \le C_v\sigma^{-v}\{1+\log\sigma\}^v=o(1).
\end{equation}
Both constants depend on $v,\tau,\gamma,G_0$, but not on $\mu,s$ or the sample index.

Next let $\sigma$ range over a compact subset of $(0,\infty)$. If $|\mu|\to\infty$, then $J_v(\mu,\sigma,s)\to0$ uniformly in $\sigma,s$. Indeed, for fixed $R$, the likelihood tends to zero uniformly on $|Y|\le R$, while H\"older's inequality gives the uniform tail control
\[
 E\{q_{\mu,\sigma,s}(Y)^v\mathbf1(|Y|>R)\}
 \le\|g_0\|_\infty^v
       P(|Y|>R)^{1-v}.
\]
To see the uniform design separation used here, put $B_x=\sup_{n,i}\|\bm x_i\|$ and $\bm Q_{K,n}=k^{-1}\sum_{i\in K_n}\bm x_i\bm x_i^\top$. Assumption (A2) implies $\lambda_{\min}(\bm Q_{K,n})\ge\lambda_0>0$ eventually. For every unit vector $\bm h$, if $N_h=\#\{i\in K_n:|\bm x_i^\top\bm h|\ge\sqrt{\lambda_0/2}\}$, then $\lambda_0\le k^{-1}\sum_i(\bm x_i^\top\bm h)^2\le\lambda_0/2+B_x^2N_h/k$; hence $N_h/k\ge\lambda_0/(2B_x^2)$. This gives fixed $c_0,c_1>0$ such that at least $c_0k$ clean rows satisfy $|\bm x_i^\top\bm\Delta|\ge c_1\|\bm\Delta\|$, uniformly in nonzero $\bm\Delta$. Consequently, if $\|\bm\beta\|\to\infty$, a fixed fraction of the locations diverge uniformly while the other $J_v$ terms stay bounded by $\|g_0\|_\infty^v$. Their normalized logarithmic product therefore tends to $-\infty$.

Finally, if $\rho>0$ then $\rho_n$ is eventually bounded below by a positive constant and $\rho_n\log s\to-\infty$ uniformly as $s\downarrow0$. When $\rho=0$, the endpoint $s=0$ is retained in the compactification and the term $\rho_n\log s\le0$ can only improve the upper bound. Combining the preceding escaping-parameter estimates gives the coercivity assertion for every fixed $v\in(0,1)$; choosing a sufficiently large compact set gives \eqref{eq:t1_fractional_coercivity} for any $B>0$.

On a fixed parameter compact with $\sigma$ bounded away from zero, the negative part of $\log q$ is bounded by $C(1+|U|)$, uniformly in $\theta$; indeed $f_0/f_1$ is globally bounded so that $f_s\ge c f_0$. The positive part has uniformly integrable logarithmic singularities by boundedness of $g_0$. Thus the class of log likelihoods is uniformly integrable in $L^1(P_0)$. For $X=\log q$,
\[
 \frac1v\log E e^{vX}\longrightarrow EX
 \qquad(v\downarrow0)
\]
uniformly over the compact: truncate $|X|$ first, use the elementary Taylor expansion for bounded $X$, and then remove the truncation using uniform integrability and the uniform bound on positive exponential moments for sufficiently small $v$. This proves \eqref{eq:t1_fractional_continuity}.

\emph{Part (c).} Bounded design makes the location arguments range over a fixed compact interval.  Uniformly on $D$ the positive part of $\log q_\theta$ is controlled by a translated $\log_+\log_+(1/|U-a|)$ singularity, which is uniformly integrable under $\|g_0\|_\infty<\infty$. Since $\sup_u f_0(u)/f_1(u)<\infty$, $f_s\ge c f_0$ uniformly in $s\in[0,1]$; hence the negative log-density is dominated by $C_D(1+|U|)$. The map $(\theta,z)\mapsto E_0\log q_\theta$ is therefore continuous on the compact parameter--design domain, and the associated family of design functions is bounded and equicontinuous.  A finite-net argument and $\nu_n\Rightarrow\nu$ give uniform convergence of $A_n$. The assertion for $\rho>0$ follows because $\log s$ is bounded on $D$.  For $\rho=0$, use $\rho_n\log s\le0$; uniform convergence at $s\downarrow0$ is neither claimed nor needed.
\end{proof}

\begin{lemma}[Expanding-grid likelihood comparison]
\label{lem:t1_grid_ratio}

Fix $A>2$. Partition $\mathcal T_n(A)$ into cells with mesh $n^{-Q}$ in $\beta$, $\log\sigma$, and $\log s$ for $s\le1/2$, and mesh $n^{-Q}$ in $s$ for $s\ge1/2$. Choose each reference $\theta_j$ with its $s$ coordinate equal to the upper endpoint of the cell. There exist finite $Q$ and $\epsilon_n=o(k^{-1})$ such that the number of cells is polynomial in $n$ and
\begin{equation}
\label{eq:t1_ratio_modulus}
\sup_j\sup_{\theta\in C_{j,n}}\sup_{i\in K_n}
E_0\left\{
\frac{q_\theta(Y_i\mid x_i)}
{q_{\theta_j}(Y_i\mid x_i)}\right\}
\le1+\epsilon_n.
\end{equation}
The ratio is defined outside a $P_0$-null set.
\end{lemma}
\begin{proof}

Let $\Delta=n^{-Q}$ and put $h_n=n^{3A+3}\Delta$. Increase $Q$ so that $h_n\to0$ faster than every predetermined negative power of $n$. The reference $s_j$ is the upper endpoint of its cell. Hence $s\le s_j$, and $|s-s_j|\le C\Delta$.

Write $\mu_j=x_i^\top\beta_j$, $u=(Y_i-\mu_j)/\sigma_j$, and $u'=(Y_i-x_i^\top\beta)/\sigma$. Since $\sigma_j\in[n^{-A},n^A]$, the bounded design and cell mesh give, on $|u|\le n^{2A}$,
\[
|u'-u|\le Cn^{2A+2}\Delta\le h_n.
\]
The ratio $\sigma_j/\sigma$ is $1+O(\Delta)$.

We first compare the mixture weights at a common standardized residual. If $s_j\le1/2$, then
\[
\frac{f_s(u)}{f_{s_j}(u)}
\le\max\left\{\frac{s}{s_j},
\frac{1-s}{1-s_j}\right\}
\le1+C\Delta.
\]
If $s_j>1/2$, write $f_s=f_{s_j}+(s_j-s)(f_0-f_1)$. The continuous ratio $f_0/f_1$ is uniformly bounded on $\mathbb R\setminus\{0\}$, including both tails and the logarithmic singularity, so again $f_s(u)/f_{s_j}(u)\le1+C\Delta$.

Split the reference residual domain into
\[
D_0=\{|u|\le2\sqrt{h_n}\},\quad
D_1=\{2\sqrt{h_n}<|u|\le n^{2A}\},\quad
D_2=\{|u|>n^{2A}\}.
\]
On $D_0$, the denominator is bounded below by $c/\sigma_j$ if $s_j\le1/2$, because the AL component is bounded below near zero. If $s_j>1/2$, it is bounded below by $c f_1(u)/\sigma_j\ge c'/\sigma_j$. The numerator is bounded above by $C\sigma^{-1}\{1+\log_+(1/|u'|)\}$. Since the clean-response density is bounded by $\|g_0\|_\infty$ and $\sigma_j\le n^A$,
\begin{align*}
&E_0\left[\frac{q_\theta(Y_i\mid x_i)}
{q_{\theta_j}(Y_i\mid x_i)}\mathbf1_{D_0}\right]\\
&\quad\le Cn^A\int_{|u|\le2\sqrt{h_n}}
\{1+\log_+(1/|u'|)\}\,du\\
&\quad\le Cn^A\sqrt{h_n}
\{1+|\log h_n|\}.
\end{align*}
The last inequality follows from the affine relation $u'=(\sigma_j/\sigma)u+c$, whose slope is bounded above and below, and the elementary integral of a logarithmic singularity over an interval.

On $D_1$, both $u$ and $u'$ remain separated from zero by $\sqrt{h_n}$. The uniform derivative bound of Lemma~\ref{lem:s1_analytic_bounds}, part (b),
\[
|\partial_u\log f_s(u)|
\le C\{1+|u|^{-1}\}
\]
for interior and boundary mixture weights, combined with the preceding weight comparison, gives
\[
\sup_{D_1}\frac{q_\theta}{q_{\theta_j}}
\le1+C\sqrt{h_n}+C\Delta.
\]
Finally, on $D_2$, $f_0(u)$ and $f_0(u')$ are exponentially smaller than $s_j f_1(u)$, uniformly for $s_j\ge e^{-n^A}$. Lemma~\ref{lem:s1_analytic_bounds}, part (b) gives $|u\partial_u\log f_1(u)|\le C$ for large $|u|$. Indeed $f_0(u)/(s_jf_1(u))\le C\exp\{-c n^{2A}+n^A+O(\log n)\}$, which is exponentially small uniformly in the reference $s_j$. Together with $u'/u=1+O(\Delta)$, this gives
\[
\sup_{D_2}\frac{q_\theta}{q_{\theta_j}}
\le1+C\Delta+e^{-c n^{2A}}.
\]
For the transition $|u|=n^{2A}$, the same estimate is uniform: $|u'-u|\le C(\Delta|u|+h_n)$, and hence $|\log f_1(u')-\log f_1(u)|\le C\{\Delta+h_n/n^{2A}\}$ by the log-tail derivative bound. Furthermore $f_0(u)/(s_jf_1(u))\le C\exp\{n^A-cn^{2A}\}$, uniformly for $|u|\ge n^{2A}/2$, so neither the upper reference weight nor the AL component changes the comparison. Thus the $D_2$ bound is uniform across all cells and the transition region. Combining the three regions and using $P_0(D_1\cup D_2)\le1$ gives
\[
E_0\frac{q_\theta(Y_i\mid x_i)}
{q_{\theta_j}(Y_i\mid x_i)}
\le1+C\{n^A\sqrt{h_n}(1+|\log h_n|)
+\sqrt{h_n}+\Delta+e^{-c n^{2A}}\}.
\]
Choose, for example, $Q>12A+30$. The displayed error is then $o(n^{-1})$, uniformly in the cell, the reference, and the clean-design observation. The ranges of $\beta$, $\log\sigma$ and $\log s$ in the sieve are polynomial in $n$; hence the prescribed grids contain only polynomially many cells. This proves the lemma.

\end{proof}

\begin{lemma}[Global tilted integrated-likelihood bounds]
\label{lem:integrated_likelihood_bounds}

For every nonempty relatively closed $F\subseteq\Theta$ and each $\delta>0$, with $\mathcal Z_n(F)=\int_Fs^mL_{K,n}(\theta)\pi(d\theta)$,
\begin{align}
&P_0\left\{k^{-1}\log\mathcal Z_n(F)
>\sup_{\theta\in F}\mathcal M_\rho(\theta)+\delta\right\}
\longrightarrow0,
\label{eq:global_integrated_upper}\\
&P_0\left\{k^{-1}\log\widetilde Z_n
<M^\star-\delta\right\}\longrightarrow0,
\label{eq:integrated_lower}
\end{align}
where $\widetilde Z_n=\mathcal Z_n(\Theta)$ and $M^\star=\sup_\theta\mathcal M_\rho(\theta)$.
\end{lemma}
\begin{proof}

\emph{Lower bound.} Choose an interior $\theta_\delta$ with $\mathcal M_\rho(\theta_\delta)>M^\star-\delta/4$. Continuity allows a sufficiently small fixed interior prior-neighborhood $B$ with $\inf_{\theta\in B}\mathcal M_\rho(\theta)>M^\star-\delta/2$. Write $\pi_B=\pi(\cdot\mid B)$.  Jensen's inequality yields
\[
\log\widetilde Z_n\ge \log\pi(B)
+\sum_{i\in K_n}\int_B\log q_\theta(Y_i\mid x_i)
\,\pi_B(d\theta)
+m\int_B\log s\,\pi_B(d\theta).
\]
The integrated log-likelihood summands are independent, with uniformly integrable absolute values: their positive parts admit a uniformly integrable translated log--log envelope, and their negative parts are bounded by $C_B(1+|U_i|)$.  Truncate at a fixed level, apply the variance bound to the bounded summands, and remove the truncation.  The resulting triangular weak law, together with Lemma~\ref{lem:t1_coercivity}(c), yields $\liminf k^{-1}\log\widetilde Z_n\ge M^\star-\delta/2$ in $P_0$-probability, proving \eqref{eq:integrated_lower}.

\emph{Upper bound: quantifier order.} Fix $F$ and $\delta$ and put $M_F=\sup_{\theta\in F}\mathcal M_\rho(\theta)$ and $u=M_F+\delta$.  Select $B_0>1+|u|+|M^\star|$. Lemma~\ref{lem:t1_coercivity} supplies a fixed exponent $v_T\in(0,1)$ and compact $C_{B_0}$ such that
\[
\sup_{\theta\notin C_{B_0}}\Phi_{v_T,n}(\theta)
\le-B_0,\qquad
\Phi_{v,n}(\theta)=\rho_n\log s+
\frac1{vk}\sum_{i\in K_n}
\log E_0\{q_\theta(Y_i\mid x_i)^v\}.
\]
Choose fixed compact enlargements $C_{B_0}\subset\operatorname{int}C_{B_0}^+ \subset\operatorname{int}C_{B_0}^{++}$. The compact expansion can include $s=0$ if $\rho=0$; for $\rho>0$ it is separated from that boundary. The small-$v$ statement of Lemma~\ref{lem:t1_coercivity}(b)--(c) provide a possibly different fixed exponent $v_L\in(0,1)$ such that
\[
\sup_{\theta\in C_{B_0}^{++}}
\{\Phi_{v_L,n}(\theta)-\mathcal M_\rho(\theta)\}
<\delta/4
\]
for all large $n$.  For $\rho=0$ this is the one-sided comparison, using $\rho_n\log s\le0$.

Choose $A$ large enough that $\mathcal T_n(A)$ eventually contains all points of $C_{B_0}^{++}$ with $s\ge e^{-n^A}$, while Lemma~\ref{lem:t1_integrated_envelope} gives
\[
\int_{\mathcal T_n(A)^c}
L_{K,n}(\theta)\pi(d\theta)
=o(e^{ku})
\]
uniformly in the observations. When $\rho=0$ and $C_{B_0}^{++}$ meets $s=0$, this assertion does not say that the sieve contains that compactification boundary: the excluded $0\le s<e^{-n^A}$ layer belongs entirely to the sieve complement and is controlled by the preceding integrated bound. Only after $A$ is chosen, fix the polynomial grid exponent $Q$ of Lemma~\ref{lem:t1_grid_ratio}.  Let $\{C_{j,n},\theta_j\}$ be the resulting cells and upper-$s$ reference points.  Call a cell local when $\theta_j\in C_{B_0}^+$ and a tail cell otherwise. Every local cell meeting $F$ has a point of $F\cap C_{j,n}$ within $o(1)$ of $\theta_j$; by compact continuity of $\mathcal M_\rho$,
\[
\Phi_{v_L,n}(\theta_j)\le M_F+\delta/2.
\]
Every tail reference lies outside $C_{B_0}$, so $\Phi_{v_T,n}(\theta_j)\le-B_0$.

\emph{Fractional integration.} Fix either exponent $v\in\{v_L,v_T\}$, write $t=v/(1+v)$, and set
\[
V_{j,n}=\int_{F\cap C_{j,n}}
\left(\frac{s}{s_j}\right)^m
\prod_{i\in K_n}
\frac{q_\theta(Y_i\mid x_i)}
{q_{\theta_j}(Y_i\mid x_i)}\,\pi(d\theta),
\]
where $s_j$ is the reference's upper mixture weight. Since $s/s_j\le1$ and the clean observations are independent, Fubini, the grid-ratio bound, and Jensen give
\[
E_0V_{j,n}\le\pi(F\cap C_{j,n})(1+\epsilon_n)^k.
\]
Moreover $\mathcal Z_n(F\cap C_{j,n}) =s_j^mL_{K,n}(\theta_j)V_{j,n}$. H\"older with exponents $1+v$ and $(1+v)/v$ yields
\begin{align}
&E_0\{\mathcal Z_n(F\cap C_{j,n})^t\}
\nonumber\\
&\quad\le
\pi(F\cap C_{j,n})^t(1+\epsilon_n)^{kt}
\exp\{kt\Phi_{v,n}(\theta_j)\}.
\label{eq:t1_cell_fractional_revised}
\end{align}
The number of cells $N_n$ is polynomial in $n$. For $c\in\{L,T\}$, use $(\sum_j a_j)^{t_c}\le\sum_j a_j^{t_c}$ and Markov's inequality to obtain
\begin{align*}
P_0\{\mathcal Z_{n,L}(F)>e^{ku}/3\}
&\le 3^{t_L}N_n
\exp\{-kt_L\delta/2+o(1)\},\\
P_0\{\mathcal Z_{n,T}(F)>e^{ku}/3\}
&\le 3^{t_T}N_n
\exp\{-kt_T(u+B_0)+o(1)\}.
\end{align*}
Both probabilities tend to zero.  The sieve-complement integral is $o(e^{ku})$ deterministically, establishing \eqref{eq:global_integrated_upper}.
\end{proof}

\subsection{Transfer to the full posterior}

Define the $s^m$-tilted clean posterior by
\[
\widetilde\Pi_n(d\theta)=\widetilde Z_n^{-1}s^mL_{K,n}(\theta)\pi(d\theta),\qquad \widetilde Z_n=\mathcal Z_n(\Theta).
\]
Put $B_n=\prod_{i\in L_n}f_1(Y_i)$ and $R_n(\theta)=\prod_{i\in L_n}q_\theta(Y_i\mid\bm x_i)/\{sf_1(Y_i)\}$, with empty products equal to one. Then
\begin{equation}\label{eq:exact_posterior_identity}
\Pi_n(A\mid D_n)=\frac{E_{\widetilde\Pi_n}\{\mathbf1_A R_n\}}{E_{\widetilde\Pi_n}R_n}
\end{equation}
for every measurable $A$.

\begin{lemma}[Full-posterior concentration transfer]
\label{lem:growing_contamination_reduction}
Suppose (A1)--(A5) hold and $M^\star=\sup_{\theta\in\Theta}\mathcal M_\rho(\theta)$. For every relatively closed $F\subset\Theta$ such that $M_F=\sup_{\theta\in F}\mathcal M_\rho(\theta)<M^\star$, there exists $c_F>0$ satisfying
\[
P_0\{\Pi_n(F\mid D_n)\le e^{-c_Fk}\}\longrightarrow1.
\]
The assertion requires no total-variation approximation between the full and tilted posteriors.
\end{lemma}

\begin{proof}
Write $J_n(A)=\int_A L_{K,n}(\theta)\prod_{i\in L_n}q_\theta(Y_i\mid\bm x_i)\,\pi(d\theta)$ for the full unnormalized posterior mass of $A$. If $m=0$, $J_n(A)=\mathcal Z_n(A)$ and the conclusion follows from Lemma~\ref{lem:integrated_likelihood_bounds}. We henceforth consider $m\ge1$. All likelihood comparisons below are deterministic conditional on the fixed contaminated observations and design; only the clean likelihood is random.

\emph{The normalization and the moderate region.} For $i\in L_n$ put $b_i=f_1(Y_i)$ and $B_n=\prod_{i\in L_n}b_i$. The contamination bounds in (A3) and the two-sided tail equivalents give, uniformly in the contaminated indices,
\begin{equation}\label{eq:s2_Bn_asymptotic}
\log B_n=-m\log\omega_n-(1+\gamma)m\log\log\omega_n+O(m).
\end{equation}
For every fixed compact interior neighborhood $N$ of $\theta^\dagger$, the tail equivalents are uniform in $\theta\in N$ and $i\in L_n$; in particular, $\sup_{\theta\in N}|\log R_n(\theta)|\le C_Nm\{(\log\omega_n)^{-1}+\omega_n^{-1}\}=o(k)$. Choose $N$ sufficiently small that $\inf_{\theta\in N}\mathcal M_\rho(\theta)>M^\star-\delta$ for a fixed $\delta>0$. The local lower-bound argument of Lemma~\ref{lem:integrated_likelihood_bounds}, applied to $N$, yields
\begin{equation}\label{eq:t1_new_full_normalizer}
P_0\{\widetilde Z_n\ge e^{k(M^\star-2\delta)},\ J_n(\Theta)\ge B_ne^{k(M^\star-3\delta)}\}\longrightarrow1.
\end{equation}
We next show that a fixed polynomially growing parameter region captures the full posterior at arbitrarily high exponential precision. Fix $\lambda\in(\eta,1)$ (also when $\eta=0$); the fixed growth exponent $B_*$ in (A3) will be chosen sufficiently large below. Define $\mathcal E_n(\lambda)=\{\|\bm\beta\|\le\omega_n^\lambda,\ \omega_n^{-\lambda}\le\sigma\le\omega_n^\lambda\}$. The data-uniform $r=2$ integrated-likelihood bound in Lemma~\ref{lem:t1_integrated_envelope}, Cauchy--Schwarz, and the Gaussian and inverse-gamma prior tails imply
\begin{align*}
J_n\{\|\bm\beta\|>\omega_n^\lambda\}&\le\exp\{Cn\log(n+1)-c\omega_n^{2\lambda}\},\\
J_n\{\sigma<\omega_n^{-\lambda}\}&\le\exp\{Cn\log(n+1)-c\omega_n^\lambda\},\\
J_n\{\sigma>\omega_n^\lambda\}&\le\exp\{-\lambda n\log\omega_n+Cn\log(n+1)+Cn\log\log\omega_n\}.
\end{align*}
The final bound follows directly by integrating the upper-scale density envelope; for $\sigma>\omega_n^\lambda$, the factor $\sigma^{-n}(1+\log\sigma)^n$ is decreasing for large $n$. Since $m/n\to\eta<\lambda$, the normalization in \eqref{eq:s2_Bn_asymptotic} gives the upper-scale logarithmic ratio
$-(\lambda-\eta+o(1))n\log\omega_n+O\{n\log(n+1)+n\log\log\omega_n\}$.
For the lower-scale and coefficient escapes the corresponding ratios are bounded by
$-c\omega_n^\lambda+O\{n\log(n+1)+m\log\omega_n\}$ and
$-c\omega_n^{2\lambda}+O\{n\log(n+1)+m\log\omega_n\}$, respectively.
Thus, since $\omega_n\ge n^{B_*}$, the strict inequalities
$B_*\lambda>1$ and $B_*(\lambda-\eta)>C_*$ for a sufficiently large fixed $C_*$ depending only on the data-uniform likelihood constants give an $-cn\log n$ bound on all three logarithmic ratios. This choice of $B_*$ is independent of $D$. Consequently, for every prescribed $D>0$,
\begin{equation}\label{eq:t1_Etail}
J_n\{\mathcal E_n(\lambda)^c\}\le B_ne^{-Dk}
\end{equation}
for all sufficiently large $n$. In particular, the bound holds for every fixed $D$ after a single sufficiently large choice of $B_*$, since the three logarithmic ratios to $B_n$ then have order at most $-cn\log n$.

\emph{Comparison with the tilted likelihood.} On $\mathcal E_n(\lambda)$, the contaminated residuals satisfy $|Y_i-\bm x_i^\top\bm\beta|\asymp\omega_n$ and $|(Y_i-\bm x_i^\top\bm\beta)/\sigma|\ge c\omega_n^{1-\lambda}$. The componentwise AL and LPAL tails therefore give fixed $C_\lambda,c>0$ and $\varepsilon_n\le\exp(-c\omega_n^{1-\lambda})$ such that, uniformly over $\theta\in\mathcal E_n(\lambda)$ and $i\in L_n$,
\begin{equation}\label{eq:s2_component_outlier_comparison}
\frac{\sigma^{-1}f_1\{(Y_i-\bm x_i^\top\bm\beta)/\sigma\}}{f_1(Y_i)}\le C_\lambda,\qquad \frac{\sigma^{-1}f_0\{(Y_i-\bm x_i^\top\bm\beta)/\sigma\}}{f_1(Y_i)}\le\varepsilon_n.
\end{equation}
For the first inequality, the scale factor cancels the $|u|^{-1}$ tail, while $\log\{|Y_i-\bm x_i^\top\bm\beta|/\sigma\}\ge(1-\lambda+o(1))\log\omega_n$; the second follows from the exponential AL tail, with its polynomial prefactor absorbed into the exponential. Fix $A$ sufficiently large for the sieve bound of Lemma~\ref{lem:t1_integrated_envelope} and increase $B_*$, if necessary, so that $B_*(1-\lambda)>A$ and $B_*\lambda>1$. Then $m\varepsilon_ne^{n^A}=o(1)$. Writing $\mathcal T_n=\mathcal T_n(A)$, \eqref{eq:s2_component_outlier_comparison} gives the uniform inequalities
\begin{equation}\label{eq:t1_new_R_bound}
\sup_{\theta\in\mathcal E_n(\lambda),\ s\ge e^{-n^A}}R_n(\theta)\le (C_\lambda+\varepsilon_ne^{n^A})^m\le e^{C_\lambda' m+o(1)},\qquad
\prod_{i\in L_n}\frac{q_\theta(Y_i\mid\bm x_i)}{f_1(Y_i)}\le (C_\lambda e^{-n^A}+\varepsilon_n)^m\quad(s<e^{-n^A}).
\end{equation}
The first inequality holds on the entire moderate region, not merely its intersection with the sieve. Consequently, using the clean-data sieve-complement estimate for $s\ge e^{-n^A}$ and the data-uniform first integrated-moment bound for $s<e^{-n^A}$,
\begin{align}
J_n\{\mathcal E_n(\lambda)\cap\mathcal T_n^c\}
&\le B_n\left[e^{C_\lambda'm-c_A n\log n}+e^{-mn^A+O(n\log(n+1))}\right]\nonumber\\
&\le B_ne^{-Dk}
\label{eq:s2_sieve_full_tail}
\end{align}
for every prescribed $D>0$ and all sufficiently large $n$. For the second term, we used $m\ge1$, $A>2$, and $\varepsilon_n=o(e^{-n^A})$. The analogous tilted posterior bound follows directly from $s^m\le1$ and the clean-data sieve estimate.

\emph{Exponential compact trimming.} The fractional coercivity in Lemma~\ref{lem:t1_coercivity}, together with Jensen's inequality, yields a fixed compact $C_D$ such that $\sup_{\theta\notin\operatorname{int}C_D}\mathcal M_\rho(\theta)<-2D$ for any prescribed $D>0$: indeed $E_0\log q_\theta\le v^{-1}\log E_0q_\theta^v$ pointwise, and the design-limit criterion inherits the upper bound by compact convergence and the uniform escaping-parameter estimates in the proof of Lemma~\ref{lem:t1_coercivity}(b). When $\rho=0$, the compactification may include $s=0$, whereas for $\rho>0$ it is separated from that boundary. Lemma~\ref{lem:integrated_likelihood_bounds}, applied to the relatively closed complement, then gives
\begin{equation}\label{eq:s2_compact_tilted_trim}
P_0\{\mathcal Z_n(\Theta\setminus\operatorname{int}C_D)\le e^{-Dk}\}\longrightarrow1.
\end{equation}
Combining \eqref{eq:t1_Etail}, \eqref{eq:t1_new_R_bound}, \eqref{eq:s2_sieve_full_tail}, and \eqref{eq:s2_compact_tilted_trim} shows that $J_n(C_D^c)\le B_n\{e^{C_\lambda'm-Dk}+e^{-D_0k}\}$ with probability tending to one, where $D_0>0$ can be chosen arbitrarily in the two sieve/moderate-region bounds. In view of \eqref{eq:t1_new_full_normalizer}, both the tilted posterior and full posterior assign exponentially small mass to $C_D^c$, with any prescribed fixed exponential rate upon enlarging $D$.

\emph{Compact likelihood equivalence at the exponential scale.} On the fixed compact $C_D$, the regression coefficients and scale remain bounded and the scale is separated from zero. The log-regularly varying tail equivalents, uniformly in $i\in L_n$, give
\[
\sup_{\theta\in C_D}\left|\frac{\sigma^{-1}f_1\{(Y_i-\bm x_i^\top\bm\beta)/\sigma\}}{f_1(Y_i)}-1\right|\le C_D\{(\log\omega_n)^{-1}+\omega_n^{-1}\}.
\]
On $C_D\cap\mathcal T_n$, $s\ge e^{-n^A}$, and the AL contribution relative to $sf_1(Y_i)$ is bounded by $e^{n^A-c\omega_n}$, which is negligible since $B_*>A$. Thus
\begin{equation}\label{eq:t1_new_compact_logR}
\sup_{\theta\in C_D\cap\mathcal T_n}|\log R_n(\theta)|\le C_Dm\{(\log\omega_n)^{-1}+\omega_n^{-1}\}+O(me^{n^A-c\omega_n})=o(k).
\end{equation}
This remains valid when $\eta=0$ and $C_D$ includes $s=0$ in its compactification; the actual parameter values in the comparison satisfy $s\ge e^{-n^A}$. Now put $g=M^\star-M_F>0$ and choose $0<\delta<g/8$ in the normalizer bound. Lemma~\ref{lem:integrated_likelihood_bounds} gives $\mathcal Z_n(F)\le e^{k(M_F+\delta)}$ with probability tending to one. Choose $D$ and the tail exponent $D_0$ so large that the complement terms in the preceding paragraph are at most $B_ne^{k(M^\star-4\delta)}$. On their intersection with the high-probability events from \eqref{eq:t1_new_full_normalizer} and \eqref{eq:t1_new_compact_logR},
\[
\frac{J_n(F)}{J_n(\Theta)}\le e^{k(M_F+2\delta-M^\star+3\delta)}+e^{-\delta k}\le 2e^{-\delta k}
\]
for all sufficiently large $n$. Replacing $\delta$ by a smaller fixed constant absorbs the factor two and proves the assertion.
\end{proof}

In particular, Lemmas~\ref{lem:t1_coercivity} and~\ref{lem:growing_contamination_reduction} imply, for every $\varepsilon>0$,
\begin{equation}\label{eq:s2_global_posterior_consistency}
\Pi_n\{\|\bm\beta-\bm\beta^\dagger\|>\varepsilon\mid D_n\}\xrightarrow{P_0}0.
\end{equation}
Indeed, the complement of an $\varepsilon$-neighborhood is relatively closed and has a strictly smaller population criterion by Lemma~\ref{lem:t1_coercivity}(a). This also gives concentration of $\bm\xi$ at $\bm\xi^\dagger$.

\section{Proof of Theorem 1}
\label{sec:supp_theorem1}

We use the notation of Section S2 and Assumptions (A1)--(A5) of the main manuscript. In local calculations, $\bm{\beta}$ and $\bm{b}$ denote the regression and slope vectors, respectively, and $\bm{\xi}=(\log\sigma,\operatorname{logit}s)$ denotes the nuisance vector. Write $k=|K_n|$, $m=|L_n|$, $\rho_n=m/k$, and let $r=r_k$ be the unique solution of $kr^3\{\log(1/r)\}^2=1$ in $(0,e^{-2/3})$ for all sufficiently large $k$ (the small root, satisfying $r_k\downarrow0$). The full posterior is $\Pi_n(\cdot\mid D_n)$ and its slope marginal is $\Pi_n^{\bm{b}}(\cdot\mid D_n)$. Define the theoretical screening set $\widehat K_n=\{i:|Y_i|\le n^{1/q}\}$, where $q\in(1,2)$ is the exponent in (A1), and put $T_n(\bm{b})=\widehat Q_\tau\{Y_i-\bm{z}_i^\top \bm{b}:i\in\widehat K_n\}$, with $T_n(\bm{b})=0$ when $\widehat K_n$ is empty. The corrected posterior $\Pi_n^{\mathrm{cor}}$ is the pushforward of $\Pi_n^{\bm{b}}$ under $\bm{b}\mapsto (T_n(\bm{b}),\bm{b}^\top)^\top$, and the correction variable is $\Delta_n=\alpha-T_n(\bm{b})$:
\begin{equation}\label{eq:s3_corrected_posterior}
\Pi_n^{\mathrm{cor}}(A\mid D_n)=\int\mathbf1\{(T_n(\bm{b}),\bm{b}^\top)^\top\in A\}\,\Pi_n^{\bm{b}}(d\bm{b}\mid D_n),\qquad A\in\mathcal B(\mathbb R^p).
\end{equation}
The local arguments do not require $|\rho_n-\rho|=o(r)$ or $m/\log\omega_n\to0$; the moving population center, rather than $\bm\vartheta^\dagger$, is used for all local likelihood comparisons.

\subsection{Moving-center and nuisance-profile analysis}

For a fixed sufficiently small compact convex neighborhood $N$ of the pseudo-true parameter $\bm{\vartheta}^\dagger$, define the normalized expected full-data criterion
\begin{equation}\label{eq:new_full_criterion}
M_{n,\omega}^{\rm full}(\bm{\vartheta})=\frac1k\sum_{i\in K_n}E_0\log q_{\bm{\vartheta}}(Y_i\mid \bm{x}_i)+\frac1k\sum_{i\in L_n}\{\log q_{\bm{\vartheta}}(Y_i\mid \bm{x}_i)-\log f_1(Y_i)\},\qquad\bm{\vartheta}\in N.
\end{equation}
The subtraction is parameter-independent. This criterion is used only locally: the full-data log likelihood is infinite at exact fits and need not admit an unrestricted maximizer. The next lemma constructs a local stationary center that does not introduce a rate condition on $\rho_n-\rho$.

\begin{lemma}[Outlier differentiation and the moving finite-sample stationary center]
\label{lem:new_moving_center}
\medskip\noindent\textbf{(a) Two-sided differentiated tail bounds.}

Let $b_\gamma=1+\gamma$, and let $N$ be as above. Uniformly in $\bm{\vartheta}\in N$ and $i\in L_n$, with $e_i=Y_i-\bm{x}_i^\top\bm{\beta}$,
\begin{align}\label{eq:new_outlier_expansion}
 \log q_{\bm{\vartheta}}(Y_i\mid \bm{x}_i)-\log f_1(Y_i)
 &=\log s+O\{(\log\omega_n)^{-1}\},\\
 \|\partial_{\bm{\beta}}\log q_{\bm{\vartheta}}(Y_i\mid \bm{x}_i)\|
 &\le C\omega_n^{-1},\qquad
 \partial_\lambda\log q_{\bm{\vartheta}}(Y_i\mid \bm{x}_i)
 =\frac{b_\gamma}{\log\omega_n}+O\{(\log\omega_n)^{-2}\},
 \label{eq:new_outlier_scores}\\
 \partial_\zeta\log q_{\bm{\vartheta}}(Y_i\mid \bm{x}_i)
 &=1-s+O(e^{-c\omega_n}),
\end{align}
where $c,C>0$ are independent of $n,i,\bm{\vartheta}$. All second $\bm{\beta}\bm{\beta}$ derivatives have size $O(\omega_n^{-2})$, all mixed $\bm{\beta}\lambda$ derivatives have size $O\{\omega_n^{-1}(\log\omega_n)^{-2}\}$, and
\begin{equation}\label{eq:new_outlier_hess}
 \partial_{\lambda\lambda}^2\log q_{\bm{\vartheta}}(Y_i\mid \bm{x}_i)
 =O\{(\log\omega_n)^{-2}\}.
\end{equation}
All second derivatives involving $\zeta$, after subtracting $\log s$, are exponentially small up to polynomial factors. The same assertions hold for each fixed order of mixed derivatives, with the natural analogous bounds.

\medskip\noindent\textbf{(b) Local stationary point and slope bias.}

There exist $\delta>0$ and, for all sufficiently large $n$, a unique local stationary point $\bm{\vartheta}_{n,\omega}^\circ$ of $M_{n,\omega}^{\rm full}$ in $\{\|\bm{\vartheta}-\bm{\vartheta}^\dagger\|<\delta\}$. It is the strict maximizer within this neighborhood, and
\begin{align}\label{eq:new_center_rates}
 \|\bm{\vartheta}_{n,\omega}^\circ-\bm{\vartheta}^\dagger\|
 &\le C\left\{|\rho_n-\rho|+
 \frac{\rho_n}{\log\omega_n}\right\},\\
 \|\bm{b}_{n,\omega}^\circ-\bm{b}_0\|
 &\le C\frac{\rho_n}{\omega_n}=o(r).
 \label{eq:new_slope_bias}
\end{align}
No rate assumption on $\rho_n-\rho$ is used except convergence, and no condition $m_n/\log\omega_n\to0$ is used. If $m_n=0$, the $\rho_n/\omega_n$ terms are interpreted as zero. Furthermore,
\begin{equation}\label{eq:new_full_hessian}
 -\nabla^2 M_{n,\omega}^{\rm full}
 (\bm{\vartheta}_{n,\omega}^\circ)\longrightarrow \bm{H}\succ0.
\end{equation}

\end{lemma}
\begin{proof}
\medskip\noindent\textit{Proof of (a).} Put $K(a)=\int_0^\infty v^{-1}e^{-a/v}H(v;\gamma)\,dv$. For $e\ne0$ let $a=\rho_\tau(e)/\sigma$ and $A_s(a)=(1-s)e^{-a}+sK(a)$. Then exactly
\[
 q_{\bm{\vartheta}}(Y_i\mid \bm{x}_i)=\frac{\tau(1-\tau)}{\sigma}A_s(a).
\]
The differentiated LPAL tail expansion, including its uniform remainder for every fixed derivative order, is proved once in Lemma~\ref{lem:s3_nuisance}(i). In particular, with $D=a\partial_a$ and $b_\gamma=1+\gamma$, that lemma gives $D\log K(a)=-1-b_\gamma/\log a+O\{(\log a)^{-2}\}$ and $D^j\log K(a)=O_j\{(\log a)^{-2}\}$ for all fixed $j\ge2$. Since $s$ is bounded away from zero, the AL-to-LPAL term $(1-s)e^{-a}/\{sK(a)\}$ and all its fixed-order derivatives are $O\{e^{-a/2}\}$ uniformly on $N$. Here $|e_i|\asymp\omega_n$ and $a\asymp\omega_n$, so
\[
 \log q_{\bm{\vartheta}}(Y_i\mid \bm{x}_i)
 =-\log|e_i|-b_\gamma\log\log a
 +\log s+C_{\operatorname{sgn}e_i}
 +O\{(\log\omega_n)^{-1}\}.
\]
The corresponding expression for $\log f_1(Y_i)$ has the same sign-dependent constant, while $e_i/Y_i=1+O(\omega_n^{-1})$. This proves the first display. The derivative identities
\[
 \partial_\lambda\log q=-1-D\log A_s(a),\quad
 \partial_\zeta\log q=\frac{sK(a)}{A_s(a)}-s,\quad
 \partial_{\bm{\beta}}\log q=-\frac{\bm{x}_i}{e_i}D\log A_s(a)
\]
and their derivatives establish all remaining claims. For $\partial_{\bm{\beta}\lambda}$ the leading constant $-1$ in $D\log K$ cancels, leaving $D^2\log K/e_i$.

\medskip\noindent\textit{Proof of (b).} Define the one-dimensional clean criterion $F(a,\bm{\xi})=E_{g_0}\log\{\sigma^{-1}f_s((U-a)/\sigma)\}$. The local $C^2$ regularity and continuous Hessian of $F$ are explicitly guaranteed by (A5). For every $n$ the no-outlier, tilted population criterion is exactly
\begin{equation}\label{eq:new_clean_factorization}
 M_n^{\rm tilt}(\alpha,\bm{b},\bm{\xi})
 =k^{-1}\sum_{i\in K_n}
 F(\alpha-\alpha_0+\bm{z}_i^\top(\bm{b}-\bm{b}_0),\bm{\xi})
 +\rho_n\log s.
\end{equation}
At $\bm{b}=\bm{b}_0$, the dependence on the design disappears: $M_n^{\rm tilt}(\alpha_0+a,\bm{b}_0,\bm{\xi})=F(a,\bm{\xi})+\rho_n\log s$. At $\rho$ this function has the stationary point $(a^\dagger,\bm{\xi}^\dagger)$. Its Hessian on $(a,\bm{\xi})$ is negative definite because it is a principal restriction of $-\bm{H}\prec0$. By the inverse function theorem, for $\rho_n$ near $\rho$ there is a unique local root $(a_n,\bm{\xi}_n)$ of the $(a,\bm{\xi})$ score and $\|(a_n,\bm{\xi}_n)-(a^\dagger,\bm{\xi}^\dagger)\| \le C|\rho_n-\rho|$. The full $\bm{\beta}$ score of \eqref{eq:new_clean_factorization} at $(\alpha_0+a_n,\bm{b}_0,\bm{\xi}_n)$ equals zero, since its slope score is $\bar{\bm{z}}_{K,n}F_a(a_n,\bm{\xi}_n)=0$. The clean Hessian is negative definite uniformly for large $n$: at $\bm{b}=\bm{b}_0$ its $\bm{\beta}\bm{\beta}$ block is $F_{aa}(a_n,\bm{\xi}_n)\,k^{-1}\sum_{K_n}\bm{x}_i\bm{x}_i^\top$; this matrix is negative definite by (A5) and (A2), and the joint Hessian converges to $-\bm{H}$ by weak convergence of the bounded design. The implicit-function construction therefore gives the unique local stationary point $\bar{\bm{\vartheta}}_n=(\alpha_0+a_n,\bm{b}_0,\bm{\xi}_n)$ of $M_n^{\rm tilt}$, with negative-definite Hessian.

By part (a), the difference $M_{n,\omega}^{\rm full}-M_n^{\rm tilt}$ is $o(1)$ in $C^2(N)$ and its gradient is bounded by $C\rho_n/\log\omega_n$. Uniform negative definiteness, Taylor's theorem and the inverse-function theorem give a unique stationary maximizer $\bm{\vartheta}_{n,\omega}^\circ$ within a fixed sufficiently small ball and the first bound in \eqref{eq:new_center_rates}. Curvature convergence \eqref{eq:new_full_hessian} follows from the $C^2$ control.

For the sharper slope bound, fix $\bm{\xi}$ near $\bm{\xi}^\dagger$. Since $F_{aa}(a^\dagger,\bm{\xi}^\dagger)<0$, there is a unique local root $a(\bm{\xi})$ of $F_a(a(\bm{\xi}),\bm{\xi})=0$. The entire clean $\bm{\beta}$ score vanishes at $\bm{\beta}^\flat(\bm{\xi})=(\alpha_0+a(\bm{\xi}),\bm{b}_0^\top)^\top$ for \emph{every} clean design array. Its $\bm{\beta}\bm{\beta}$ Hessian is $F_{aa}(a(\bm{\xi}),\bm{\xi})\,k^{-1}\sum \bm{x}_i\bm{x}_i^\top$, uniformly nonsingular by (A2). At the full local stationary point, the contamination $\bm{\beta}$ gradient is bounded by $C\rho_n/\omega_n$. Applying the multivariate mean-value identity to the clean $\bm{\beta}$ score between $\bm{\beta}^\flat(\bm{\xi}_{n,\omega}^\circ)$ and $\bm{\beta}_{n,\omega}^\circ$, and using a uniformly nonsingular averaged Jacobian, yields $\|\bm{\beta}_{n,\omega}^\circ-\bm{\beta}^\flat(\bm{\xi}_{n,\omega}^\circ)\| \le C\rho_n/\omega_n$. The slope of $\bm{\beta}^\flat$ is exactly $\bm{b}_0$, proving \eqref{eq:new_slope_bias}. Because $r\asymp n^{-1/3}(\log n)^{-2/3}$ and $\omega_n\ge n^{B_*}$ with $B_*>2$, this bound is $o(r)$.

\end{proof}

\begin{lemma}[Integrated spike suppression and smooth capping]
\label{lem:s3_spikes}
\medskip\noindent\textbf{(i) Exponential-width spike removal.} Let $C$ be a fixed compact neighborhood of $\bm{\vartheta}^\dagger$ contained in the interior of the nuisance parameter space. Put $\delta_n=\exp(-n^4)$ and
\[
\mathcal T_n
=\left\{\bm{\vartheta}\in C:
\min_{i\in K_n}|Y_i-\bm{x}_i^\top\bm{\beta}|
<\delta_n\right\}.
\]
There exist constants $c,C'>0$ such that, for all sufficiently large $n$, conditionally on the observations,
\begin{equation}
\label{eq:t2a_spike_integral}
\int_{\mathcal T_n}
\prod_{i\in K_n}q_{\bm{\vartheta}}(Y_i\mid \bm{x}_i)
\pi(d\bm{\vartheta})
\le\exp\{-c n^3+C'n\log n\}.
\end{equation}
Consequently, $\widetilde\Pi_n(\mathcal T_n)\to0$ in $P_0$-probability.

\medskip\noindent\textbf{(ii) Smooth removal of the singular spikes.} Choose $A_n^{\mathrm{cap}}=C_0(1+n^4)$ sufficiently large that $q_{\bm{\vartheta}}(y\mid\bm{x})\le A_n^{\mathrm{cap}}/e$ for $\bm{\vartheta}\in C$ whenever $|y-\bm{x}^\top\bm{\beta}|\ge\delta_n$. Set $\Lambda_n^{\mathrm{cap}}=\log A_n^{\mathrm{cap}}$ and choose a fixed $C^\infty$ nondecreasing $\psi$ with $\psi(v)=v$ for $v\le-1$, constant for $v\ge1$, and $0\le\psi'(v)\le1$. Define
\begin{equation}
\label{eq:t2a_smoothcap}
 \ell_{n,\bm{\vartheta}}^{\rm sc}(y\mid\bm{x})
 =\Lambda_n^{\mathrm{cap}}+\psi\{\log q_{\bm{\vartheta}}(y\mid\bm{x})-\Lambda_n^{\mathrm{cap}}\},
 \qquad
 q_{n,\bm{\vartheta}}^{\rm sc}
 =\exp\{\ell_{n,\bm{\vartheta}}^{\rm sc}\}.
\end{equation}
At a fitted residual of zero extend $\ell_{n,\bm{\vartheta}}^{\rm sc}$ by its finite constant limiting value. On $C$ the original and smoothed integrated posteriors differ by $o_{P_0}(1)$ in total variation. The nuisance scores and Hessians of the smoothly capped log likelihood exist continuously, including on the transition region, and are uniformly bounded on $C$ independently of $n$.
\end{lemma}
\begin{proof}
\smallskip\noindent\emph{Part (i).} For $\bm{\vartheta}\in C$, the central logarithmic singularity and boundedness of the density away from zero imply
\[
q_{\bm{\vartheta}}(Y_i\mid \bm{x}_i)
\le C_0\left\{1+
\log_+\frac1{|Y_i-\bm{x}_i^\top\bm{\beta}|}\right\}
=:C_0 w_i(\bm{\beta}).
\]
Write $k=|K_n|$. The intercept has a Gaussian conditional distribution given the slopes, with variance bounded away from zero. Uniformly in the observed responses and design,
\[
\int w_i(\bm{\beta})^k\,\pi_{\bm{\beta}}(d\bm{\beta})
\le (C_1 k)^k.
\]
For a distinguished observation $j$, the same one-dimensional integral, restricted to $|Y_j-\bm{x}_j^\top\bm{\beta}|<\delta_n$, satisfies
\[
\int w_j(\bm{\beta})^k
\mathbf1\{|Y_j-\bm{x}_j^\top\bm{\beta}|<\delta_n\}
\pi_{\bm{\beta}}(d\bm{\beta})
\le C_2\delta_n
\{1+\log(1/\delta_n)\}^k.
\]
Apply H\"older's inequality with all $k$ exponents equal to $k$ and sum over the at most $k$ choices of $j$. The integral is at most
\[
k C_0^k(C_1 k)^{k-1}
(C_2\delta_n)^{1/k}
\{1+\log(1/\delta_n)\}.
\]
Since $k\asymp n$ and $\delta_n^{1/k}=\exp(-n^4/k)$, this is bounded by $\exp\{-c n^3+C'n\log n\}$. The factor $s^{m_n}\le1$ can only decrease the numerator. The tilted normalizer admits a lower bound $\exp(-C_3 n)$ with probability approaching one, by integration over a fixed positive-prior neighborhood of $\bm{\vartheta}^\dagger$ and a uniform law of large numbers. Dividing proves the conclusion.

\smallskip\noindent\emph{Part (ii).} The function $\psi$ can be constructed by integrating a smooth nonincreasing cutoff from one to zero on $[-1,1]$. By construction $\ell^{\rm sc}=\log q$ when $\log q\le \Lambda_n^{\mathrm{cap}}-1$, and $\ell^{\rm sc}\le\log q$ everywhere. For $|y-\bm{x}^\top\bm{\beta}|\ge\delta_n$, $\log q\le \Lambda_n^{\mathrm{cap}}-1$, so the difference is supported on $\mathcal T_n$. Hence
\[
 0\le \int_C\{L_{K,n}-L_{K,n}^{\rm sc}\}
 s^m\pi(d\bm{\vartheta})
 \le \int_{\mathcal T_n}L_{K,n}\pi(d\bm{\vartheta}),
\]
and the latter is negligible by the spike assertion in Lemma~\ref{lem:s3_spikes}(i) relative to the tilted normalizer. The usual comparison between two nonnegative integrable densities proves total-variation convergence. For any nuisance coordinate $\xi_j$,
\[
 \partial_{\xi_j}\ell^{\rm sc}
 =\psi'(\log q-\Lambda_n^{\mathrm{cap}})\partial_{\xi_j}\log q,
\]
\[
 \partial_{\xi_j\xi_l}^2\ell^{\rm sc}
 =\psi'\partial_{\xi_j\xi_l}^2\log q
 +\psi''(\partial_{\xi_j}\log q)
 (\partial_{\xi_l}\log q).
\]
The uncapped nuisance scores and Hessians are bounded uniformly over the compact interior nuisance set; this follows from the two-sided central expansion and differentiated LPAL tail estimates. Since the first two derivatives of $\psi$ are bounded, the displayed derivatives are bounded. Smooth extension at residual zero follows because $\psi$ is constant in the upper-cap region.
\end{proof}

\begin{lemma}[Uniform nuisance derivatives and empirical-process bounds]
\label{lem:s3_nuisance}
The following uniformity is always on a fixed interior nuisance compact $\Xi$, never on the unrestricted nuisance parameter space.

\medskip\noindent\textbf{(i) Uniform expansion of all required Laplace derivatives.} For fixed $\gamma>0$, put
\[
 K(a)=\int_0^\infty v^{-1}e^{-a/v}H(v;\gamma)\,dv,
 \qquad D=a\frac{d}{da},\qquad b=1+\gamma.
\]
For every fixed integer $J\ge1$, uniformly for $j\in\{0,\dots,J\}$, as $a\to\infty$,
\begin{equation}
\label{eq:t2_K_full_tail_v11}
 (-1)^j a^j K^{(j)}(a)
 =\frac{\gamma j!}{a(\log a)^b}
 \left[1+\frac{b(H_j-\gamma_{\rm E}-1)}{\log a}
 +O_J\{(\log a)^{-2}\}\right],
\end{equation}
where $H_0=0$, $H_j=\sum_{l=1}^jl^{-1}$, and $\gamma_{\rm E}$ is Euler's constant. In particular,
\begin{align}
\label{eq:t2_K_D_tail_v11}
 D\log K(a)&=-1-\frac b{\log a}
                   +O\{(\log a)^{-2}\},\\
 D^j\log K(a)&=O_J\{(\log a)^{-2}\},
                   \qquad 2\le j\le J.
\end{align}
As $a\downarrow0$, for each fixed $J$,
\begin{align}
\label{eq:t2_K_D_origin_v11}
 K(a)&=\gamma\log(1/a)+C_\gamma
                         +O\{a\log(1/a)\},\\
 DK(a)&=-\gamma+O\{a\log(1/a)\},\\
 D^jK(a)&=O_J\{a\log(1/a)\},
                         \qquad 2\le j\le J.
\end{align}
All constants depend on $J,\gamma$ only; none depends on the mixture weight, scale, sample size or cap index.

\medskip\noindent\textbf{(ii) Compact-uniform nuisance derivatives and variation.} Let $\Xi=\{(\lambda,\zeta): \underline\sigma\le e^\lambda\le\overline\sigma, \underline s\le(1+e^{-\zeta})^{-1}\le\overline s\}$ with $0<\underline\sigma<\overline\sigma<\infty$ and $0<\underline s<\overline s<1$. For each fixed integer $h\ge1$, all mixed nuisance partial derivatives up to order $h$ of
\[
 \ell(v,\lambda,\zeta)
 =\log(c_\tau)-\lambda
 +\log\{(1-s)e^{-a}+sK(a)\},
 \quad a=\rho_\tau(v)e^{-\lambda},
\]
are uniformly bounded on $v\ne0,\ \bm{\xi}\in\Xi$. For the score and Hessian profiles, their residual-direction total variations are uniformly bounded. In particular, with one constant $C_{\Xi,\tau,\gamma}$ independent of $n$,
\begin{equation}
\label{eq:t2_compact_BV_v11}
 \sup_{\bm{\xi}\in\Xi}\sum_{1\le|\alpha|\le2}
 \bigl(\|\partial_{\bm{\xi}}^\alpha\ell\|_\infty
 +\operatorname{TV}_v(\partial_{\bm{\xi}}^\alpha\ell)\bigr)
 \le C_{\Xi,\tau,\gamma}.
\end{equation}
The same estimate, and the uniform boundedness of third nuisance derivatives, holds after applying any smooth cap $L+\psi(\ell-L)$ with one fixed $C^\infty$ cutoff $\psi$ satisfying bounded derivatives and $0\le\psi'\le1$.

\medskip\noindent\textbf{(iii) Uniform nuisance empirical-process bounds.} Let $B$ and $\Xi$ be sufficiently small fixed compact interior neighborhoods of $\bm{\beta}^\dagger$ and $\bm{\xi}^\dagger$. Under (A1), (A2), and (A5), let $\dot\ell_{\bm{\xi},i}^{\rm sc}$ and $\ddot\ell_{\bm{\xi}\bm{\xi},i}^{\rm sc}$ denote the nuisance score and Hessian of the smoothly capped single-observation log likelihood. Then
\begin{align}
 \label{eq:t2a_nuisance_score_bound}
 &\sup_{\bm{\beta}\in B,\bm{\xi}\in\Xi}\left\|
 \sum_{i\in K_n}\{\dot\ell_{\bm{\xi},i}^{\rm sc}
 -E_0\dot\ell_{\bm{\xi},i}^{\rm sc}\}\right\|
 =O_{P_0}(\sqrt k),\\
 \label{eq:t2a_nuisance_hessian_bound}
 &\sup_{\bm{\beta}\in B,\bm{\xi}\in\Xi}\left\|
 \sum_{i\in K_n}\{\ddot\ell_{\bm{\xi}\bm{\xi},i}^{\rm sc}
 -E_0\ddot\ell_{\bm{\xi}\bm{\xi},i}^{\rm sc}\}\right\|
 =O_{P_0}(\sqrt k).
\end{align}
Uniformly on $B\times\Xi$, the normalized expectations of these capped derivatives differ from the derivatives of $\mathcal M_n$ by $o(1)$; the difference is, in fact, exponentially small in $n$.
\end{lemma}
\begin{proof}
\smallskip\noindent\emph{Part (i).} For $j\ge0$, differentiation under the integral and the change of variables $u=a/v$ give the
\emph{exact} identity
\begin{equation}
\label{eq:t2_K_exact_j_v11}
 (-1)^ja^jK^{(j)}(a)
 =\int_0^\infty u^{j-1}e^{-u}H(a/u;\gamma)\,du.
\end{equation}
This also holds for $j=0$, with $u^{-1}$. Write $L=\log a$. From the explicit $H$,
\[
 H(v;\gamma)=\frac{\gamma}{v(\log v)^b}
 \left[1-\frac{b}{\log v}
 +O\{(\log v)^{-2}+v^{-1}\}\right]
\]
uniformly as $v\to\infty$. Split \eqref{eq:t2_K_exact_j_v11} into $u<a^{-1/4}$, $a^{-1/4}\le u\le a^{1/4}$ and $u>a^{1/4}$. The first range is $O(a^{-1-(j+1)/4}L^{-b})$, and the third is $O(e^{-a^{1/4}/2})$, by boundedness of $H$. On the middle range, $|\log u|/L\le1/4$ and
\[
 H(a/u;\gamma)
 =\frac{\gamma u}{aL^b}
 \left[1+\frac{b(\log u-1)}{L}
 +O\left(\frac{1+(\log u)^2}{L^2}
            +\frac{u}{a}\right)\right].
\]
For each $j\le J$, the $u^je^{-u}$-integrals of the remainder are $O_J(L^{-2})$ times the leading term. The identities
\[
 \int_0^\infty u^je^{-u}\,du=j!,\qquad
 \int_0^\infty u^je^{-u}\log u\,du
 =j!(H_j-\gamma_{\rm E})
\]
yield \eqref{eq:t2_K_full_tail_v11} with a uniform constant for the finite collection of derivatives. Dividing the $j$th expansion by the $j=0$ expansion gives
\[
 \frac{(-1)^ja^jK^{(j)}(a)}{K(a)}
 =j!\left\{1+\frac{bH_j}{L}+O_J(L^{-2})\right\}.
\]
The expressions for $D^j\log K$ are universal polynomials in these derivative ratios; the constant and $L^{-1}$ coefficients cancel for $j\ge2$. Equivalently, these coefficients coincide with differentiating the explicit leading function $-\log a-b\log\log a$. This proves \eqref{eq:t2_K_D_tail_v11} without differentiation of an uncontrolled $o(1)$ remainder.

At the origin, use $H(v)=\gamma+O(v)$ for $v\le1$ and $E_1(a)=-\log a-\gamma_{\rm E} +O(a)$. This gives the expansion of $K$. For $j\ge1$, apply the integral representation of $K^{(j)}$ on $(0,1)$: the constant-$\gamma$ part gives $(-1)^j\gamma(j-1)!a^{-j}$, while the $O(v)$ remainder is $O(a^{1-j})$ for $j\ge2$ and $O(\log(1/a))$ for $j=1$. Thus $a^jK^{(j)}=(-1)^j\gamma(j-1)!+ O_J(a\log(1/a))$. The combination $D^jK=\sum_{l=1}^j S(j,l)a^lK^{(l)}$, where $S(j,l)$ are Stirling numbers of the second kind, has a vanishing constant term for $j\ge2$; the case $j=1$ has limit $-\gamma$. This proves the remaining expansions. All estimates used fixed integrable dominating functions, with constants depending only on the stated fixed parameters.

\smallskip\noindent\emph{Part (ii).} Let $P_s(a)=sK(a)/A_s(a)$ and $R_s(a)=D\log A_s(a)$. The five nuisance score and Hessian profiles are
\[
 h_\lambda=-1-R_s,\quad h_\zeta=P_s-s,\quad
 h_{\lambda\lambda}=DR_s,\quad
 h_{\lambda\zeta}=-DP_s,\quad
 h_{\zeta\zeta}=P_s(1-P_s)-s(1-s).
\]
Their residual derivatives satisfy $\partial_a=a^{-1}D$. For $a\downarrow0$, put $L_0=\log(1/a)$. the derivative expansion in Lemma~\ref{lem:s3_nuisance}(i) implies, uniformly for $s\in[\underline s,\overline s]$,
\[
 A_s(a)=s\gamma L_0+O(1),\quad
 DA_s(a)=-s\gamma+O(aL_0),\quad
 D^jA_s(a)=O_j(aL_0),\quad j\ge2.
\]
Hence, for any fixed $j\ge2$,
\[
 D^j\log A_s(a)=O_j(L_0^{-j}+a),\qquad
 D^jP_s(a)=O_j(L_0^{-j-1}+a),
\]
while $DP_s(a)=O(L_0^{-2}+a)$ and $P_s(a)=1+O(L_0^{-1})$. Taking one $D$ derivative of the displayed score/Hessian formulas therefore yields
\begin{equation}
\label{eq:t2_small_Var_v11}
 \sum_{h\in\{h_\lambda,h_\zeta,h_{\lambda\lambda},
 h_{\lambda\zeta},h_{\zeta\zeta}\}}
 |\partial_ah(a)|
 \le\frac{C}{a\{\log(1/a)\}^2},
 \quad 0<a<a_0.
\end{equation}

For $a\to\infty$, put $L_\infty=\log a$. The LPAL term dominates the AL term uniformly in $s\in[\underline s,\overline s]$ because $e^{-a}/K(a)=O\{aL_\infty^{1+\gamma}e^{-a}\}$. the derivative expansion in Lemma~\ref{lem:s3_nuisance}(i) gives
\[
 R_s(a)=-1-\frac{1+\gamma}{L_\infty}
           +O(L_\infty^{-2}),\quad
 D^jR_s(a)=O_j(L_\infty^{-2}),\quad j\ge1,
\]
and $D^jP_s(a)=O_j(a^{j+1}e^{-a}L_\infty^{1+\gamma})$ for fixed $j\ge1$. Thus
\begin{equation}
\label{eq:t2_large_Var_v11}
 \sum_{h\in\{h_\lambda,h_\zeta,h_{\lambda\lambda},
 h_{\lambda\zeta},h_{\zeta\zeta}\}}
 |\partial_ah(a)|
 \le\frac{C}{a(\log a)^2},
 \quad a>a_0^{-1}.
\end{equation}
Since both majorants are integrable and the profiles are jointly smooth on the intervening compact $a$-interval, the total-variation bound follows. The same expansions through any fixed order yield uniform sup-norm bounds on the required third (or higher fixed-order) nuisance derivatives.

Finally, the uncapped log density decreases strictly as a function of $a$ on either residual half-line. For each derivative order $j\le3$, $\psi^{(j)}(\ell-L)$ has variation no larger than twice the variation of $\psi^{(j)}$, independently of $L$. The score/Hessian of the capped log density are finite sums of products of bounded BV profiles with these cutoff terms. The BV product inequality proves \eqref{eq:t2_compact_BV_v11} for the capped class, with exactly the same dependence of its constants.

\smallskip\noindent\emph{Part (iii).} We give the bracketing calculation to make the constants independent of the triangular-array index $n$. Let $h_{\bm{\xi}}$ be any scalar component of a nuisance score or Hessian, capped or uncapped. By the variation bound in Lemma~\ref{lem:s3_nuisance}(ii),
\[
 \|h_{\bm{\xi}}\|_\infty\le C,
 \quad\operatorname{TV}(h_{\bm{\xi}})\le C,
 \quad\sup_v\|\partial_{\bm{\xi}} h_{\bm{\xi}}(v)\|\le C.
\]
A location change $a\mapsto a+u$ with $|u|\le\delta$ produces the pointwise bracket
\[
 \inf_{|u|\le\delta}h_{\bm{\xi}}(v-a-u)
 \le h_{\bm{\xi}}(v-a-u)
 \le\sup_{|u|\le\delta}h_{\bm{\xi}}(v-a-u).
\]
Its width is bounded by the variation measure of $h_{\bm{\xi}}$ on $[v-a-\delta,v-a+\delta]$. Since the true error density is bounded, Fubini gives
\[
 E_0\Bigl[
 \sup_{|u|\le\delta}h_{\bm{\xi}}(U-a-u)
 -\inf_{|u|\le\delta}h_{\bm{\xi}}(U-a-u)
 \Bigr]
 \le 2\delta\|g_0\|_\infty
 \operatorname{TV}(h_{\bm{\xi}}).
\]
The bracket width is at most $2C$, hence its $L^2(P_0)$ norm is bounded by $C\sqrt\delta$. Moving $\bm{\xi}$ inside a Euclidean cell of diameter $\delta$ adds at most $C\delta$ pointwise. Because $\sup_i\|\bm{x}_i\|<\infty$, a parameter cell of mesh $\delta$ induces a residual shift at most $C\delta$, uniformly in the deterministic design.

A Euclidean grid on $B\times\Xi$ has at most $C\delta^{-(p+d_\xi)}$ cells. Applying the above brackets in each cell, with possibly different endpoints for different $i$ but the same uniform width bound, gives the triangular-array bracketing estimate
\begin{equation}
 \label{eq:t2a_entropy_polynomial}
 N_{[]}\bigl(\varepsilon,\mathcal H_n,
 \overline L_2(P_{0,n})\bigr)
 \le C\varepsilon^{-2(p+d_\xi)},
 \qquad 0<\varepsilon<1,
\end{equation}
where $\overline L_2(P_{0,n})^2=k^{-1}\sum_{i\in K_n}P_0(\cdot)^2$ and $\mathcal H_n$ denotes any nuisance score or Hessian component class. Therefore
\[
 \sup_n\int_0^1
 \sqrt{1+\log N_{[]} (\varepsilon,\mathcal H_n,
 \overline L_2(P_{0,n}))}\,d\varepsilon<\infty.
\]
The bounded envelopes, independence and this finite entropy integral give, by symmetrization and the bracketing maximal inequality for independent triangular arrays \citep{vanDerVaart1996},
\[
 E_0\sup_{h\in\mathcal H_n}
 \left|\sum_{i\in K_n}\{h_i(U_i)-E_0h_i(U_i)\}\right|
 \le C\sqrt{k}.
\]
This proves both \eqref{eq:t2a_nuisance_score_bound} and \eqref{eq:t2a_nuisance_hessian_bound}.

Finally, the capped derivatives differ from the uncapped ones only if the fitted residual lies in an interval of width at most $2\delta_n$ around zero. The derivative differences are uniformly bounded; $\|g_0\|_\infty<\infty$ therefore gives, uniformly in $\bm{\beta},\bm{\xi}$,
\[
 k^{-1}\sum_{i\in K_n}E_0
 \|\partial_{\bm{\xi}}^j\ell_{n,i}^{\rm sc}
 -\partial_{\bm{\xi}}^j\log q_{\bm{\vartheta}}(Y_i\mid \bm{x}_i)\|
 \le C\delta_n,\qquad j=1,2.
\]
Differentiation under the expectation for the uncapped nuisance derivatives is valid by their bounded envelopes. This proves the final assertion.
\end{proof}

\begin{lemma}[Uniform full-data nuisance profiling]
\label{lem:s3_profile}
On sufficiently small fixed interior compacts $B$ and $\Xi$, write $\mathfrak m_{n,\omega}(\bm{\beta})=\sup_{\bm{\xi}\in\Xi}M_{n,\omega}^{\rm full}(\bm{\beta},\bm{\xi})$ and
\begin{equation}\label{eq:full_profile_definition}
\bar\ell_{n,\omega}^{\rm sc}(\bm{\beta})=\log\int_\Xi\exp\left\{\sum_{i\in K_n}\ell_{n,(\bm{\beta},\bm{\xi})}^{\rm sc}(Y_i\mid \bm{x}_i)+\sum_{i\in L_n}\log q_{(\bm{\beta},\bm{\xi})}(Y_i\mid \bm{x}_i)\right\}\pi_{\bm{\xi}}(\bm{\xi})\,d\bm{\xi}.
\end{equation}
The cap $\ell^{\rm sc}$ is defined in Lemma~\ref{lem:s3_spikes}(ii). Put $\bm{\beta}_c=\bm{\beta}_{n,\omega}^\circ$, $\bm{\xi}_c=\bm{\xi}_{n,\omega}^\circ$, and $S_{n,\bm{\xi}}(\bm{\beta})=\sum_{i\in K_n}\{\ell_{n,(\bm{\beta},\bm{\xi})}^{\rm sc}(Y_i\mid \bm{x}_i)-E_0\ell_{n,(\bm{\beta},\bm{\xi})}^{\rm sc}(Y_i\mid \bm{x}_i)\}$. There exists a single nonnegative $\mathfrak R_n=O_{P_0}(1)$, independent of $d$, such that simultaneously for $r_k\le d\le\delta_0$,
\begin{equation}\label{eq:new_full_profile_comparison}
\begin{aligned}
\sup_{\|\bm{\beta}-\bm{\beta}_c\|\le d}\bigg|&\bar\ell_{n,\omega}^{\rm sc}(\bm{\beta})-\bar\ell_{n,\omega}^{\rm sc}(\bm{\beta}_c)-k\{\mathfrak m_{n,\omega}(\bm{\beta})-\mathfrak m_{n,\omega}(\bm{\beta}_c)\}\\
&-\{S_{n,\bm{\xi}_c}(\bm{\beta})-S_{n,\bm{\xi}_c}(\bm{\beta}_c)\}\bigg|\le \mathfrak R_n(\sqrt{k}d+1).
\end{aligned}
\end{equation}
The compacts can be chosen so that $\mathfrak m_{n,\omega}(\bm{\beta})-\mathfrak m_{n,\omega}(\bm{\beta}_c)\le-c\|\bm{\beta}-\bm{\beta}_c\|^2$ on $B$, for some $c>0$ independent of $n$.
\end{lemma}
\begin{proof}
\emph{Population nuisance profile.} By Lemma~\ref{lem:new_moving_center} and the local curvature assumption (A5), choose fixed convex compacts $B$ and $\Xi$, with $\bm{\vartheta}_{n,\omega}^\circ\in\operatorname{int}(B\times\Xi)$ eventually, such that for fixed constants $0<c<C<\infty$,
\begin{equation}\label{eq:s3_profile_hessian_bound}
cI\preceq-\partial_{\bm{\xi}\bm{\xi}}^2M_{n,\omega}^{\rm full}(\bm{\beta},\bm{\xi})\preceq C\bm{I},\qquad(\bm{\beta},\bm{\xi})\in B\times\Xi.
\end{equation}
Shrinking $B$ if necessary, the unique population nuisance maximizer $\bm{\xi}_n^0(\bm{\beta})=\arg\max_{\bm{\xi}\in\Xi}M_{n,\omega}^{\rm full}(\bm{\beta},\bm{\xi})$ lies in a fixed interior subcompact of $\Xi$, uniformly in $\bm{\beta}\in B$. In particular, $\bm{\xi}_n^0(\bm{\beta}_c)=\bm{\xi}_c$. The implicit-function theorem and the uniformly bounded mixed population Hessian imply
\begin{equation}\label{eq:s3_profile_lipschitz}
\|\bm{\xi}_n^0(\bm{\beta})-\bm{\xi}_c\|\le C\|\bm{\beta}-\bm{\beta}_c\|,\qquad\bm{\beta}\in B.
\end{equation}
Because the full Hessian in $(\bm{\beta},\bm{\xi})$ is uniformly negative definite near $\bm{\vartheta}_{n,\omega}^\circ$, the Hessian of $\mathfrak m_{n,\omega}$ is the corresponding Schur complement and is uniformly negative definite. Its gradient vanishes at $\bm{\beta}_c$, proving the final quadratic assertion after another fixed reduction of $B$.

\emph{Uniform score control and the empirical nuisance maximizer.} Write
\begin{align*}
Q_n(\bm{\beta},\bm{\xi})&=\sum_{i\in K_n}\ell_{n,(\bm{\beta},\bm{\xi})}^{\rm sc}(Y_i\mid \bm{x}_i)+\sum_{i\in L_n}\log q_{(\bm{\beta},\bm{\xi})}(Y_i\mid \bm{x}_i),\\
c_n&=\sum_{i\in L_n}\log f_1(Y_i),\qquad\qquad S_n(\bm{\beta},\bm{\xi})=S_{n,\bm{\xi}}(\bm{\beta}).
\end{align*}
The bounded clean density and the exponentially small cap width give, uniformly in $(\bm{\beta},\bm{\xi})\in B\times\Xi$, $\sum_{i\in K_n}\big|E_0\ell_{n,(\bm{\beta},\bm{\xi})}^{\rm sc}-E_0\log q_{(\bm{\beta},\bm{\xi})}\big|=o(1)$, and the analogous sums of first and second nuisance derivative differences are $o(1)$ by Lemma~\ref{lem:s3_nuisance}(iii). Consequently, with the expectation interpreted observationwise,
\begin{equation}\label{eq:s3_profile_decomp}
Q_n(\bm{\beta},\bm{\xi})=c_n+kM_{n,\omega}^{\rm full}(\bm{\beta},\bm{\xi})+S_n(\bm{\beta},\bm{\xi})+e_n(\bm{\beta},\bm{\xi}),\qquad\sup_{B\times\Xi}|e_n|=o(1).
\end{equation}
By Lemma~\ref{lem:s3_nuisance}(iii), there is $V_n^{\mathrm{prof}}=O_{P_0}(1)$ such that
\begin{equation}\label{eq:s3_profile_empirical_derivative}
\sup_{B\times\Xi}\left\{\|\partial_{\bm{\xi}} S_n(\bm{\beta},\bm{\xi})\|+\|\partial_{\bm{\xi}\bm{\xi}}^2S_n(\bm{\beta},\bm{\xi})\|\right\}\le V_n^{\mathrm{prof}}\sqrt{k}.
\end{equation}
The outliers contribute no empirical fluctuation because their values are deterministic in the contamination experiment. The uniform nuisance-derivative bounds in Lemma~\ref{lem:s3_nuisance}(ii), together with (\ref{eq:s3_profile_hessian_bound}) and (\ref{eq:s3_profile_empirical_derivative}), show that with probability tending to one the sample criterion $Q_n(\bm{\beta},\cdot)$ is uniformly strongly concave, with
\begin{equation}\label{eq:s3_profile_sample_hessian}
ckI/2\preceq-\partial_{\bm{\xi}\bm{\xi}}^2Q_n(\bm{\beta},\bm{\xi})\preceq C'kI,\qquad(\bm{\beta},\bm{\xi})\in B\times\Xi.
\end{equation}
For every $\bm{\beta}\in B$ its maximizer $\widehat{\bm{\xi}}_n(\bm{\beta})$ is therefore unique and lies in the same fixed interior subcompact of $\Xi$, with probability tending to one. At $\bm{\xi}_n^0(\bm{\beta})$, the population nuisance score is zero, while (\ref{eq:s3_profile_empirical_derivative}) bounds the sample score by $V_n^{\mathrm{prof}}\sqrt{k}+o(1)$. The mean-value identity and (\ref{eq:s3_profile_sample_hessian}) then give, simultaneously over $\bm{\beta}\in B$,
\begin{equation}\label{eq:s3_profile_argmax_shift}
\|\widehat{\bm{\xi}}_n(\bm{\beta})-\bm{\xi}_n^0(\bm{\beta})\|\le C(V_n^{\mathrm{prof}}+1)k^{-1/2},\qquad 0\le Q_n\{\bm{\beta},\widehat{\bm{\xi}}_n(\bm{\beta})\}-Q_n\{\bm{\beta},\bm{\xi}_n^0(\bm{\beta})\}\le C(V_n^{\mathrm{prof}}+1)^2.
\end{equation}

\emph{Uniform Laplace sandwich.} The transformed nuisance prior $\pi_{\bm{\xi}}$ has a continuous density bounded above and away from zero on $\Xi$. By (\ref{eq:s3_profile_sample_hessian}), Taylor's theorem at $\widehat{\bm{\xi}}_n(\bm{\beta})$ bounds $Q_n(\bm{\beta},\bm{\xi})-Q_n\{\bm{\beta},\widehat{\bm{\xi}}_n(\bm{\beta})\}$ between two negative quadratic forms with fixed coefficients times $k$. The upper Gaussian integral is over all of $\mathbb R^{d_\xi}$, while the lower integral is restricted to a radius-$k^{-1/2}$ ball contained in $\Xi$ uniformly in $\bm{\beta}$. Thus, for constants $0<c_1<C_1<\infty$ independent of $n$ and $\bm{\beta}$,
\begin{equation}\label{eq:s3_profile_gaussian_sandwich}
c_1 k^{-d_\xi/2}e^{Q_n\{\bm{\beta},\widehat{\bm{\xi}}_n(\bm{\beta})\}}\le\int_\Xi e^{Q_n(\bm{\beta},\bm{\xi})}\pi_{\bm{\xi}}(\bm{\xi})\,d\bm{\xi}\le C_1 k^{-d_\xi/2}e^{Q_n\{\bm{\beta},\widehat{\bm{\xi}}_n(\bm{\beta})\}}.
\end{equation}
Combining (\ref{eq:s3_profile_decomp}), (\ref{eq:s3_profile_argmax_shift}), and (\ref{eq:s3_profile_gaussian_sandwich}) yields one random $U_n^{\mathrm{Lap}}\le C\{1+(V_n^{\mathrm{prof}}+1)^2\}=O_{P_0}(1)$ such that, simultaneously for all $\bm{\beta}\in B$,
\begin{equation}\label{eq:s3_profile_uniform_laplace}
\left|\bar\ell_{n,\omega}^{\rm sc}(\bm{\beta})-c_n-k\mathfrak m_{n,\omega}(\bm{\beta})-S_n\{\bm{\beta},\bm{\xi}_n^0(\bm{\beta})\}+\frac{d_\xi}{2}\log k\right|\le U_n^{\mathrm{Lap}}.
\end{equation}
This is a single uniform Laplace remainder, rather than separate pointwise $O_{P_0}(1)$ terms at different values of $\bm{\beta}$.

\emph{Comparison at the common nuisance center.} By the fundamental theorem of calculus in $\bm{\xi}$, (\ref{eq:s3_profile_lipschitz}) and (\ref{eq:s3_profile_empirical_derivative}) give
\begin{equation}\label{eq:s3_profile_empirical_transfer}
\sup_{\|\bm{\beta}-\bm{\beta}_c\|\le d}\big|S_n\{\bm{\beta},\bm{\xi}_n^0(\bm{\beta})\}-S_n(\bm{\beta},\bm{\xi}_c)\big|\le CV_n^{\mathrm{prof}}\sqrt{k}\,d.
\end{equation}
At $\bm{\beta}=\bm{\beta}_c$ the two nuisance values coincide. Subtracting (\ref{eq:s3_profile_uniform_laplace}) at $\bm{\beta}$ and $\bm{\beta}_c$ cancels $c_n-(d_\xi/2)\log k$ exactly. Equations (\ref{eq:s3_profile_uniform_laplace}) and (\ref{eq:s3_profile_empirical_transfer}) prove (\ref{eq:new_full_profile_comparison}) with $\mathfrak R_n=2U_n^{\mathrm{Lap}}+CV_n^{\mathrm{prof}}+1=O_{P_0}(1)$, simultaneously for every $r_k\le d\le\delta_0$.
\end{proof}

\subsection{Local posterior contraction and quantile correction}

\begin{lemma}[Uniform local likelihood increments]
\label{lem:new_uniform_maximal}
Fix a sufficiently small compact interior location ball $B$ and nuisance compact $\Xi$, and a fixed enlargement of $B$ on which the analytic bounds hold. Uniformly over every deterministic $\bm{\beta}_c\in B$, $\bm{\xi}\in\Xi$, and $\zeta_n\le d\le\delta_0$, with $\zeta_n=n^{-(\chi+3)}$ and fixed $\chi>p+5$,
\begin{equation}\label{eq:new_uniform_maximal}
E_0\sup_{\|\bm{\beta}-\bm{\beta}_c\|\le d}\left|S_{n,\bm{\xi}}(\bm{\beta})-S_{n,\bm{\xi}}(\bm{\beta}_c)\right|\le C\frac{\sqrt{kd}}{\log(e/d)}+C(\log n)^3.
\end{equation}
The constants do not depend on $n$, $\bm{\beta}_c$, $d$, or $\bm{\xi}$.
\end{lemma}
\begin{proof}
\emph{Analytic $L^2$ increment bound.} On a fixed compact interior nuisance set, the central expansion of the AL--LPAL density and its differentiated form give, for sufficiently small $|v|$,
\begin{equation}\label{eq:s3_max_derivative_bound}
\left|\frac{\partial}{\partial v}\log\left\{\sigma^{-1}f_s(v/\sigma)\right\}\right|\le\frac{C}{|v|\log(e/|v|)},
\end{equation}
uniformly in $\bm{\xi}$. Away from the central region, the log-density derivative is bounded and is $O(1/|v|)$ in either tail, uniformly in $\bm{\xi}$. We first establish the deterministic translation modulus for the uncapped log-density $\phi_{\bm{\xi}}(v)=\log\{\sigma^{-1}f_s(v/\sigma)\}$:
\begin{equation}\label{eq:s3_max_translation_modulus}
\int_{\mathbb R}|\phi_{\bm{\xi}}(v+t)-\phi_{\bm{\xi}}(v)|^2\,dv\le C\frac{|t|}{\log^2(e/|t|)},\qquad0<|t|\le\delta_0.
\end{equation}
To prove it, put $h=|t|$ and $L_h=\log(e/h)$. On the set $\min\{|v|,|v+t|\}\le 4h$, compare both log densities to the values at residuals $\pm h$. The central derivative bound (\ref{eq:s3_max_derivative_bound}) and the matching leading constants on the two sides of the origin give
\[
|\phi_{\bm{\xi}}(v+t)-\phi_{\bm{\xi}}(v)|\le\frac{C}{L_h}\left\{1+\log_+\frac{h}{|v|}+\log_+\frac{h}{|v+t|}\right\}.
\]
The squared right-hand side is integrable over a set of length $O(h)$, and its integral is at most $Ch/L_h^2$ by the substitution $v=hu$. On the remaining central region, where both $|v|$ and $|v+t|$ exceed $4h$ but are smaller than a fixed cutoff, the mean-value theorem and (\ref{eq:s3_max_derivative_bound}) give a squared-integral bound of the form
\[
Ch^2\int_{4h}^{c_0}\frac{dv}{v^2\log^2(e/v)}\le C\frac{h}{L_h^2}.
\]
Outside the central cutoff, the $O(1)$ derivative on bounded sets and the $O(1/|v|)$ derivative in the tails give an $O(h^2)$ contribution. This proves (\ref{eq:s3_max_translation_modulus}). The cap is a scalar $1$-Lipschitz transform of $\phi_{\bm{\xi}}$, so the same modulus holds for the capped log density. Since the clean error density is bounded and $|\bm{x}_i^\top(\bm{\beta}-\bm{\beta}')|\le C\|\bm{\beta}-\bm{\beta}'\|$, (\ref{eq:s3_max_translation_modulus}) implies, uniformly in the deterministic design and in $\bm{\xi}\in\Xi$,
\begin{equation}\label{eq:s3_max_single_increment}
E_0\left[\ell_{n,(\bm{\beta},\bm{\xi})}^{\rm sc}(Y_i\mid \bm{x}_i)-\ell_{n,(\bm{\beta}',\bm{\xi})}^{\rm sc}(Y_i\mid \bm{x}_i)\right]^2\le C\frac{h}{\log^2(e/h)},\qquad h=\|\bm{\beta}-\bm{\beta}'\|\le\delta_0.
\end{equation}

\emph{Bernstein chaining on finite nets.} For $\bm{\beta},\bm{\beta}'$ in the fixed compact enlargement, the capped single-observation log-likelihood difference is bounded by $C\log n$, uniformly over the observed response. Indeed, on bounded residual sets the cap supplies the upper bound $O(\log n)$ and the log density is bounded below, whereas on the tails the log-density ratio between boundedly separated locations is uniformly bounded by the analytic tail derivative bound. Hence every centered increment is bounded by $C\log n$. Set $\mathbb G_n(\bm{\beta})=S_{n,\bm{\xi}}(\bm{\beta})-S_{n,\bm{\xi}}(\bm{\beta}_c)$. For a ball of radius $d$, choose nested Euclidean nets of mesh at most $h_j=d2^{-j}$ and cardinality at most $C2^{jp}$, with each point linked to a parent at distance at most $Ch_j$. For an edge at level $j$, independence, (\ref{eq:s3_max_single_increment}), and Bernstein's maximal inequality give
\begin{equation}\label{eq:s3_max_edge_bound}
E_0\max_{\text{edges at level }j}|\Delta\mathbb G_n|\le C\left\{\frac{\sqrt{kh_j(j+1)}}{\log(e/h_j)}+(j+1)\log n\right\}.
\end{equation}
Choose $J$ so that $h_J\le\zeta_n<h_{J-1}$. Since $d\le\delta_0$, we have $J\le C\log n$. Summing (\ref{eq:s3_max_edge_bound}) over $j\le J$ and using $\log(e/h_j)\ge\log(e/d)$ yields
\begin{equation}\label{eq:s3_max_chain_total}
E_0\max_{\bm{\beta}\in\mathcal N_J}|\mathbb G_n(\bm{\beta})|\le C\frac{\sqrt{kd}}{\log(e/d)}+C(\log n)^3.
\end{equation}
The same construction accommodates an arbitrary deterministic center $\bm{\beta}_c$ with constants independent of that center. It remains to control the interpolation between the final net and the continuum.

\emph{Uniform last-net interpolation.} Put $\eta_n=n^{-\chi}$ and let $\mathcal A_n$ be the event that for every $\bm{\beta}$ in the compact enlargement at most $p$ clean observations satisfy $|Y_i-\bm{x}_i^\top\bm{\beta}|\le3\eta_n$. For any fixed set $I$ of $p+1$ clean indices, select $\bm{a}_I\in\mathbb R^{p+1}$ with $\|\bm{a}_I\|=1$ and $\bm{a}_I^\top \bm{X}_I=0$. If all $p+1$ residuals could be within $3\eta_n$ of zero at one regression vector, then $|\bm{a}_I^\top\bm{\epsilon}_I|\le3\sqrt{p+1}\,\eta_n$. At least one coefficient of $\bm{a}_I$ has magnitude at least $(p+1)^{-1/2}$, so conditioning on all other errors and using $\|g_0\|_\infty<\infty$ shows $P_0\{|\bm{a}_I^\top\bm{\epsilon}_I|\le3\sqrt{p+1}\,\eta_n\}\le C\eta_n$. A union bound over at most $n^{p+1}$ index sets therefore gives
\begin{equation}\label{eq:s3_max_anti_concentration}
P_0(\mathcal A_n^c)\le Cn^{p+1-\chi}=o(1).
\end{equation}
On $\mathcal A_n$, for any $\bm{\beta}$ and its final-net approximation $\bm{\beta}'$ with $\|\bm{\beta}-\bm{\beta}'\|\le C\zeta_n$, at most $p$ observations lie within $3\eta_n$ of zero residual at $\bm{\beta}$. Each contributes at most $C\log n$ to the uncentered difference. For every other observation, both residuals remain at distance at least $2\eta_n$ from zero and the analytic log-derivative bound gives an absolute contribution at most $C\zeta_n/\eta_n$. Thus, uniformly on $\mathcal A_n$,
\begin{equation}\label{eq:s3_max_interpolation_uncentered}
\sup_{\|\bm{\beta}-\bm{\beta}'\|\le C\zeta_n}\left|\sum_{i\in K_n}\{\ell_{n,(\bm{\beta},\bm{\xi})}^{\rm sc}-\ell_{n,(\bm{\beta}',\bm{\xi})}^{\rm sc}\}\right|\le Cp\log n+Ck\zeta_n/\eta_n.
\end{equation}
For the centered expectations, (\ref{eq:s3_max_single_increment}) and Cauchy--Schwarz imply the uniform bound $Ck\sqrt{\zeta_n}/\log(e/\zeta_n)=o(1)$. This replaces the generally unjustified $O(k\zeta_n)$ expectation-interpolation assertion. On $\mathcal A_n^c$, the centered last-net interpolation is bounded deterministically by $Ck\log n$ from the uniform single-observation increment envelope. By (\ref{eq:s3_max_anti_concentration}) its contribution to the expectation is at most $Ck\log n\,P_0(\mathcal A_n^c)=o(1)$ because $\chi>p+5$. Finally, $k\zeta_n/\eta_n=O(n^{-2})$. Combining these interpolation estimates with (\ref{eq:s3_max_chain_total}) proves (\ref{eq:new_uniform_maximal}), with constants uniform in $\bm{\beta}_c$, $\bm{\xi}$, and $d$.
\end{proof}

\begin{lemma}[Direct full-posterior shell contraction]
\label{lem:new_full_shell}
Under the assumptions of Lemma~\ref{lem:new_moving_center}, for each $\varepsilon>0$,
\begin{equation}\label{eq:new_full_contraction}
 \lim_{M\to\infty}\limsup_{n\to\infty}
 P_0\!\left[\Pi_n\left\{\|\bm{\beta}-\bm{\beta}_{n,\omega}^\circ\|
 >Mr\mid D_n\right\}>\varepsilon\right]=0.
\end{equation}
This is for the \emph{full} posterior, and does not require full-to-tilted total variation or additional relative-growth restrictions.
\end{lemma}
\begin{proof}
\emph{Global reduction and capping.} By Lemma~\ref{lem:growing_contamination_reduction} and the identification in Lemma~\ref{lem:t1_coercivity}, the original posterior assigns probability $1-o_{P_0}(1)$ to a sufficiently small fixed product neighborhood $B\times\Xi$ of $\bm{\vartheta}^\dagger$. Since $\bm{\vartheta}_{n,\omega}^\circ\to\bm{\vartheta}^\dagger$, the location component of this neighborhood contains a fixed-radius ball about $\bm{\beta}_{n,\omega}^\circ$ for all sufficiently large $n$. By Lemma~\ref{lem:s3_spikes}(i)--(ii), the prior-integrated clean likelihood on exact-fit strips of width $\delta_n=e^{-n^4}$ is at most $\exp\{-cn^3+O(n\log n)\}$, and capping alters the likelihood only on these strips. On $B\times\Xi$ the outlier likelihood is at most $B_n e^{Cm}$, where $B_n=\prod_{i\in L_n}f_1(Y_i)$, whereas the full normalizer is at least $B_ne^{k(M^\star-\delta)}$ with probability tending to one by the normalizer argument of Lemma~\ref{lem:growing_contamination_reduction}. Therefore the total-variation discrepancy between the local original and capped full posteriors is $o_{P_0}(1)$, using the same elementary normalization inequality already established in Lemma~\ref{lem:s3_spikes}(ii). It is therefore enough to study the capped nuisance-integrated likelihood $\bar\ell_{n,\omega}^{\rm sc}$ on $B$.

\emph{Local quadratic drift and stochastic increments.} Write $\bm{\beta}_c=\bm{\beta}_{n,\omega}^\circ$, $\bm{\xi}_c=\bm{\xi}_{n,\omega}^\circ$, and $r=r_k$. Reduce $B$ if necessary so that $B\subset\{\bm{\beta}:\|\bm{\beta}-\bm{\beta}_c\|\le\delta_0\}$ eventually and, for fixed $c_0,C_0>0$ and all sufficiently large $n$, the population profile satisfies
\begin{equation}\label{eq:s3_shell_two_sided_population}
-C_0\|\bm{\beta}-\bm{\beta}_c\|^2\le\mathfrak m_{n,\omega}(\bm{\beta})-\mathfrak m_{n,\omega}(\bm{\beta}_c)\le-c_0\|\bm{\beta}-\bm{\beta}_c\|^2,\qquad\bm{\beta}\in B.
\end{equation}
Lemma~\ref{lem:s3_profile} gives a single $\mathfrak R_n=O_{P_0}(1)$, independent of the radius $d$, such that for every $r\le d\le\delta_0$,
\begin{equation}\label{eq:s3_shell_profile}
\begin{aligned}
\sup_{\|\bm{\beta}-\bm{\beta}_c\|\le d}\big|&\bar\ell_{n,\omega}^{\rm sc}(\bm{\beta})-\bar\ell_{n,\omega}^{\rm sc}(\bm{\beta}_c)-k\{\mathfrak m_{n,\omega}(\bm{\beta})-\mathfrak m_{n,\omega}(\bm{\beta}_c)\}\\
&-\{S_{n,\bm{\xi}_c}(\bm{\beta})-S_{n,\bm{\xi}_c}(\bm{\beta}_c)\}\big|\le \mathfrak R_n(\sqrt{k}\,d+1).
\end{aligned}
\end{equation}
Also, uniformly for these radii, Lemma~\ref{lem:new_uniform_maximal} yields
\begin{equation}\label{eq:s3_shell_maximal}
E_0\sup_{\|\bm{\beta}-\bm{\beta}_c\|\le d}|S_{n,\bm{\xi}_c}(\bm{\beta})-S_{n,\bm{\xi}_c}(\bm{\beta}_c)|\le C\frac{\sqrt{kd}}{\log(e/d)}+C(\log n)^3.
\end{equation}
Here $\bm{\beta}_c$ and $\bm{\xi}_c$ are deterministic functions of the design and contamination array, so the arbitrary-deterministic-center uniformity in Lemma~\ref{lem:new_uniform_maximal} applies.

\emph{Uniform negative drift on dyadic shells.} Fix $M\ge2$ and, for $n$ sufficiently large that $Mr<\delta_0$, let $d_j=2^jMr$ and $D_j=\min(2d_j,\delta_0)$, $0\le j\le J_n$, where $J_n$ is the last index with $d_j<\delta_0$. Set $\mathcal A_{j,n}=\{\bm{\beta}\in B:d_j<\|\bm{\beta}-\bm{\beta}_c\|\le D_j\}$ and let $\mathcal E_{n,M}$ be the event on which, simultaneously for all these shells,
\begin{equation}\label{eq:s3_shell_good_event}
\sup_{\bm{\beta}\in\mathcal A_{j,n}}\{\bar\ell_{n,\omega}^{\rm sc}(\bm{\beta})-\bar\ell_{n,\omega}^{\rm sc}(\bm{\beta}_c)\}\le-\frac{c_0}{2}kd_j^2.
\end{equation}
Since $\mathfrak R_n=O_{P_0}(1)$, the common remainder in \eqref{eq:s3_shell_profile} divided by $kd_j^2$ is bounded uniformly over $j$ by $\mathfrak R_n\{2/(\sqrt{k}Mr)+1/(kM^2r^2)\}=o_{P_0}(1)$. Markov's inequality, \eqref{eq:s3_shell_two_sided_population}--\eqref{eq:s3_shell_maximal}, and a union bound thus give
\begin{equation}\label{eq:s3_shell_probability_bound}
P_0(\mathcal E_{n,M}^c)\le o(1)+C\sum_{j=0}^{J_n}\left\{\frac{\sqrt{kD_j}}{kd_j^2\log(e/D_j)}+\frac{(\log n)^3}{kd_j^2}\right\}.
\end{equation}
The second sum is $O\{(\log n)^3/(kM^2r^2)\}=o(1)$. For the first sum, split the indices into $d_j\le r^{1/2}$ and $d_j>r^{1/2}$. On the first range, $\log(e/D_j)\ge c\log(1/r)$, and the defining identity $kr^3\{\log(1/r)\}^2=1$ together with a geometric series gives a bound $CM^{-3/2}$. On the remaining range, $\log(e/D_j)$ is bounded away from zero and the contribution is at most $Ck^{-1/2}r^{-3/4}=o(1)$. Hence
\begin{equation}\label{eq:s3_shell_good_probability}
\limsup_{n\to\infty}P_0(\mathcal E_{n,M}^c)\le CM^{-3/2}.
\end{equation}
On $\mathcal E_{n,M}$, every $\bm{\beta}\in B$ with $\|\bm{\beta}-\bm{\beta}_c\|>Mr$ belongs to a shell, and \eqref{eq:s3_shell_good_event} implies $\bar\ell_{n,\omega}^{\rm sc}(\bm{\beta})-\bar\ell_{n,\omega}^{\rm sc}(\bm{\beta}_c)\le-c_1k\|\bm{\beta}-\bm{\beta}_c\|^2$ for a fixed $c_1>0$ (one may take $c_1=c_0/8$).

\emph{Local integrated denominator.} Let $\pi_{\bm{\beta}}$ be the Gaussian prior density of the regression vector, which is continuous and bounded away from zero on $B$. For sufficiently large $n$, the Euclidean ball $\mathbb B_n=\{\bm{\beta}:\|\bm{\beta}-\bm{\beta}_c\|\le r\}$ is contained in $B$. The lower population-profile inequality in \eqref{eq:s3_shell_two_sided_population}, the uniform comparison \eqref{eq:s3_shell_profile}, and the maximal bound \eqref{eq:s3_shell_maximal} show that
\begin{equation}\label{eq:s3_shell_local_fluctuation}
\frac1{kr^2}\sup_{\bm{\beta}\in\mathbb B_n}\big|\bar\ell_{n,\omega}^{\rm sc}(\bm{\beta})-\bar\ell_{n,\omega}^{\rm sc}(\bm{\beta}_c)\big|=O_{P_0}(1).
\end{equation}
Indeed, $\sqrt{kr}/\log(e/r)=O(kr^2)$ by the definition of $r$, $(\log n)^3=o(kr^2)$, and $\mathfrak R_n(\sqrt{k}r+1)=o_{P_0}(kr^2)$. Consequently, for every $\delta>0$ one can choose a fixed $K_\delta<\infty$ such that, for all sufficiently large $n$, with probability at least $1-\delta$,
\begin{equation}\label{eq:s3_shell_denominator}
\begin{aligned}
\int_B e^{\bar\ell_{n,\omega}^{\rm sc}(\bm{\beta})}\pi_{\bm{\beta}}(d\bm{\beta})&\ge\int_{\mathbb B_n}e^{\bar\ell_{n,\omega}^{\rm sc}(\bm{\beta})}\pi_{\bm{\beta}}(d\bm{\beta})\\
&\ge c_\pi r^p\exp\{\bar\ell_{n,\omega}^{\rm sc}(\bm{\beta}_c)-K_\delta kr^2\},
\end{aligned}
\end{equation}
where $c_\pi>0$ is independent of $n$. On $\mathcal E_{n,M}$, the corresponding numerator outside the $Mr$ ball is at most $\pi_{\bm{\beta}}(B)\exp\{\bar\ell_{n,\omega}^{\rm sc}(\bm{\beta}_c)-c_1kM^2r^2\}$. Intersecting this event with \eqref{eq:s3_shell_denominator} gives
\begin{equation}\label{eq:s3_shell_posterior_ratio}
\Pi_n^{\rm sc}\{\|\bm{\beta}-\bm{\beta}_c\|>Mr\mid\bm{\beta}\in B,\bm{\xi}\in\Xi,D_n\}\le C r^{-p}\exp\{-(c_1M^2-K_\delta)kr^2\}.
\end{equation}
For every fixed $M$ satisfying $c_1M^2>K_\delta$, the right-hand side tends to zero because $kr^2\to\infty$ and $|\log r|=o(kr^2)$. By \eqref{eq:s3_shell_good_probability}, the upper limiting probability that this conditional posterior mass exceeds any fixed $\varepsilon>0$ is at most $CM^{-3/2}+\delta$. The full posterior outside $B\times\Xi$ and the capping discrepancy are $o_{P_0}(1)$ by the first paragraph. Letting first $n\to\infty$, then $M\to\infty$, and finally $\delta\downarrow0$ proves \eqref{eq:new_full_contraction}.
\end{proof}

\begin{lemma}[Exact screening and residual-quantile contraction transfer]
\label{lem:s3_correction}
Under (A1)--(A5), suppose the full posterior satisfies the $r_k$-shell bound of Lemma~\ref{lem:new_full_shell}, with $k r_k^3\{\log(1/r_k)\}^2=1$. Then (i) $P_0(\widehat K_n=K_n)\to1$; (ii) the slope marginal contracts around $\bm{b}_0$ at rate $r_k$; and (iii) the corrected posterior defined in \eqref{eq:s3_corrected_posterior} jointly contracts around $\bm{\beta}_0=(\alpha_0,\bm{b}_0^\top)^\top$ at rate $r_k$: for every $\varepsilon>0$,
\begin{equation}\label{eq:t1_dq_rate_transfer}
\lim_{M\to\infty}\limsup_{n\to\infty}P_0\left[
\Pi_n^{\mathrm{cor}}\left\{
\|\bm{\beta}^{\mathrm{cor}}-\bm{\beta}_0\|>Mr_k\mid D_n
\right\}>\varepsilon\right]=0.
\end{equation}
\end{lemma}
\begin{proof}
\emph{Screening.}\quad By (A1)--(A2), bounded design and the finite $q$th clean-error moment imply $t^qP_0(|\epsilon|>t)\to0$ as $t\to\infty$. Hence, with $t_n=n^{1/q}-C_0$ for a deterministic $C_0$ bounding $|\bm{x}_i^\top\bm{\beta}_0|$,
\[
P_0\left(\max_{i\in K_n}|Y_i|>n^{1/q}\right)\le kP_0(|\epsilon|>t_n)=o(1).
\]
By (A3), $\min_{i\in L_n}|Y_i|\ge c_-\omega_n-C_a>n^{1/q}$ eventually, since $B_*>2>1/q$. Thus $P_0(\widehat K_n=K_n)\to1$.

\medskip\noindent\emph{Slope contraction.}\quad Write $\bm{\beta}_{n,\omega}^\circ=(\alpha_{n,\omega}^\circ, (\bm{b}_{n,\omega}^\circ)^\top)^\top$. Lemma~\ref{lem:new_moving_center} gives $\|\bm{b}_{n,\omega}^\circ-\bm{b}_0\|=o(r_k)$, while Lemma~\ref{lem:new_full_shell} implies
\[
\lim_{M\to\infty}\limsup_{n\to\infty}P_0\left[
\Pi_n^{\bm{b}}\left\{\|\bm{b}-\bm{b}_0\|>Mr_k\mid D_n\right\}
>\varepsilon\right]=0.
\]

\medskip\noindent\emph{Draw-wise transfer.}\quad Let $Q_{K,n}=\widehat Q_\tau(U_i:i\in K_n)$. On $\{\widehat K_n=K_n\}$, location equivariance and the 1-Lipschitz property of the lower empirical quantile in the sup-norm yield, \emph{for every} $\bm{b}\in\mathbb R^{p-1}$,
\begin{equation}\label{eq:t1_dq_quantile_bound}
|T_n(\bm{b})-\alpha_0|
\le |Q_{K,n}|+B_z\|\bm{b}-\bm{b}_0\|,
\qquad B_z=\sup_{n,i}\|\bm{z}_i\|<\infty.
\end{equation}
Consequently,
\begin{equation}\label{eq:t1_dq_joint_bound}
\|\bm{\beta}^{\mathrm{cor}}-\bm{\beta}_0\|
\le |Q_{K,n}|+C_z\|\bm{b}-\bm{b}_0\|,
\qquad C_z=\sqrt{1+B_z^2}.
\end{equation}
Under (A1), $G_0(0)=\tau$ and $g_0$ is continuous and strictly positive near zero. Therefore $Q_{K,n}=O_{P_0}(k^{-1/2})=o_{P_0}(r_k)$. For any $M>0$, on exact screening, \eqref{eq:t1_dq_joint_bound} gives
\begin{align*}
&\Pi_n^{\mathrm{cor}}\left\{
\|\bm{\beta}^{\mathrm{cor}}-\bm{\beta}_0\|>Mr_k\mid D_n\right\}\\
&\quad\le
\Pi_n^{\bm{b}}\left\{\|\bm{b}-\bm{b}_0\|>\frac{Mr_k}{2C_z}\mid D_n\right\}
+\mathbf1\{|Q_{K,n}|>Mr_k/2\}.
\end{align*}
The screening failure probability tends to zero. Taking $\limsup_{n\to\infty}$ followed by $M\to\infty$ proves \eqref{eq:t1_dq_rate_transfer}. No rate statement for the raw intercept posterior or its median is needed for this transfer.
\end{proof}

\begin{proof}[Proof of Theorem 1]
Lemma~\ref{lem:new_moving_center} constructs a local full-data stationary center whose slope differs from $\bm{b}_0$ by $o(r_k)$. Lemmas~\ref{lem:s3_spikes}--\ref{lem:s3_profile} justify the local capped and nuisance-profile likelihood, while Lemmas~\ref{lem:new_uniform_maximal} and~\ref{lem:new_full_shell} establish $r_k$-scale contraction of the original slope marginal. Lemma~\ref{lem:s3_correction} transfers this contraction through the residual-quantile map to $\Pi_n^{\mathrm{cor}}$. Since $k/n\to1-\eta>0$, we have $r_k\asymp r_n$, and the double-limit contraction statement implies that, for every deterministic $M_n\to\infty$,
\[
\Pi_n^{\mathrm{cor}}\{\|\bm{\beta}^{\mathrm{cor}}-\bm{\beta}_0\|>M_n r_n\mid D_n\}\xrightarrow{P_0}0.
\]
To prove consistency of the correction variable, let $\bm{\beta}^\dagger=(\alpha_0+a^\dagger,\bm{b}_0^\top)^\top$. By Lemma~\ref{lem:growing_contamination_reduction} and the global identification in Lemma~\ref{lem:t1_coercivity}, the original posterior satisfies $\Pi_n(|\alpha-(\alpha_0+a^\dagger)|>\delta\mid D_n)\to0$ in $P_0$-probability for every $\delta>0$. By the first part of the theorem, $T_n(\bm{b})$ converges to $\alpha_0$ in corrected posterior probability, equivalently in the original slope posterior. Therefore, for every $\delta>0$,
\begin{align*}
\Pi_n^\Delta\{|\Delta_n-a^\dagger|>\delta\mid D_n\}
&\le\Pi_n\{|\alpha-(\alpha_0+a^\dagger)|>\delta/2\mid D_n\}\\
&\quad+\Pi_n^{\bm{b}}\{|T_n(\bm{b})-\alpha_0|>\delta/2\mid D_n\}\xrightarrow{P_0}0,
\end{align*}
which proves the second assertion.
\end{proof}

\section{Proofs of Theorem 2 and Proposition 5}
\label{sec:supp_nonregular}

We retain the model, assumptions, and notation of Sections S1--S3. In particular, $\bm{\vartheta}_{n,\omega}^{\circ}=(\bm{\beta}_{n,\omega}^{\circ\top},\bm{\xi}_{n,\omega}^{\circ\top})^\top$ is the deterministic moving stationary center of Lemma~\ref{lem:new_moving_center}, and $\bar\ell_{n,\omega}^{\rm sc}$ is the locally nuisance-integrated, smoothly capped log likelihood of Lemma~\ref{lem:s3_profile}. Write $r=r_k$, $L_r=\log(1/r)$, and $a_k=kr^2$, so that $kr^3L_r^2=1$, $a_k\to\infty$, and $\log(1/r)=o(a_k)$. Put $d_\tau=\log\{\tau/(1-\tau)\}$, $K_\tau=\pi^2+d_\tau^2$, and $m_\nu(\bm t)=\int|\bm x(\bm z)^\top\bm t|\,\nu(d\bm z)$. The centered Gaussian process $\mathbb G$ has the covariance specified in Theorem 2, equivalently $\operatorname{Cov}\{\mathbb G(\bm t),\mathbb G(\bm u)\}=g_0(a^\dagger)K_\tau\{m_\nu(\bm t)+m_\nu(\bm u)-m_\nu(\bm t-\bm u)\}/2$, and its limiting criterion is $\mathcal Q(\bm t)=\mathbb G(\bm t)-\bm t^\top\bm H_{\rm eff}\bm t/2$. These definitions refer to the full, original working posterior; the smooth cap is an auxiliary device used in the local proofs.

\subsection{Singular likelihood and Gaussian-process limit}

\begin{lemma}[Singular central expansion and covariance kernel]
\label{lem:s4_kernel}
The following conclusions hold for fixed $\tau\in(0,1)$ and $\gamma>0$.
\begin{enumerate}
\item[(i)] Put $a_+=\tau$, $a_-=1-\tau$, and $c_\tau=\tau(1-\tau)$. As $u\to0$ from either side,
\begin{equation}\label{eq:s4_central_expansion}
\begin{aligned}
f_1(u)&=c_\tau\left\{\gamma\log\frac1{|u|}+D_H-\gamma\gamma_{\rm E}-\gamma\log a_{\operatorname{sgn}u}+o(1)\right\},\\
D_H&=\int_0^1\frac{H(v;\gamma)-\gamma}{v}\,dv+\int_1^\infty\frac{H(v;\gamma)}v\,dv.
\end{aligned}
\end{equation}
Uniformly over every fixed compact interior set of $\bm\xi=(\log\sigma,\operatorname{logit}s)^\top$, the log working density, written as a function of a nonzero residual $e$, satisfies
\begin{equation}\label{eq:s4_log_central}
\log\left\{\sigma^{-1}f_s(e/\sigma)\right\}=A(\bm\xi)+\log\log\frac1{|e|}+\frac{B(\bm\xi)-\log a_{\operatorname{sgn}e}}{\log(1/|e|)}+o\left\{\frac1{\log(1/|e|)}\right\},
\end{equation}
where $A(\bm\xi)=\log(c_\tau s\gamma/\sigma)$ and $B(\bm\xi)=\log\sigma+D_H/\gamma-\gamma_{\rm E}+(1-s)/(s\gamma)$. Moreover, uniformly on the same compact sets,
\begin{equation}\label{eq:s4_central_score}
e\,\partial_e\log\left\{\sigma^{-1}f_s(e/\sigma)\right\}=-\frac1{\log(1/|e|)}+O\left\{\frac1{\log^2(1/|e|)}\right\}.
\end{equation}
\item[(ii)] For $h\in\mathbb R$, define the singular increment kernel
\begin{equation}\label{eq:s4_kernel_def}
\kappa_{\tau,h}(v)=\log\left|\frac{v}{v-h}\right|+d_\tau\{\mathbf1(v>0)-\mathbf1(v-h>0)\},\qquad v\notin\{0,h\}.
\end{equation}
It belongs to $L^2(\mathbb R)$, and for all $h,g\in\mathbb R$,
\begin{equation}\label{eq:s4_kernel_identity}
\int_{\mathbb R}\kappa_{\tau,h}(v)^2\,dv=K_\tau|h|,\qquad \int_{\mathbb R}\kappa_{\tau,h}(v)\kappa_{\tau,g}(v)\,dv=\frac{K_\tau}{2}\{|h|+|g|-|h-g|\}.
\end{equation}
\end{enumerate}
\end{lemma}
\begin{proof}
\emph{Central expansion.} For $z>0$, put $K(z)=\int_0^\infty v^{-1}\exp(-z/v)H(v;\gamma)\,dv$. Splitting at $v=1$ and adding and subtracting $\gamma$ gives
\[
K(z)=\gamma\int_0^1v^{-1}e^{-z/v}\,dv+\int_0^1e^{-z/v}\frac{H(v;\gamma)-\gamma}{v}\,dv+\int_1^\infty e^{-z/v}\frac{H(v;\gamma)}v\,dv.
\]
The first integral is the exponential integral $E_1(z)=-\log z-\gamma_{\rm E}+o(1)$ as $z\downarrow0$. Since $H(v;\gamma)-\gamma=O(v)$ near zero and $H(v;\gamma)/v$ is integrable at infinity, dominated convergence identifies the origin constant in the expansion proved in Lemma~\ref{lem:s3_nuisance}(i) as $C_\gamma=D_H-\gamma\gamma_{\rm E}$. Thus $K(z)=\gamma\log(1/z)+D_H-\gamma\gamma_{\rm E}+o(1)$. Substituting $z=a_{\operatorname{sgn}u}|u|$ proves \eqref{eq:s4_central_expansion}. Since $f_s=(1-s)f_0+sf_1$ and $f_0(u)=c_\tau\{1+O(|u|)\}$ near zero, expansion of $\log\{\sigma^{-1}f_s(e/\sigma)\}$ proves \eqref{eq:s4_log_central}, with the displayed $A$ and $B$. The differentiated origin bounds of Lemma~\ref{lem:s3_nuisance}(i), specifically $DK(z)=-\gamma+O\{z\log(1/z)\}$ and $K(z)=\gamma\log(1/z)+O(1)$, give \eqref{eq:s4_central_score} directly. All expansions are uniform on fixed interior nuisance compacts because $\sigma$ and $s$ are bounded above and away from zero there.

\emph{Kernel identities.} It is enough to consider $h=1$, since $\kappa_{\tau,h}(hv)=\kappa_{\tau,1}(v)$ for $h>0$, with reflection giving the same squared norm for $h<0$. Direct integration gives
\[
\int_{\mathbb R}\log^2\left|\frac{v}{v-1}\right|\,dv=\pi^2,\qquad \int_0^1\log\frac{v}{1-v}\,dv=0.
\]
For completeness, the first identity follows by splitting the integral over $(-\infty,0)$, $(0,1)$, and $(1,\infty)$: each contribution is $\pi^2/3$, by the changes of variables $v=-t/(1-t)$ and $v=1/t$ and the identities $\int_0^1\log^2t/(1-t)^2\,dt=\pi^2/3$ and $\int_0^1\log^2\{t/(1-t)\}\,dt=\pi^2/3$. The indicator difference for $h=1$ is $\mathbf1(0<v<1)$; the cross term vanishes by the second identity, and its squared integral is one. Hence the first formula of \eqref{eq:s4_kernel_identity} holds. Finally, $\kappa_{\tau,h}(v)-\kappa_{\tau,g}(v)=\kappa_{\tau,h-g}(v-g)$ almost everywhere. Translation invariance of Lebesgue measure and polarization establish the second formula.
\end{proof}

For a clean response, the residual at $\bm\beta^\dagger$ is $V_i=Y_i-\bm x_i^\top\bm\beta^\dagger=\epsilon_i-a^\dagger$, with density $v\mapsto g_0(v+a^\dagger)$. For $\bm\xi=(\log\sigma,\operatorname{logit}s)^\top$ in a fixed compact interior set, write $\phi_{\bm\xi}(v)=\log\{\sigma^{-1}f_s(v/\sigma)\}$ and $\Delta_{r,h,\bm\xi}(v)=\phi_{\bm\xi}(v-rh)-\phi_{\bm\xi}(v)$. The next lemma identifies the finite-dimensional limit and the increment bounds used for functional convergence.

\begin{lemma}[Uniform increment bounds and finite-dimensional Gaussian-process convergence]
\label{lem:s4_fidi}
Fix $H<\infty$ and a compact interior nuisance set $\Xi$. Uniformly in $|h|,|g|\le H$, $\bm\xi\in\Xi$, and for residual laws whose densities are \emph{uniformly bounded}, are continuous near zero, and converge locally uniformly to $g_0(a^\dagger+\cdot)$, the following assertions hold as $r\downarrow0$.
\begin{enumerate}
\item[(i)] At $\bm\xi=\bm\xi^\dagger$,
\begin{equation}\label{eq:s4_covariance_limit}
\frac{L_r^2}{r}\operatorname{Cov}_0\{\Delta_{r,h,\bm\xi^\dagger}(V),\Delta_{r,g,\bm\xi^\dagger}(V)\}\longrightarrow g_0(a^\dagger)\int_{\mathbb R}\kappa_{\tau,h}(v)\kappa_{\tau,g}(v)\,dv,
\end{equation}
and, uniformly over $|h|\le H$, $E_0|\Delta_{r,h,\bm\xi}(V)|^4\le C_H\{rL_r^{-4}+r^4\}$. The same covariance limit holds for any deterministic $\bm\xi_r\to\bm\xi^\dagger$ and clean residual shifts tending uniformly to $a^\dagger$.
\item[(ii)] Uniformly over the same nuisance compact and sufficiently small $|w|$,
\begin{equation}\label{eq:s4_translation_modulus}
E_0|\phi_{\bm\xi}(V-w)-\phi_{\bm\xi}(V)|^2\le C\frac{|w|}{\log^2(e/|w|)}.
\end{equation}
\item[(iii)] Define $\bm h_i(\bm t)=\bm x_i^\top\bm t$ and $Z_n(\bm t)=a_k^{-1}\sum_{i\in K_n}[\Delta_{r,\bm h_i(\bm t),\bm\xi^\dagger}(V_i)-E_0\Delta_{r,\bm h_i(\bm t),\bm\xi^\dagger}(V_i)]$. The finite-dimensional distributions of $Z_n$ converge to those of $\mathbb G$. This also holds when each summand is replaced by the smoothly capped increment defined in Lemma~\ref{lem:s3_spikes}.
\end{enumerate}
\end{lemma}
\begin{proof}
\emph{Local moment calculation.} Put $L=L_r$ and choose a small $\varepsilon_0>0$ on which \eqref{eq:s4_log_central}--\eqref{eq:s4_central_score} hold uniformly. Divide the residual domain into $|V|\le Mr$, $Mr<|V|\le\varepsilon_0$, and $|V|>\varepsilon_0$, where $M>2H$ is fixed. For $V=rv$, \eqref{eq:s4_log_central} gives, for $v\notin\{0,h\}$,
\begin{equation}\label{eq:s4_kernel_pointwise}
L\Delta_{r,h,\bm\xi_r}(rv)\longrightarrow\kappa_{\tau,h}(v),
\end{equation}
uniformly for convergent interior nuisance sequences $\bm\xi_r$. On $|v|\le M$, the log--log expansion and the elementary inequality $|\log(1+x)|\le C|x|$ for $x\ge-1/2$, applied away from exponentially narrow neighborhoods of $0$ and $h$, yield an integrable majorant of the form $C_{M,H}\{1+|\log|v||^4+|\log|v-h||^4\}$. On the exceptional neighborhoods $\min(|v|,|v-h|)<e^{-L/2}$, direct integration of the same logarithmic powers gives a contribution tending to zero. Since the residual density is uniformly bounded and converges locally uniformly near zero, dominated convergence gives the covariance contribution $rg_0(a^\dagger)L^{-2}\int_{-M}^M\kappa_{\tau,h}\kappa_{\tau,g}+o(rL^{-2})$ and a fourth-moment contribution $O_H(rL^{-4})$.

\emph{Intermediate and outer residual regions.} For $Mr<|V|\le\varepsilon_0$, the mean-value theorem and \eqref{eq:s4_central_score} give $|\Delta_{r,h,\bm\xi}(V)|\le CrH/\{|V|\log(e/|V|)\}$. The contribution of this region, after multiplication by $L^2/r$, to the second moment is at most $C_H/M+o(1)$, and its fourth moment is $O_H\{r/(M^3L^4)\}$. For $|V|>\varepsilon_0$, the log-density derivative along these shifts is bounded, so the second and fourth moments are $O_H(r^2)$ and $O_H(r^4)$. Letting first $r\downarrow0$ and then $M\uparrow\infty$ proves the corresponding uncentered second-moment limit and the fourth-moment bound. For the centered covariance, use $|E_0\Delta_{r,h,\bm\xi}(V)|\le(E_0|\Delta_{r,h,\bm\xi}(V)|^2)^{1/2}$, together with the differentiability of the population expected log likelihood: at $\bm\xi^\dagger$ its location score vanishes and the expected increment is $O_H(r^2)$, while at convergent moving centers it is $O_H(r|\partial_a F_0|+r^2)=o(r)+O(r^2)$. These estimates imply that the product of the increment means is $o(rL^{-2})$; the uniform local Lipschitz bound of the population score near $\bm\xi^\dagger$ justifies this last step. The translation modulus \eqref{eq:s4_translation_modulus} is precisely the uniform $L^2$ estimate already proved in \eqref{eq:s3_max_translation_modulus} of Lemma~\ref{lem:new_uniform_maximal}, after multiplication by the uniformly bounded residual density. The same bound applies to the smoothly capped log likelihood because $x\mapsto L_n+\psi(x-L_n)$ is 1-Lipschitz. No second translation-modulus argument is required.

\emph{Uniformity in projected shifts and moving parameters.}
The preceding central-region calculation also yields uniformity on bounded shift sets; we give the missing compactness step. For $|h|\le H$, set $D_{r,h,\bm\xi}(v)=L_r\{\phi_{\bm\xi}(rv-rh)-\phi_{\bm\xi}(rv)\}$. The translation bound \eqref{eq:s4_translation_modulus} (equivalently \eqref{eq:s3_max_translation_modulus}) and the identity $L_r^2/r$ for the rescaling show, uniformly in $\bm\xi\in\Xi$ and $|h|,|g|\le H$,
\[
 \|D_{r,h,\bm\xi}-D_{r,g,\bm\xi}\|_{L^2(dv)}^2
 \le C\,|h-g|\frac{L_r^2}{\log^2\{e/(r|h-g|)\}}
 \le C_H|h-g|,
\]
where the left side is zero when $h=g$ and the bound applies for sufficiently small $r$. The limiting kernels have exactly the same $L^2$ modulus, since $\|\kappa_{\tau,h}-\kappa_{\tau,g}\|_2^2=K_\tau|h-g|$. The pointwise-in-$h$ $L^2$ convergence obtained by the three-region split therefore upgrades, by a finite $h$-net, to uniform $L^2$ convergence for $|h|\le H$. The convergence is also uniform for $\bm\xi$ in any fixed compact interior set: on $|v|\le M$ the central expansion \eqref{eq:s4_log_central} is uniform in $\bm\xi$ away from the two singular points, with the same integrable logarithmic majorants; on $|v|>M$, the outer derivative bound is independent of $\bm\xi$ and the normalized $L^2$ remainder is at most $C_H/M+o(1)$. Letting first $r\downarrow0$ and then $M\uparrow\infty$ gives the uniform result. For any triangular family of residual densities uniformly bounded and locally uniformly convergent near zero, the same decomposition gives uniform covariance convergence, since the normalized squared increments are uniformly integrable on $|v|\le M$ and have tails $O(M^{-1})$. In particular the result applies uniformly across all the bounded clean-design rows at the moving center. The squared means are $o(r/L_r^2)$ by the differentiability of the population criterion and the convergence of the moving score to zero. Thus the covariance limit used below is uniform in the projected arguments, as required for design averaging.

\emph{Finite-dimensional convergence.} For fixed $\bm t_1,\ldots,\bm t_J$, bounded design and \eqref{eq:s4_covariance_limit} imply that the covariance of the centered sums, divided by $a_k^2=k^2r^4$, converges to the covariance of $\mathbb G$. Indeed, $kr^3L_r^2=1$, the empirical design measures converge weakly to $\nu$, and the kernel identity \eqref{eq:s4_kernel_identity} supplies the limit. For each fixed linear combination of these sums, the fourth-moment Lyapunov ratio is bounded by
\[
\frac{Ck\{rL_r^{-4}+r^4\}}{(kr^2)^4}=O(r^2L_r^2)+O\{k^{-3}r^{-4}\}=o(1).
\]
The fourth-moment bound verifies the Lindeberg condition for each nondegenerate linear combination (and Chebyshev's inequality handles combinations with vanishing limiting variance). The Cram\'er--Wold device therefore proves finite-dimensional Gaussian convergence. Finally, at any fixed finite collection of $\bm t$ values, the probability that an observation has a fitted residual in the cap transition region of width at most $e^{-n^4}$ is $O(ke^{-n^4})=o(1)$, and the expected removed contribution is $o(a_k/k)$ by the bounded clean density and the direct bound $\int_0^{e^{-n^4}}\log\{1+\log(1/v)\}\,dv\le Ce^{-n^4}\log n$. Hence the smoothly capped increments have the same finite-dimensional limit.
\end{proof}

\begin{lemma}[Moving-center functional convergence and Gaussian argmax]
\label{lem:s4_functional}
For each fixed compact $T\subset\mathbb R^p$, define
\begin{equation}\label{eq:s4_local_process}
\mathcal Q_{n,\omega}(\bm t)=\frac{\bar\ell_{n,\omega}^{\rm sc}(\bm\beta_{n,\omega}^\circ+r\bm t)-\bar\ell_{n,\omega}^{\rm sc}(\bm\beta_{n,\omega}^\circ)}{a_k},\qquad \bm t\in T.
\end{equation}
Under (A1)--(A5), $\mathcal Q_{n,\omega}\Rightarrow\mathcal Q$ in $C(T)$. The field $\mathbb G$ admits a continuous version, and almost surely $\mathcal Q(\bm t)\to-\infty$ as $\|\bm t\|\to\infty$ and $\mathcal Q$ has a unique maximizer $\bm T^\star$ on $\mathbb R^p$.
\end{lemma}
\begin{proof}
\emph{Moving-center finite-dimensional limit.} Lemma~\ref{lem:new_moving_center} gives $\bm\vartheta_{n,\omega}^\circ\to\bm\vartheta^\dagger$. The clean residual at $\bm\beta_{n,\omega}^\circ$ is $V_{ni}=\epsilon_i-(\alpha_{n,\omega}^\circ-\alpha_0)-\bm z_i^\top(\bm b_{n,\omega}^\circ-\bm b_0)$, whose density near zero converges uniformly in $i$ to $g_0(a^\dagger+\cdot)$, by (A2), (A4), and Lemma~\ref{lem:new_moving_center}. The covariance, fourth-moment, and Lindeberg estimates of Lemma~\ref{lem:s4_fidi} apply uniformly to this triangular array and deterministic $\bm\xi_{n,\omega}^\circ$, because the shifted densities are uniformly bounded and converge locally uniformly near zero. This gives the finite-dimensional Gaussian limit without the stronger, generally invalid claim that the moving-center and fixed-center empirical processes are pathwise close on the $r$ scale. On the other hand, $\bm\vartheta_{n,\omega}^\circ$ is an exact stationary point of the finite-sample full population criterion, and Lemma~\ref{lem:new_moving_center} gives convergence of its negative Hessian to $\bm H$. Taylor expansion of the population nuisance-profile criterion, uniformly on bounded $\bm t$ sets, gives
\begin{equation}\label{eq:s4_profile_drift}
\frac{k\{\mathfrak m_{n,\omega}(\bm\beta_{n,\omega}^\circ+r\bm t)-\mathfrak m_{n,\omega}(\bm\beta_{n,\omega}^\circ)\}}{a_k}=-\frac12\bm t^\top\bm H_{\rm eff}\bm t+o(1).
\end{equation}
Lemma~\ref{lem:s3_profile} controls the difference between the log-integrated profile increment and the sum of the deterministic profile increment and the centered clean empirical increment. On $\|\bm t\|\le M$ its remainder, divided by $a_k$, is bounded by $\mathfrak R_n\{M/(\sqrt{k}r)+1/(kr^2)\}=o_{P_0}(1)$. Thus the finite-dimensional distributions of $\mathcal Q_{n,\omega}$ converge to those of $\mathcal Q$.

\emph{Tightness on compact parameter sets.} By \eqref{eq:s4_translation_modulus} and bounded design, the centered clean empirical increments indexed by $\bm t\in T$ have an $L^2$ increment semimetric bounded by $C\|\bm t-\bm u\|^{1/2}$ after division by $a_k$, since $kr^3L_r^2=1$. Here is a two-index version of the chaining calculation, which is needed because a fixed-center maximal inequality by itself is insufficient. Cover $T$ by at most $C_T\delta^{-p}$ Euclidean balls of radius $\delta$ and select one anchor in each ball. For each anchor construct dyadic Euclidean nets at mesh $\delta2^{-j}$, stopping when the coefficient-space mesh is at most $\zeta_n=n^{-(\chi+3)}$, as in Lemma~\ref{lem:new_uniform_maximal}. At level $j$, the total number of edges across all balls is bounded by $C_T\delta^{-p}2^{jp}$, so its logarithm is at most $C_T+p\log(1/\delta)+pj$. The $L^2$ translation bound \eqref{eq:s4_translation_modulus}, Bernstein's maximal inequality, and the $C\log n$ envelope of the capped single-observation differences give, after normalization by $a_k$,
\[
 C_T\left\{\sqrt{\delta2^{-j}\{1+\log(1/\delta)+j\}}+
 \frac{\log n}{a_k}\{1+\log(1/\delta)+j\}\right\}
\]
for the expected maximum at that level, uniformly over all anchors. The geometric sum of the square-root terms is at most $C_T\sqrt{\delta\{1+\log(1/\delta)\}}$. There are $O(\log n)$ levels for each fixed $\delta$ and the second terms sum to $O_T\{(\log n)^3/a_k\}=o(1)$. The terminal-net interpolation is controlled by the same rank-free anti-concentration event and cap argument as in \eqref{eq:s3_max_anti_concentration}--\eqref{eq:s3_max_interpolation_uncentered}, uniformly over the finitely many anchors: its normalized expectation vanishes. Since every pair at distance at most $\delta$ lies in one enlarged anchor ball and can be compared via its anchor, we conclude, for each fixed $\delta\in(0,1)$,
\begin{equation}\label{eq:s4_modulus_bound}
\limsup_{n\to\infty}E_0\sup_{\substack{\bm t,\bm u\in T\\\|\bm t-\bm u\|\le\delta}}\left|\frac{S_{n,\bm\xi_{n,\omega}^\circ}(\bm\beta_{n,\omega}^\circ+r\bm t)-S_{n,\bm\xi_{n,\omega}^\circ}(\bm\beta_{n,\omega}^\circ+r\bm u)}{a_k}\right|\le C_T\sqrt{\delta\{1+\log(1/\delta)\}}.
\end{equation}
Here the finite cover contributes only the fixed-dimensional logarithmic entropy, and the residual $O\{(\log n)^3/a_k\}$ from the terminal dyadic net tends to zero. This bound tends to zero with $\delta$, proving asymptotic equicontinuity. The capped processes have continuous sample paths on $T$, while the population drift converges uniformly by \eqref{eq:s4_profile_drift}; tightness and finite-dimensional convergence give the asserted weak convergence in $C(T)$.

\emph{Existence and uniqueness of the limit argmax.} Put $B_x=\sup_{\bm z\in\operatorname{supp}\nu}\|\bm x(\bm z)\|<\infty$ and $\lambda_0=\lambda_{\min}(\bm Q_\nu)>0$. For every $\bm h\in\mathbb R^p$, the design norm satisfies
\[
\frac{\lambda_0}{B_x}\|\bm h\|\le m_\nu(\bm h)\le B_x\|\bm h\|,
\]
because $\int(\bm x^\top\bm h)^2\nu(d\bm z)\le B_x\|\bm h\|m_\nu(\bm h)$ and $\bm h^\top\bm Q_\nu\bm h\ge\lambda_0\|\bm h\|^2$. Consequently, $\operatorname{Var}\{\mathbb G(\bm t)-\mathbb G(\bm u)\}=g_0(a^\dagger)K_\tau m_\nu(\bm t-\bm u)\asymp\|\bm t-\bm u\|$. Gaussian moments and the Kolmogorov continuity criterion give a continuous modification on every compact ball. Gaussian entropy bounds on Euclidean balls show $E\sup_{\|\bm t\|\le R}|\mathbb G(\bm t)|\le C\sqrt R$; Gaussian concentration applied on dyadic balls then gives $\sup_{\|\bm t\|\le R}|\mathbb G(\bm t)|=o(R^2)$ almost surely as $R\to\infty$. Since $\bm H_{\rm eff}\succ0$, the continuous criterion $\mathcal Q$ attains its global maximum on a random compact ball. For distinct $\bm t$ and $\bm u$, $\operatorname{Var}\{\mathbb G(\bm t)-\mathbb G(\bm u)\}>0$; the Gaussian argmax uniqueness criterion of \citet{Kim1990} yields uniqueness almost surely on every fixed compact ball, and hence globally by compact exhaustion.
\end{proof}

\subsection{Random-center posterior localization}

\begin{lemma}[Random-argmax localization of the original posterior]
\label{lem:s4_argmax}
There exist measurable random vectors $\bm T_{n,\omega}\in\mathbb R^p$ such that $\bm T_{n,\omega}\Rightarrow\bm T^\star$ and, for every $\varepsilon>0$,
\begin{equation}\label{eq:s4_original_posterior_collapse}
\Pi_n\left\{\left\|r^{-1}(\bm\beta-\bm\beta_{n,\omega}^\circ)-\bm T_{n,\omega}\right\|>\varepsilon\mid D_n\right\}\xrightarrow{P_0}0.
\end{equation}
\end{lemma}
\begin{proof}
\emph{Construction and tightness of the argmax.} Let $B\subset\mathbb R^p$ be the fixed compact location ball used in Lemma~\ref{lem:s3_profile}, with $\bm\beta_{n,\omega}^\circ\in\operatorname{int}B$ eventually. The smoothly capped nuisance-integrated log likelihood is continuous on $B$. Choose a measurable maximizer $\widehat{\bm\beta}_n^{\rm sc}$ on $B$ (for example, the lexicographically first maximizer) and set $\bm T_{n,\omega}=r^{-1}(\widehat{\bm\beta}_n^{\rm sc}-\bm\beta_{n,\omega}^\circ)$. The negative-drift event constructed in the proof of Lemma~\ref{lem:new_full_shell} implies that, for every sufficiently large fixed $M$, the maximizer lies within distance $Mr$ of $\bm\beta_{n,\omega}^\circ$ except on an event with limiting upper probability at most $CM^{-3/2}$: outside this ball the local criterion is strictly smaller than at the center. Hence $\bm T_{n,\omega}=O_{P_0}(1)$. Lemma~\ref{lem:s4_functional} gives local uniform weak convergence and the limiting maximizer is unique almost surely. To apply the argmax theorem without treating compact convergence as global convergence, first choose $R$ with $P(\|\bm T^\star\|=R)=0$ and $P(\|\bm T^\star\|>R)$ arbitrarily small. The tightness just established makes $P_0(\|\bm T_{n,\omega}\|>R)$ comparably small after increasing $R$ if needed. On the closed ball $B_R$, \eqref{eq:s4_local_process} converges in $C(B_R)$; its limit has an almost surely unique maximizer and, outside an event of arbitrarily small probability, this maximizer is interior. The compact argmax continuous-mapping theorem \citep{Kim1990,vanDerVaart1996} gives convergence of the maximizers restricted to $B_R$. Their agreement with the unrestricted maximizers on the tightness event, followed by $R\uparrow\infty$, yields $\bm T_{n,\omega}\Rightarrow\bm T^\star$.

\emph{Posterior concentration around the sample argmax.} Fix $\varepsilon>0$ and a large $M$ so that the probability of $\|\bm T_{n,\omega}\|>M/2$ and the full posterior probability outside the $Mr$ neighborhood of $\bm\beta_{n,\omega}^\circ$ are arbitrarily small for large $n$, by Lemma~\ref{lem:new_full_shell}. On $\{\|\bm T_{n,\omega}\|\le M/2\}$, define the random criterion gap
\[
\Gamma_{n,M,\varepsilon}=\mathcal Q_{n,\omega}(\bm T_{n,\omega})-\sup\{\mathcal Q_{n,\omega}(\bm t):\|\bm t\|\le M,\ \|\bm t-\bm T_{n,\omega}\|\ge\varepsilon\}.
\]
On the compact ball, continuity and almost-sure uniqueness of the limiting maximizer imply a strictly positive separation between the maximum and the supremum outside its $\varepsilon$-neighborhood. By compact functional convergence and the argmax theorem, for every $\delta>0$ there is $\eta_\delta>0$ such that $\limsup_n P_0(\Gamma_{n,M,\varepsilon}<4\eta_\delta)<\delta$, after an arbitrarily small enlargement of $M$ to avoid the boundary of the localization event. Asymptotic equicontinuity in \eqref{eq:s4_modulus_bound} then provides a fixed $h\in(0,\varepsilon/2)$ such that, with probability at least $1-\delta$ asymptotically, $\mathcal Q_{n,\omega}(\bm t)\ge\mathcal Q_{n,\omega}(\bm T_{n,\omega})-\eta_\delta$ throughout $\|\bm t-\bm T_{n,\omega}\|\le h$. Since the Gaussian prior density is bounded below on $B$, integration over $\bm\beta=\bm\beta_{n,\omega}^\circ+r\bm t$ in this ball gives the denominator lower bound
\begin{equation}\label{eq:s4_argmax_denominator}
\int_B e^{\bar\ell_{n,\omega}^{\rm sc}(\bm\beta)}\pi_{\bm\beta}(d\bm\beta)\ge c_h r^p\exp\left[\bar\ell_{n,\omega}^{\rm sc}(\bm\beta_{n,\omega}^\circ)+a_k\{\mathcal Q_{n,\omega}(\bm T_{n,\omega})-\eta_\delta\}\right].
\end{equation}
The corresponding numerator over $\{\|\bm\beta-\bm\beta_{n,\omega}^\circ\|\le Mr,\ \|r^{-1}(\bm\beta-\bm\beta_{n,\omega}^\circ)-\bm T_{n,\omega}\|>\varepsilon\}$ is at most $\pi_{\bm\beta}(B)$ times the same exponential factor with $-4\eta_\delta$ in place of $-\eta_\delta$. Its ratio to \eqref{eq:s4_argmax_denominator} is at most $Cr^{-p}e^{-3\eta_\delta a_k}\to0$, because $a_k\to\infty$ and $\log(1/r)=o(a_k)$. Lemmas~\ref{lem:growing_contamination_reduction} and~\ref{lem:s3_spikes}, as applied in Lemma~\ref{lem:new_full_shell}, transfer this conditional capped-posterior assertion back to the original full posterior. Taking $n\to\infty$ and then making $M$ arbitrarily large and $\delta$ arbitrarily small proves \eqref{eq:s4_original_posterior_collapse}.
\end{proof}

\begin{lemma}[Uniform residual-quantile Bahadur expansion]
\label{lem:s4_bahadur}
Let the clean errors $\epsilon_i$, $i\in K_n$, be i.i.d. with $G_0(0)=\tau$, and let the deterministic regressors $\bm z_i$ be bounded. Put $\widehat G_{K,n}(u)=k^{-1}\sum_{i\in K_n}\mathbf1(\epsilon_i\le u)$ and $\bar{\bm z}_{K,n}=k^{-1}\sum_{i\in K_n}\bm z_i$. For $\bm d\in\mathbb R^{p-1}$, let $\widehat q_n(\bm d)=\widehat Q_\tau\{\epsilon_i-\bm z_i^\top\bm d:i\in K_n\}$, where $\widehat Q_\tau$ is the lower empirical quantile. Under (A1)--(A2), put $\omega_{g_0}(t)=\sup_{|v|\le t}|g_0(v)-g_0(0)|$; for every fixed $M<\infty$,
\begin{equation}\label{eq:s4_bahadur}
\sup_{\|\bm d\|\le Mr}\left|\widehat q_n(\bm d)+\bar{\bm z}_{K,n}^\top\bm d-\frac{\tau-\widehat G_{K,n}(0)}{g_0(0)}\right|=O_{P_0}\left\{r\,\omega_{g_0}(C_Mr)+\sqrt{\frac{r\log k}{k}}+\frac{\log k}{k}\right\}=o_{P_0}(r),
\end{equation}
\end{lemma}
\begin{proof}
For $\bm d\in\mathbb R^{p-1}$ and $u\in\mathbb R$, let $\widehat F_{n,\bm d}(u)=k^{-1}\sum_{i\in K_n}\mathbf1(\epsilon_i\le u+\bm z_i^\top\bm d)$ and $F_{n,\bm d}(u)=k^{-1}\sum_{i\in K_n}G_0(u+\bm z_i^\top\bm d)$. For fixed $C,M<\infty$, continuity of $g_0$ at zero and bounded design imply, uniformly over $|u|\le Cr$ and $\|\bm d\|\le Mr$,
\begin{equation}\label{eq:s4_bahadur_pop}
F_{n,\bm d}(u)-\tau=g_0(0)\{u+\bar{\bm z}_{K,n}^\top\bm d\}+O\{r\,\omega_{g_0}(C_Mr)\}.
\end{equation}
The centered empirical increments $\widehat F_{n,\bm d}(u)-\widehat G_{K,n}(0)-\{F_{n,\bm d}(u)-\tau\}$ form a fixed-dimensional VC class of affine threshold indicators, with envelope $\mathbf1\{|\epsilon_i|\le(C+B_zM)r\}$ and average second moment $O(r)$, where $B_z=\sup_{n,i}\|\bm z_i\|$. Symmetrization and the local VC maximal inequality \citep{vanDerVaart1996} yield
\begin{equation}\label{eq:s4_bahadur_emp}
\sup_{|u|\le Cr,\ \|\bm d\|\le Mr}\left|\widehat F_{n,\bm d}(u)-\widehat G_{K,n}(0)-\{F_{n,\bm d}(u)-\tau\}\right|=O_{P_0}\left\{\sqrt{\frac{r\log k}{k}}+\frac{\log k}{k}\right\}.
\end{equation}
Put $\delta_{n,M}=r\,\omega_{g_0}(C_Mr)+\sqrt{r\log k/k}+\log k/k=o(r)$ and $u_n^*(\bm d)=\{\tau-\widehat G_{K,n}(0)\}/g_0(0)-\bar{\bm z}_{K,n}^\top\bm d$. Since $\widehat G_{K,n}(0)-\tau=O_{P_0}(k^{-1/2})=o_{P_0}(r)$ and $\sup_{\|\bm d\|\le Mr}|\bar{\bm z}_{K,n}^\top\bm d|\le B_zMr$, one may choose $C$ large enough that $\sup_{\|\bm d\|\le Mr}|u_n^*(\bm d)|\le Cr/2$ with probability tending to one. Equations \eqref{eq:s4_bahadur_pop}--\eqref{eq:s4_bahadur_emp} show that, for every sufficiently large fixed $A$, with arbitrarily high probability for all large $n$,
\[
\widehat F_{n,\bm d}\{u_n^*(\bm d)-A\delta_{n,M}\}<\tau\le\widehat F_{n,\bm d}\{u_n^*(\bm d)+A\delta_{n,M}\},\qquad\|\bm d\|\le Mr.
\]
By monotonicity and the definition of the lower empirical quantile, $\sup_{\|\bm d\|\le Mr}|\widehat q_n(\bm d)-u_n^*(\bm d)|\le A\delta_{n,M}$ on this event. The constants can be selected for any prescribed error probability, proving \eqref{eq:s4_bahadur}. The uniform statement is required because the posterior slope and the residual quantile are computed from the same data.
\end{proof}

\begin{proof}[Proof of Theorem 2]
Write $\bm T_{n,\omega}$ for the random center of Lemma~\ref{lem:s4_argmax} and $\bm\beta_{n,\omega}^\circ=(\alpha_{n,\omega}^\circ,(\bm b_{n,\omega}^\circ)^\top)^\top$. Let $\bm P_b=(\bm0_{(p-1)\times1},\bm I_{p-1})$ and $\bm A=\begin{pmatrix}-\bm\mu_z^\top\bm P_b\\\bm P_b\end{pmatrix}$, where $\bm\mu_z=\int\bm z\,\nu(d\bm z)$; this is the same linear map as in the main manuscript. By Lemma~\ref{lem:new_moving_center}, $\|\bm b_{n,\omega}^\circ-\bm b_0\|=o(r)$, so \eqref{eq:s4_original_posterior_collapse} implies, for every $\varepsilon>0$,
\begin{equation}\label{eq:s4_slope_collapse}
\Pi_n^{\bm b}\left\{\left\|r^{-1}(\bm b-\bm b_0)-\bm P_b\bm T_{n,\omega}\right\|>\varepsilon\mid D_n\right\}\xrightarrow{P_0}0.
\end{equation}
Lemma~\ref{lem:s3_correction} gives $P_0(\widehat K_n=K_n)\to1$. On this event, for $\bm d=\bm b-\bm b_0$ the residual-quantile correction satisfies $T_n(\bm b)-\alpha_0=\widehat q_n(\bm d)$. For every fixed $M$, Lemma~\ref{lem:s4_bahadur} therefore gives, uniformly over $\|\bm b-\bm b_0\|\le Mr$,
\begin{equation}\label{eq:s4_corrected_expansion}
T_n(\bm b)-\alpha_0=-\bar{\bm z}_{K,n}^\top(\bm b-\bm b_0)+q_{K,n}+o_{P_0}(r),\qquad q_{K,n}=\frac{\tau-\widehat G_{K,n}(0)}{g_0(0)}=O_{P_0}(k^{-1/2})=o_{P_0}(r).
\end{equation}
Because $\bm T_{n,\omega}=O_{P_0}(1)$, Lemma~\ref{lem:s4_argmax} permits restriction of the posterior to such a ball for arbitrarily large fixed $M$, and (A2) gives $\bar{\bm z}_{K,n}\to\bm\mu_z$. Combining \eqref{eq:s4_slope_collapse} and \eqref{eq:s4_corrected_expansion}, uniformly over posterior-relevant slopes,
\[
\begin{aligned}
r^{-1}\{T_n(\bm b)-\alpha_0\}+\bm\mu_z^\top\bm P_b\bm T_{n,\omega}&=-\bar{\bm z}_{K,n}^\top\{r^{-1}(\bm b-\bm b_0)-\bm P_b\bm T_{n,\omega}\}\\
&\quad-(\bar{\bm z}_{K,n}-\bm\mu_z)^\top\bm P_b\bm T_{n,\omega}+q_{K,n}/r+o_{P_0}(1)\xrightarrow{\Pi_n^{\bm b},P_0}0.
\end{aligned}
\]
The slope coordinates obey \eqref{eq:s4_slope_collapse}. By the pushforward definition of $\Pi_n^{\rm cor}$ in \eqref{eq:s3_corrected_posterior}, these two statements prove $\Pi_n^{\rm cor}\{\|r^{-1}(\bm\beta^{\rm cor}-\bm\beta_0)-\bm A\bm T_{n,\omega}\|>\varepsilon\mid D_n\}\xrightarrow{P_0}0$. The weak convergence $\bm T_{n,\omega}\Rightarrow\bm T^\star$ was established in Lemma~\ref{lem:s4_argmax}.
\end{proof}

\subsection{Frequentist undercoverage}

\begin{proof}[Proof of Proposition 5]
Fix a contrast $\bm v\in\mathbb R^p$ satisfying $P(\bm v^\top\bm A\bm T^\star=0)=0$, and let $c_n=\bm v^\top\bm\beta_0+r\bm v^\top\bm A\bm T_{n,\omega}$. Theorem 2 implies that for every $\varepsilon>0$,
\begin{equation}\label{eq:s4_scalar_collapse}
\Pi_n^{\rm cor}\{|\bm v^\top\bm\beta^{\rm cor}-c_n|>\varepsilon r\mid D_n\}\xrightarrow{P_0}0.
\end{equation}
Let $\zeta\in(0,1)$ be fixed and let $I_{n,1-\zeta}^{\rm cor}(\bm v)$ be the equal-tailed credible interval for $\bm v^\top\bm\beta^{\rm cor}$. If the posterior probability in \eqref{eq:s4_scalar_collapse} is smaller than $\zeta/2$, both its lower $\zeta/2$ and upper $1-\zeta/2$ quantiles lie in $[c_n-\varepsilon r,c_n+\varepsilon r]$; hence
\[
P_0\{\bm v^\top\bm\beta_0\in I_{n,1-\zeta}^{\rm cor}(\bm v)\}\le P_0\{|\bm v^\top\bm A\bm T_{n,\omega}|\le\varepsilon\}+o(1).
\]
By $\bm T_{n,\omega}\Rightarrow\bm T^\star$, the portmanteau theorem gives $\limsup_nP_0\{|\bm v^\top\bm A\bm T_{n,\omega}|\le\varepsilon\}\le P\{|\bm v^\top\bm A\bm T^\star|\le\varepsilon\}$. Letting $\varepsilon\downarrow0$ and using $P(\bm v^\top\bm A\bm T^\star=0)=0$ proves the asserted convergence of frequentist coverage to zero.
\end{proof}

\section{Supplementary Details for Posterior Computation and Interval Calibration}
\label{sec:supp_computation}

This section provides the conditional distributions and numerical constructions underlying posterior sampling, residual-quantile correction, and interval calibration in Section 4 of the main manuscript. All posterior simulation uses the original AL--LPAL likelihood for the full sample; screening is performed only after posterior sampling. Write $\bm x_i=(1,\bm z_i^\top)^\top\in\mathbb R^p$, $\bm\beta=(\alpha,\bm b^\top)^\top$, and $\bm\xi=(\log\sigma,\operatorname{logit}s)^\top$. We retain the independent priors $\bm\beta\sim N_p(\bm\mu_\beta,\bm\Sigma_\beta)$, $\sigma\sim\mathrm{IG}(a_\sigma,b_\sigma)$, and $s\sim\mathrm{Beta}(a_s,b_s)$, with $\bm\Sigma_\beta\succ0$, and a fixed $\gamma>0$. Here $\mathrm{IG}(a,b)$ has density $b^a\Gamma(a)^{-1}v^{-a-1}\exp(-b/v)$ on $v>0$. All vectors and matrices are printed in bold, while observation-specific latent variables and scalar parameters are not.

\subsection{Full conditionals and partially collapsed Gibbs updates}
\label{subsec:s5_gibbs}

Put $c_\tau=\tau(1-\tau)$, $\kappa_1=(1-2\tau)/c_\tau$, and $\kappa_2^2=2/c_\tau$. The normal--exponential augmentation of the asymmetric Laplace distribution \citep{Kozumi2011}, combined with the log-Pareto scale mixture \citep{Hamura2022}, introduces $z_i\mid s\sim\mathrm{Bernoulli}(s)$ and $u_i\sim H(\cdot;\gamma)$, independently before conditioning on the data, and sets $\lambda_i=u_i^{z_i}$. Conditional on these variables, the augmented observation model is
\begin{equation}\label{eq:s5_augmentation}
\theta_i\mid\sigma,\lambda_i\sim\mathrm{Exp}\!\left(\frac{1}{\sigma\lambda_i}\right),\qquad Y_i\mid\theta_i,\lambda_i,\bm\beta,\sigma\sim N\!\left(\bm x_i^\top\bm\beta+\kappa_1\theta_i,\;\kappa_2^2\sigma\lambda_i\theta_i\right),
\end{equation}
where the argument of $\mathrm{Exp}$ is its rate. Integrating out $\theta_i$ yields $f_{\mathrm{AL},\tau}(Y_i-\bm x_i^\top\bm\beta;0,\sigma\lambda_i)$; further integration over $(z_i,u_i)$ yields the original AL--LPAL working likelihood. Denote the current residual by $e_i=Y_i-\bm x_i^\top\bm\beta$ and let $d_i=\rho_\tau(e_i)/\sigma\ge0$.

\emph{Collapsed allocation and mixing-scale updates.} The partially collapsed step integrates out $\theta_i$ before drawing $(z_i,u_i)$. Given the current $u_i$, the conditional allocation probability is
\begin{equation}\label{eq:s5_allocation}
P(z_i=1\mid u_i,\bm\beta,\sigma,s,Y_i)=\frac{s u_i^{-1}\exp(-d_i/u_i)}{(1-s)\exp(-d_i)+s u_i^{-1}\exp(-d_i/u_i)}.
\end{equation}
The conditional density of $u_i$, after sampling $z_i$ and still integrating out $\theta_i$, is proportional to $H(u_i;\gamma)$ if $z_i=0$, and to $u_i^{-1}H(u_i;\gamma)\exp(-d_i/u_i)$ if $z_i=1$. In the former case $u_i$ is sampled directly from $H$; indeed, $P_H(u_i\le v)=1-\{1+\log(1+v)\}^{-\gamma}$. In the latter case, the transformation $h_i=\log u_i$ gives the conditional log density
\begin{equation}\label{eq:s5_log_u_conditional}
\log\pi(h_i\mid z_i=1,\bm\beta,\sigma,Y_i)=\log H(e^{h_i};\gamma)-d_i e^{-h_i}+C,
\end{equation}
which can be updated by univariate slice sampling \citep{Neal2003}. The Jacobian $e^{h_i}$ in this change of variables cancels the factor $u_i^{-1}$ in the conditional density. In particular, equation \eqref{eq:s5_log_u_conditional} is the transformed conditional density for $z_i=1$, not for $z_i=0$. The densities are understood for nonzero residuals, as appropriate for the continuous posterior; exact zero residuals may require numerical safeguards because the LPAL marginal density is unbounded at zero.

\emph{Refreshing the exponential variables.} After updating every $(z_i,u_i)$, regenerate $\theta_i$ from its full conditional under \eqref{eq:s5_augmentation}. Expanding the Gaussian quadratic form gives
\[
\begin{aligned}
\pi(\theta_i\mid-)&\propto\theta_i^{-1/2}\exp\left\{-\frac12\left(\frac{\chi_i}{\theta_i}+\psi_i\theta_i\right)\right\},\\
\chi_i&=\frac{e_i^2}{\kappa_2^2\sigma\lambda_i},\qquad\psi_i=\frac{2+\kappa_1^2/\kappa_2^2}{\sigma\lambda_i},
\end{aligned}
\]
so that $\theta_i\mid-\sim\mathrm{GIG}(1/2,\chi_i,\psi_i)$ under the convention $\mathrm{GIG}(a,\chi,\psi)\propto t^{a-1}\exp\{-(\chi/t+\psi t)/2\}$. The refresh must follow, rather than precede, the collapsed $(z_i,u_i)$ update: the subsequent conditional parameter draws then use the joint augmented model. This ordering follows the partially collapsed Gibbs principle \citep{vanDykPark2008}.

\emph{Scale and regression-coefficient updates.} Define $\bm y=(Y_1,\ldots,Y_n)^\top$, $\bm\theta=(\theta_1,\ldots,\theta_n)^\top$, and the design matrix $\bm X$ with rows $\bm x_i^\top$. The $\sigma$-dependent part of the augmented likelihood is proportional to
\[
\sigma^{-3n/2}\exp\left[-\frac1\sigma\sum_{i=1}^n\left\{\frac{\theta_i}{\lambda_i}+\frac{(Y_i-\bm x_i^\top\bm\beta-\kappa_1\theta_i)^2}{2\kappa_2^2\lambda_i\theta_i}\right\}\right].
\]
Consequently,
\[
\sigma\mid-\sim\mathrm{IG}\left(a_\sigma+\frac{3n}{2},\;b_\sigma+\sum_{i=1}^n\left\{\frac{\theta_i}{\lambda_i}+\frac{(Y_i-\bm x_i^\top\bm\beta-\kappa_1\theta_i)^2}{2\kappa_2^2\lambda_i\theta_i}\right\}\right).
\]
Set $\bm D=\operatorname{diag}\{(\kappa_2^2\sigma\lambda_i\theta_i)^{-1}:i=1,\ldots,n\}$. Completing the square in $\bm\beta$ yields
\[
\begin{aligned}
\bm\beta\mid-&\sim N_p(\bm m_\beta,\bm V_\beta),\\
\bm V_\beta&=(\bm\Sigma_\beta^{-1}+\bm X^\top\bm D\bm X)^{-1},\\
\bm m_\beta&=\bm V_\beta\{\bm\Sigma_\beta^{-1}\bm\mu_\beta+\bm X^\top\bm D(\bm y-\kappa_1\bm\theta_{\mathrm{lat}})\}.
\end{aligned}
\]
Finally, the allocation indicators give the conjugate update
\[
s\mid-\sim\mathrm{Beta}\left(a_s+\sum_{i=1}^n z_i,\;b_s+n-\sum_{i=1}^n z_i\right).
\]
A complete Gibbs sweep updates $z_i$ and $u_i$ using the collapsed conditionals for all $i$, refreshes all $\theta_i$, and then updates $\sigma$, $\bm\beta$, and $s$ in this order. No observations are discarded during these updates. In implementations, the probability in \eqref{eq:s5_allocation} should be evaluated on the log scale to avoid underflow for large residuals; this does not change its target conditional distribution.

\subsection{Residual-quantile correction and screening}
\label{subsec:s5_quantile_correction}

Suppose that $T$ Gibbs iterations are performed and the first $B$ are discarded. For each retained draw $\bm\beta^{(t)}=(\alpha^{(t)},(\bm b^{(t)})^\top)^\top$, $t=B+1,\ldots,T$, the correction uses a single screened index set $\widehat K_n$ constructed after posterior sampling. The theoretical rule studied in Sections S3--S4 is $\widehat K_n=\{i:|Y_i|\le n^{1/q}\}$ for the fixed $q\in(1,2)$ of (A1). For the numerical rule in Section 4 of the main manuscript, first set $\widehat{\bm b}=\operatorname{med}\{\bm b^{(t)}:t>B\}$ coordinatewise and define $R_i^{\mathrm{scr}}=Y_i-\bm z_i^\top\widehat{\bm b}$. Then retain
\[
\widehat K_n=\left\{i:\left|R_i^{\mathrm{scr}}-\operatorname{med}_{j}R_j^{\mathrm{scr}}\right|\le(\log n)\operatorname{MAD}(R_1^{\mathrm{scr}},\ldots,R_n^{\mathrm{scr}})\right\},
\]
where $\operatorname{MAD}$ is the conventionally scaled median absolute deviation \citep{Rousseeuw1993}. In particular, the practical rule is based on residuals computed at the original posterior slope pilot and is not a draw-dependent screening operation. Its selection properties are not asserted to follow from the exact-recovery result for theoretical screening.

For a nonempty index set $K$ and observations $v_i$, define the lower empirical quantile by
\[
\widehat Q_\tau\{v_i:i\in K\}=\inf\left\{v:\frac1{|K|}\sum_{i\in K}\mathbf1(v_i\le v)\ge\tau\right\}.
\]
For $K=\widehat K_n$, write $T_n(\bm b)=\widehat Q_\tau\{Y_i-\bm z_i^\top\bm b:i\in\widehat K_n\}$, with the convention $T_n(\bm b)=0$ if $\widehat K_n$ is empty. For every retained original draw, calculate
\[
\Delta_n^{(t)}=\alpha^{(t)}-T_n(\bm b^{(t)}),\qquad\bm\beta^{\mathrm{cor},(t)}=\left(T_n(\bm b^{(t)}),(\bm b^{(t)})^\top\right)^\top.
\]
These draws implement the pushforward posterior described in Section 3.3 of the main manuscript: the original slope draws are unchanged, whereas the intercept is replaced by a residual quantile computed on the common screened set. In particular, this operation does not refit the likelihood and does not alter the full-data MCMC target. For reporting point estimates, the procedure uses $\widehat{\bm b}$ and $T_n(\widehat{\bm b})$, which need not equal the coordinatewise medians of all corrected draws.

\subsection{Finite-scale Gaussian-process covariance and curvature}
\label{subsec:s5_gp_calibration}

Theorem 2 identifies the limiting Gaussian-process argmax law associated with the random center of the corrected posterior. The numerical interval procedure in Section 4 uses a plug-in, finite-scale approximation to this sampling law; it is not obtained by enlarging corrected-posterior credible intervals. Let $(\widehat{\bm\beta},\widehat\sigma,\widehat s)$ be coordinatewise original-posterior medians and put $\widehat{\bm\xi}=(\log\widehat\sigma,\operatorname{logit}\widehat s)^\top$, $k=|\widehat K_n|>0$, $\widehat\rho=(n-k)/k$. Let $r=r_k$ solve $kr^3\{\log(1/r)\}^2=1$, and for $i\in\widehat K_n$ define the uncorrected pilot residual $R_i=Y_i-\bm x_i^\top\widehat{\bm\beta}$. In particular, $R_i$ is centered at the original working-posterior pilot, not at the quantile-corrected intercept.

\emph{Smoothed residual law and increment covariance.} For a data-driven bandwidth $h>0$, the Gaussian-convolved empirical residual distribution is
\[
\widehat G_h(dv)=\frac1k\sum_{i\in\widehat K_n}\frac1h\phi\left(\frac{v-R_i}{h}\right)\,dv,
\]
where $\phi$ is the standard normal density. Write $\ell_{\bm\xi}(v)=\log\{\sigma^{-1}f_\tau(v/\sigma;s,\gamma)\}$, where $\bm\xi=(\log\sigma,\operatorname{logit}s)^\top$, and define $D_r(v,t)=\ell_{\widehat{\bm\xi}}(v-rt)-\ell_{\widehat{\bm\xi}}(v)$. For projected shifts $t,u\in\mathbb R$, put
\[
\begin{aligned}
\widehat C_{r,h}(t,u)&=\operatorname{Cov}_{V\sim\widehat G_h}\{D_r(V,t),D_r(V,u)\}\\
&=\int D_r(v,t)D_r(v,u)\,\widehat G_h(dv)-\left\{\int D_r(v,t)\,\widehat G_h(dv)\right\}\left\{\int D_r(v,u)\,\widehat G_h(dv)\right\}.
\end{aligned}
\]
The integrals may equivalently be evaluated as averages over $V=R_i+hZ$, where $Z\sim N(0,1)$. They involve the original AL--LPAL log density, including its logarithmic central behavior. The Gaussian convolution makes the expectations well defined and provides a smooth pilot criterion for subsequent differentiation; the actual bandwidth choice and quadrature resolution are numerical inputs rather than consequences of Theorem 2.

\emph{Plug-in efficient curvature.} Define the smoothed clean scalar criterion
\[
\widehat F_{0,h}(a,\bm\xi)=\int\log\left\{\sigma^{-1}f_\tau\left(\frac{v-a}{\sigma};s,\gamma\right)\right\}\,\widehat G_h(dv).
\]
Let $\widehat h_{aa}$, $\widehat{\bm h}_{a\xi}$, $\widehat{\bm h}_{\xi a}$, and $\widehat{\bm h}_{\xi\xi}$ denote the corresponding blocks of $-\nabla^2\widehat F_{0,h}$ evaluated at $(a,\bm\xi)=(0,\widehat{\bm\xi})$. With
\[
\widehat{\bm Q}=\frac1k\sum_{i\in\widehat K_n}\bm x_i\bm x_i^\top,\qquad\overline{\bm x}=\frac1k\sum_{i\in\widehat K_n}\bm x_i,\qquad\widehat{\bm J}_{\xi\xi}=\widehat{\bm h}_{\xi\xi}+\widehat\rho\operatorname{diag}\{0,\widehat s(1-\widehat s)\},
\]
define, when $\widehat{\bm J}_{\xi\xi}$ is invertible,
\begin{equation}\label{eq:s5_effective_curvature}
\widehat{\bm H}_{\mathrm{eff}}=\widehat h_{aa}\widehat{\bm Q}-\left(\widehat{\bm h}_{a\xi}\widehat{\bm J}_{\xi\xi}^{-1}\widehat{\bm h}_{\xi a}\right)\overline{\bm x}\,\overline{\bm x}^\top.
\end{equation}
The additional term in $\widehat{\bm J}_{\xi\xi}$ is the negative Hessian, in $(\log\sigma,\operatorname{logit}s)$ coordinates, of the estimated contamination contribution $\widehat\rho\log s$. Equation \eqref{eq:s5_effective_curvature} is the plug-in counterpart of the Schur-complement curvature in the nonregular limit. Numerical inversion, covariance regularization, and positive-curvature stabilization, if applied, must be recorded with the actual computational controls; no fixed regularization thresholds are imposed here.

\emph{Gaussian simulation and numerical argmax.} For each $j=1,\ldots,J$, generate independent centered Gaussian paths $\{\mathbb V_i^{(j)}(t):t\in\mathbb R\}$, indexed by $i\in\widehat K_n$, each with covariance $\widehat C_{r,h}$. Their aggregate path is
\[
\widehat{\mathbb G}_r^{(j)}(\bm t)=\frac{1}{kr^2}\sum_{i\in\widehat K_n}\mathbb V_i^{(j)}(\bm x_i^\top\bm t),\qquad\bm t\in\mathcal T\subset\mathbb R^p,
\]
where $\mathcal T$ is a bounded numerical search region. For a finite grid $\mathcal T_M=\{\bm t_1,\ldots,\bm t_M\}\subset\mathcal T$, the vector $\{\widehat{\mathbb G}_r^{(j)}(\bm t_a)\}_{a=1}^M$ may instead be drawn directly from a centered multivariate normal distribution with covariance matrix
\[
\widehat{\bm\Sigma}_{\mathbb G}[a,b]=\frac{1}{k^2r^4}\sum_{i\in\widehat K_n}\widehat C_{r,h}(\bm x_i^\top\bm t_a,\bm x_i^\top\bm t_b),\qquad a,b=1,\ldots,M.
\]
The direct multivariate simulation and the observationwise path construction have exactly the same finite-grid Gaussian law when their covariances are evaluated identically; this equivalence does not remove discretization error in approximating a continuous argmax. For each replicate compute, on the chosen numerical representation of $\mathcal T$,
\[
\widehat{\bm T}^{*(j)}=\operatorname*{argmax}_{\bm t\in\mathcal T}\left\{\widehat{\mathbb G}_r^{(j)}(\bm t)-\frac12\bm t^\top\widehat{\bm H}_{\mathrm{eff}}\bm t\right\},\qquad\bm b_{\mathrm{cal}}^{(j)}=\widehat{\bm b}-r[\widehat{\bm T}^{*(j)}]_{\mathrm{slope}},
\]
where $[\cdot]_{\mathrm{slope}}$ extracts the last $p-1$ coordinates. The negative sign constructs a parameter-side reference draw by reversing the simulated sampling error. Numerical searches restricted to $\mathcal T_M$ approximate the bounded-region argmax; the search region, grid or interpolation scheme, and optimization tolerances are implementation choices not fixed by the theoretical limit.

\subsection{Consistency of the Gaussian-process calibration reference}
\label{subsec:s5_gp_consistency}

We establish the conditional consistency of the plug-in Gaussian-process reference stated in Proposition 6 of the main manuscript. Throughout this subsection, assumptions (A1)--(A5) hold, the theoretical screening rule is used, and $h=h_k=k^{-1/8}$ is deterministic. Write $r=r_k$, $L_r=\log(1/r)$, and $kr^3L_r^2=1$. We work on the event $\{\widehat K_n=K_n\}$, whose probability tends to one, and treat the original-posterior pilot medians as exact posterior functionals. The posterior localization results of Sections S3--S4 imply $\widehat{\bm\beta}\xrightarrow{P_0}\bm\beta^\dagger$ and $\widehat{\bm\xi}\xrightarrow{P_0}\bm\xi^\dagger$. All covariance integrals and Gaussian simulations below are defined exactly; approximations from numerical quadrature or fixed Monte Carlo replication are not part of the assertion.

\begin{lemma}[Consistency of the plug-in Gaussian-process covariance]
\label{lem:s5_gp_covariance_consistency}
For every compact $T\subset\mathbb R^p$, define
\[
\widehat\Sigma_{\mathbb G,r,h}(\bm t,\bm u)=\frac{1}{k^2r^4}\sum_{i\in K_n}\widehat C_{r,h}(\bm x_i^\top\bm t,\bm x_i^\top\bm u),\qquad\bm t,\bm u\in T.
\]
Then
\begin{equation}\label{eq:s5_gp_covariance_consistency}
\sup_{\bm t,\bm u\in T}\left|\widehat\Sigma_{\mathbb G,r,h}(\bm t,\bm u)-\operatorname{Cov}\{\mathbb G(\bm t),\mathbb G(\bm u)\}\right|\xrightarrow{P_0}0.
\end{equation}
Moreover, the smoothed pilot-residual density $\widehat g_h(v)=k^{-1}\sum_{i\in K_n}h^{-1}\phi\{(v-R_i)/h\}$ satisfies $\|\widehat g_h\|_\infty=O_{P_0}(1)$.
\end{lemma}

\begin{proof}
For a fixed compact neighborhood $B$ of $\bm\beta^\dagger$, define $R_i(\bm\beta)=Y_i-\bm x_i^\top\bm\beta$ and $\widehat g_{\bm\beta,h}(v)=k^{-1}\sum_{i\in K_n}h^{-1}\phi\{(v-R_i(\bm\beta))/h\}$. For each fixed $(\bm\beta,v)$ the summands are bounded by $C/h$, and their Lipschitz constant in $(\bm\beta,v)$ is at most $C/h^2$, uniformly in the observations, by (A2). By (A1), $P_0\{\max_{i\in K_n}|\epsilon_i|>k^2/2\}\le Ck^{1-2q}=o(1)$. On the complementary event, the residuals remain in a deterministic interval of polynomial length, uniformly over $\bm\beta\in B$; outside a slightly larger interval the empirical Gaussian kernels are negligible. Their expectations are also negligible there by the $q$th-moment bound and the Gaussian tail. A net of mesh $k^{-1}$ over the resulting bounded $(\bm\beta,v)$-domain has polynomial cardinality because $p$ is fixed. Hoeffding's inequality, a union bound, and the Lipschitz bound give
\begin{equation}\label{eq:s5_gp_uniform_kde}
\sup_{\bm\beta\in B,\,v\in\mathbb R}\left|\widehat g_{\bm\beta,h}(v)-E_0\widehat g_{\bm\beta,h}(v)\right|=O_{P_0}\left\{\frac1h\sqrt{\frac{\log k}{k}}+\frac1{kh^2}\right\}=o_{P_0}(1).
\end{equation}
The expectation is $E_0\widehat g_{\bm\beta,h}(v)=k^{-1}\sum_{i\in K_n}(g_0*\phi_h)\{v+\bm x_i^\top(\bm\beta-\bm\beta_0)\}$ and is bounded by $\|g_0\|_\infty$. Since $\widehat{\bm\beta}\to^{P_0}\bm\beta^\dagger=(\alpha_0+a^\dagger,\bm b_0^\top)^\top$, bounded design implies $\max_i|\bm x_i^\top(\widehat{\bm\beta}-\bm\beta_0)-a^\dagger|=o_{P_0}(1)$. By (A4), $g_0$ is continuous near $a^\dagger$. The Gaussian approximate identity, applied locally, and \eqref{eq:s5_gp_uniform_kde} yield for some $\delta>0$,
\begin{equation}\label{eq:s5_gp_local_kde}
\sup_{|v|\le\delta}|\widehat g_h(v)-g_0(v+a^\dagger)|\xrightarrow{P_0}0,\qquad\|\widehat g_h\|_\infty=O_{P_0}(1).
\end{equation}

For a fixed interior compact set $\Xi$ containing $\bm\xi^\dagger$, put $\Delta_{r,t,\bm\xi}(v)=\ell_{\bm\xi}(v-rt)-\ell_{\bm\xi}(v)$. Lemmas~\ref{lem:s4_kernel} and \ref{lem:s4_fidi} establish, uniformly over bounded projected shifts and $\bm\xi\in\Xi$, the $L^2$ translation bound $\int\Delta_{r,t,\bm\xi}(v)^2\,dv\le C_Hr/L_r^2$ and the limiting central increment kernel, whose leading term does not depend on $\bm\xi$. In the region $|v|\le\delta$, Cauchy--Schwarz and \eqref{eq:s5_gp_local_kde} show that the normalized difference between second moments under $\widehat g_h$ and $g^\dagger(v)=g_0(v+a^\dagger)$ is at most $C_H\sup_{|v|\le\delta}|\widehat g_h(v)-g^\dagger(v)|=o_{P_0}(1)$. On $|v|>\delta$, the spatial log-density derivative is uniformly bounded for shifts of size at most $rH$, so both second moments contribute $O_{P_0}(r^2)$ and vanish after multiplication by $L_r^2/r$. The same argument remains uniform when $\bm\xi$ is replaced by the consistent random pilot $\widehat{\bm\xi}$. To make the uniformity over all projected shifts explicit, apply the finite-net upgrade and uniform $L^2$ increment modulus established in Lemma~\ref{lem:s4_fidi}(i)--(ii); it controls the entire square $[-H,H]^2$, not only individual covariance entries. Because the true and smoothed residual laws are bounded uniformly and the rescaled outer contribution is negligible, the resulting random covariance functions converge in the supremum norm on that square.

The centering correction is negligible as well. Near the origin, \eqref{eq:s4_central_score} and direct integration of the logarithmic singularity give $\int_{|v|\le\delta}|\Delta_{r,t,\bm\xi}(v)|\,dv\le C_Hr\{1+\log L_r\}$ uniformly on the same nuisance compact and $|t|\le H$. This is a \emph{local} Lebesgue-integral estimate; it cannot be extended to the entire real line, since $|\Delta_{r,t,\bm\xi}(v)|\asymp r|t|/|v|$ in the tails when $t\ne0$, whose unweighted integral diverges. Instead, on $|v|>\delta$ the spatial score is uniformly bounded along shifts of size $rH$, giving $|\Delta_{r,t,\bm\xi}(v)|\le C_{\delta,H}r$ there. As $\widehat g_h$ is a probability density and $\|\widehat g_h\|_\infty=O_{P_0}(1)$, split the \emph{weighted} integral into these two regions to obtain
\[
\sup_{|t|\le H}\left|\int_{\mathbb R}\Delta_{r,t,\widehat{\bm\xi}}(v)\widehat g_h(v)\,dv\right|
\le \|\widehat g_h\|_\infty C_Hr\{1+\log L_r\}+C_{\delta,H}r
=O_{P_0}\{r(1+\log L_r)\}.
\]
Consequently, the product of two means, after multiplication by $L_r^2/r$, is $O_{P_0}\{rL_r^2(1+\log L_r)^2\}=o_{P_0}(1)$. Lemma~\ref{lem:s4_kernel} therefore gives
\begin{equation}\label{eq:s5_gp_projected_covariance}
\sup_{|t|,|u|\le H}\left|\frac{L_r^2}{r}\widehat C_{r,h}(t,u)-\frac{g_0(a^\dagger)K_\tau}{2}\{|t|+|u|-|t-u|\}\right|\xrightarrow{P_0}0,
\end{equation}
where $K_\tau=\pi^2+\{\log(\tau/(1-\tau))\}^2$. For $\bm t,\bm u\in T$, the projected shifts are bounded uniformly in $i$ by (A2). Using $kr^3L_r^2=1$ and the convergence of the empirical clean-design distribution in (A2), the uniform limit \eqref{eq:s5_gp_projected_covariance} yields \eqref{eq:s5_gp_covariance_consistency}.
\end{proof}

\begin{lemma}[Consistency of the plug-in efficient curvature]
\label{lem:s5_gp_curvature_consistency}
The smoothed scalar Hessian and the efficient-curvature estimator of \eqref{eq:s5_effective_curvature} satisfy
\begin{equation}\label{eq:s5_gp_curvature_consistency}
-\nabla^2_{(a,\bm\xi)}\widehat F_{0,h}(0,\widehat{\bm\xi})\xrightarrow{P_0}-\nabla^2_{(a,\bm\xi)}F_0(a^\dagger,\bm\xi^\dagger),\qquad\widehat{\bm H}_{\mathrm{eff}}\xrightarrow{P_0}\bm H_{\mathrm{eff}}.
\end{equation}
\end{lemma}

\begin{proof}
Write $\Psi_h(v,\bm\xi)=E_Z\ell_{\bm\xi}(v+hZ)$, where $Z\sim N(0,1)$, and let $\mathsf H_h(v,\bm\xi)$ be the Hessian with respect to $(a,\bm\xi)$ of $\Psi_h(v-a,\bm\xi)$ at $a=0$. On every fixed interior nuisance compact $\Xi$, the central expansions in Lemma~\ref{lem:s4_kernel} and the differentiated nuisance bounds in Lemma~\ref{lem:s3_nuisance} permit differentiation under Gaussian convolution. The logarithmic-logarithmic central envelope is locally integrable, and the log density and its fixed-order nuisance derivatives have at most polynomial growth in the tails. Differentiating the Gaussian kernel in the location argument, and splitting its integral between a fixed origin neighborhood and its complement, gives the deliberately loose bounds
\begin{equation}\label{eq:s5_gp_hessian_envelopes}
\begin{aligned}
\sup_{v\in\mathbb R,\,\bm\xi\in\Xi}\|\mathsf H_h(v,\bm\xi)\|&\le Ch^{-2}\{1+\log(1/h)\},\\
\sup_{v\in\mathbb R,\,\bm\xi\in\Xi}\|\nabla_{(v,\bm\xi)}\mathsf H_h(v,\bm\xi)\|&\le Ch^{-3}\{1+\log(1/h)\}.
\end{aligned}
\end{equation}
The spatial derivatives of order $j\le3$ are bounded by $Ch^{-j}\{1+\log(1/h)\}$ because the $j$th differentiated Gaussian kernel has integral zero for $j\ge1$ and its convolution with the local singularity is finite. The bounds for mixed and nuisance derivatives follow by differentiating the interior nuisance parametrization, using the bounds of Lemma~\ref{lem:s3_nuisance}; all constants are uniform on $\Xi$.

Fix small compact neighborhoods $B$ of $\bm\beta^\dagger$ and $\Xi$ of $\bm\xi^\dagger$. For deterministic $(\bm\beta,\bm\xi)\in B\times\Xi$, the empirical Hessian has independent summands $\mathsf H_h\{R_i(\bm\beta),\bm\xi\}$, each bounded by $M_h=Ch^{-2}\{1+\log(1/h)\}$. Cover $B\times\Xi$ with a net of mesh $k^{-1}$ and polynomial cardinality. Hoeffding's inequality and the second bound in \eqref{eq:s5_gp_hessian_envelopes}, combined with (A2), yield
\begin{equation}\label{eq:s5_gp_uniform_hessian}
\begin{aligned}
\sup_{(\bm\beta,\bm\xi)\in B\times\Xi}\bigg\|\frac1k\sum_{i\in K_n}\big[\mathsf H_h\{R_i(\bm\beta),\bm\xi\}-E_0\mathsf H_h\{R_i(\bm\beta),\bm\xi\}\big]\bigg\|&=O_{P_0}\left\{\frac{\log(1/h)}{h^2}\sqrt{\frac{\log k}{k}}+\frac{\log(1/h)}{kh^3}\right\}\\
&=o_{P_0}(1).
\end{aligned}
\end{equation}
In the last step we used $h=k^{-1/8}$. Because the bound holds uniformly in $\bm\beta$ and $\bm\xi$, it remains valid when these parameters are replaced by the data-dependent original-posterior pilot; no sample splitting or quantitative pilot rate is required.

Define $F_{0,h}(a,\bm\xi)=E_ZF_0(a-hZ,\bm\xi)$. By (A5), $F_0=F_\rho-\rho\log s$ is $C^2$ with uniformly continuous Hessian in a neighborhood of $(a^\dagger,\bm\xi^\dagger)$. Its local Hessian is therefore recovered by Gaussian approximate identities. To justify passage through the unbounded location domain, multiply $F_0$ by a smooth cutoff supported in a slightly larger neighborhood before convolving; the cutoff part converges in $C^2$, while the remaining part is separated from the evaluation region and its differentiated Gaussian convolution is negligible by Gaussian tail decay and the polynomial growth of the population log likelihood and its nuisance derivatives. Hence for a sufficiently small compact neighborhood $N$ of $(a^\dagger,\bm\xi^\dagger)$,
\begin{equation}\label{eq:s5_gp_population_hessian}
\sup_{(a,\bm\xi)\in N}\|\nabla^2F_{0,h}(a,\bm\xi)-\nabla^2F_0(a,\bm\xi)\|\longrightarrow0.
\end{equation}
For each deterministic $\bm\beta\in B$, write $a_i(\bm\beta)=\bm x_i^\top(\bm\beta-\bm\beta_0)$. Independence gives $E_0\mathsf H_h\{R_i(\bm\beta),\bm\xi\}=\nabla^2F_{0,h}\{a_i(\bm\beta),\bm\xi\}$. Pilot consistency and bounded design imply $\max_{i\in K_n}|a_i(\widehat{\bm\beta})-a^\dagger|\xrightarrow{P_0}0$ and $\widehat{\bm\xi}\xrightarrow{P_0}\bm\xi^\dagger$. Thus \eqref{eq:s5_gp_population_hessian} and Hessian continuity yield convergence of the average expected Hessian at the pilot. Combining this with \eqref{eq:s5_gp_uniform_hessian} proves the first limit in \eqref{eq:s5_gp_curvature_consistency}.

Finally, by (A2) and theoretical screening, $\widehat{\bm Q}\to^{P_0}\bm Q_\nu$ and $\overline{\bm x}\to^{P_0}\overline{\bm x}_\nu$, while (A3) and pilot consistency give $\widehat\rho\to\rho$ and $\widehat s\to^{P_0}s^\dagger$. The limiting negative nuisance Hessian $\bm J_{\xi\xi}=\bm h_{\xi\xi}+\rho\operatorname{diag}\{0,s^\dagger(1-s^\dagger)\}$ is positive definite by (A5). Consequently $\widehat{\bm J}_{\xi\xi}^{-1}\xrightarrow{P_0}\bm J_{\xi\xi}^{-1}$, and the Schur-complement formula \eqref{eq:s5_effective_curvature} gives the second limit in \eqref{eq:s5_gp_curvature_consistency}.
\end{proof}

\begin{proof}[Proof of Proposition 6]
Let $\widehat{\mathbb G}_k^*$ be a conditionally centered Gaussian process with covariance $\widehat\Sigma_{\mathbb G,r,h}$ and $\widehat{\mathbb G}_k^*(\bm0)=0$. The covariance is nonnegative definite because it is a finite sum of increment covariance kernels. For every fixed compact ball $B_R\subset\mathbb R^p$, Lemma~\ref{lem:s5_gp_covariance_consistency} implies conditional convergence of all finite-dimensional Gaussian distributions to those of $\mathbb G$. We next verify conditional tightness rather than relying on covariance convergence alone. Define $\widehat d_k^2(\bm t,\bm u)=E^*|\widehat{\mathbb G}_k^*(\bm t)-\widehat{\mathbb G}_k^*(\bm u)|^2$, where $E^*$ is expectation under the conditional GP law. On the event $\{\|\widehat g_h\|_\infty\le M,\widehat{\bm\xi}\in\Xi\}$, Lemma~\ref{lem:s4_fidi}(ii), applied with the uniform density bound from Lemma~\ref{lem:s5_gp_covariance_consistency}, gives
\begin{equation}\label{eq:s5_gp_conditional_metric}
\widehat d_k^2(\bm t,\bm u)\le\frac{CM}{k^2r^4}\sum_{i\in K_n}\frac{r|\bm x_i^\top(\bm t-\bm u)|}{\log^2\{e/(r|\bm x_i^\top(\bm t-\bm u)|)\}}\le C_M\|\bm t-\bm u\|,
\end{equation}
for $\bm t,\bm u\in B_{R_k}$, $R_k=\log k$, with the fraction interpreted as zero when the projected increment vanishes. The second inequality follows from bounded design, $rR_k\to0$, and $kr^3L_r^2=1$, since $\log\{e/(2B_xrR_k)\}\asymp L_r$. The event has probability tending to one after first choosing $M$ sufficiently large and a fixed interior compact $\Xi$ containing $\bm\xi^\dagger$. Standard Gaussian chaining and entropy bounds \citep{vanDerVaart1996,Adler2007}, together with the metric bound \eqref{eq:s5_gp_conditional_metric}, give asymptotic conditional equicontinuity and tightness in $C(B_R)$ for each fixed $R$. Together with finite-dimensional convergence, this yields
\begin{equation}\label{eq:s5_gp_conditional_process}
\mathcal L(\widehat{\mathbb G}_k^*\mid D_n)\Longrightarrow^{P_0}\mathcal L(\mathbb G)\qquad\text{in }C(B_R).
\end{equation}
Lemma~\ref{lem:s5_gp_curvature_consistency} gives $\widehat{\bm H}_{\mathrm{eff}}\xrightarrow{P_0}\bm H_{\mathrm{eff}}\succ0$. Thus the conditional Slutsky theorem extends \eqref{eq:s5_gp_conditional_process} to the penalized processes $\widehat{\mathcal Q}_k^*(\bm t)=\widehat{\mathbb G}_k^*(\bm t)-\bm t^\top\widehat{\bm H}_{\mathrm{eff}}\bm t/2$ and $\mathcal Q(\bm t)=\mathbb G(\bm t)-\bm t^\top\bm H_{\mathrm{eff}}\bm t/2$ in $C(B_R)$.

It remains to exclude maximizers escaping to infinity and to control the grid. Put $\lambda_*=\lambda_{\min}(\bm H_{\mathrm{eff}})>0$. On an event with probability tending to one, $\widehat{\bm H}_{\mathrm{eff}}\succeq(\lambda_*/2)\bm I$, so $\widehat{\mathcal Q}_k^*(\bm t)\le\widehat{\mathbb G}_k^*(\bm t)-(\lambda_*/4)\|\bm t\|^2$. The canonical-metric bound \eqref{eq:s5_gp_conditional_metric}, $\widehat{\mathbb G}_k^*(\bm0)=0$, and Dudley's entropy inequality \citep{Adler2007} give $E^*\sup_{\|\bm t\|\le R}|\widehat{\mathbb G}_k^*(\bm t)|\le C_{M,p}\sqrt R$ for $1\le R\le R_k$. The Borell--TIS inequality \citep{Adler2007} then yields, uniformly over dyadic shells $\{2^jR<\|\bm t\|\le2^{j+1}R\}\cap B_{R_k}$,
\[
P^*\left\{\sup_{2^jR<\|\bm t\|\le2^{j+1}R}\widehat{\mathcal Q}_k^*(\bm t)\ge0\right\}\le2\exp\{-c2^{3j}R^3\},
\]
where empty shells have probability zero and, on boundary shells, the supremum is restricted to $B_{R_k}$. The summability of the right-hand side shows that conditional maximizers of $\widehat{\mathcal Q}_k^*$ over $B_{R_k}$ are tight, uniformly as $R_k\to\infty$. This bound also applies to grid maximizers because the grid contains the origin, at which the criterion equals zero.

For $\mathcal T_k=(\delta_k\mathbb Z^p)\cap B_{R_k}$ with $\delta_k=(\log k)^{-1}$, the Euclidean covering radius on each fixed $B_R$ tends to zero. The Gaussian metric bound \eqref{eq:s5_gp_conditional_metric} and chaining imply that, for every fixed $R$, the modulus of continuity of $\widehat{\mathbb G}_k^*$ over pairs of Euclidean distance at most $\sqrt p\,\delta_k$ tends to zero in conditional probability. The quadratic drift has stochastically bounded Lipschitz constant on $B_R$ because $\widehat{\bm H}_{\mathrm{eff}}=O_{P_0}(1)$. Hence the grid supremum approximates the continuous supremum uniformly on every fixed ball. By Lemma~\ref{lem:s4_functional}, $\mathcal Q$ has an almost surely unique global maximizer $\bm T^\star$ and tends to $-\infty$ at infinity. Conditional process convergence on fixed balls, tightness of the grid maximizers, and the approximate-argmax continuous-mapping theorem \citep{Kim1990,vanDerVaart1996} therefore give
\[
d_{\mathrm{BL}}\left\{\mathcal L(\widehat{\bm T}_k^*\mid D_n),\mathcal L(\bm T^\star)\right\}\xrightarrow{P_0}0.
\]
This is the conditional argmax consistency asserted in Proposition 6. The conclusion concerns the reference law for a single Gaussian-process draw under the theoretical screening and numerical regime; it does not assert nominal frequentist coverage for fixed-$J$ reference quantiles or for the subsequent Bayesian-bootstrap intercept calibration.
\end{proof}

\subsection{Bayesian-bootstrap intercept calibration and reported intervals}
\label{subsec:s5_bootstrap}

The slope reference sample $\{\bm b_{\mathrm{cal}}^{(j)}:j=1,\ldots,J\}$ is used to form marginal slope intervals. The intercept reference uses the same nonlinear residual-quantile transformation as the corrected point estimator, augmented by Bayesian-bootstrap weights \citep{Rubin1981,Lo1987}. Independently for $l=1,\ldots,L$, draw $J_l$ uniformly from $\{1,\ldots,J\}$ and $E_i^{(l)}\overset{\mathrm{iid}}\sim\mathrm{Exp}(1)$ for $i\in\widehat K_n$, independently of $J_l$. Define
\[
W_i^{(l)}=\frac{E_i^{(l)}}{\sum_{a\in\widehat K_n}E_a^{(l)}},\qquad i\in\widehat K_n.
\]
For a set of values $(v_i)_{i\in\widehat K_n}$, define the lower weighted empirical quantile by $\widehat Q_{\tau,\bm W^{(l)}}\{v_i\}=\inf\{v:\sum_{i\in\widehat K_n}W_i^{(l)}\mathbf1(v_i\le v)\ge\tau\}$. The intercept reference draw is
\[
\alpha_{\mathrm{cal}}^{(l)}=\widehat Q_{\tau,\bm W^{(l)}}\left\{Y_i-\bm z_i^\top\bm b_{\mathrm{cal}}^{(J_l)}:i\in\widehat K_n\right\},\qquad l=1,\ldots,L.
\]
The selection of $J_l$ transfers the simulated slope fluctuations into the intercept reference, while the normalized exponential weights retain empirical-quantile variation conditional on that slope reference. In particular, this second source of variation remains present when the leading intercept coordinate of the nonregular limit is degenerate because the limiting clean design is centered. The weighting is applied only to the residual-quantile calibration step and does not modify the MCMC posterior.

For nominal coverage $1-\zeta$, the marginal interval for slope coordinate $b_j$ is formed from the empirical $\zeta/2$ and $1-\zeta/2$ quantiles of $\{b_{\mathrm{cal},j}^{(1)},\ldots,b_{\mathrm{cal},j}^{(J)}\}$, while the intercept interval uses the corresponding empirical quantiles of $\{\alpha_{\mathrm{cal}}^{(1)},\ldots,\alpha_{\mathrm{cal}}^{(L)}\}$. The reported point estimates are $\widehat{\bm b}$ and $T_n(\widehat{\bm b})$, rather than the marginal medians of the simulated calibration references. In the main contamination-grid simulation with $n=100$, the implementation uses three MCMC chains of $T=10,000$ iterations each, discarding the first $B=5{,}000$, and the numerical priors $N_p(\bm0,1000\bm I_p)$ for $\bm\beta$, $\mathrm{IG}(0.1,0.1)$ for $\sigma$, and $\mathrm{Beta}(1,1)$ for $s$, with $\gamma=1$. The interval calculation uses $J=199$ Gaussian-process replications and $L=500$ Bayesian-bootstrap replications, as specified in Section 5 of the main manuscript. These settings apply to the stated experiment and are not additional theoretical assumptions. Bandwidth selection, covariance integration, and numerical argmax controls are part of the finite-scale implementation and should not be confused with quantities fixed by Theorem 2.

The nonregular asymptotic theory in Sections \ref{sec:supp_theorem1}--\ref{sec:supp_nonregular} establishes the sampling-limit mechanism motivating this calibration, but it does not establish nominal frequentist coverage for the finite-scale plug-in procedure, including estimated smoothing, numerical optimization, and practical screening. The calibrated intervals therefore remain a separately evaluated numerical procedure rather than credible intervals from a newly defined calibrated posterior.

\section{Additional Empirical Results}
\label{sec:supp_empirical}

This section supplements the empirical results in Section 5 of the main manuscript. We first give absolute point-estimation and interval-performance measures for representative contamination settings and compare original with quantile-corrected intercept estimates. We then describe the interval-calibration ablation and report MCMC diagnostics and holdout prediction results for the airline delay application. The numerical entries are calculated from the saved CSV summaries and replication-level outputs without further model fitting.

\subsection{Complete simulation results for robustness and quantile targeting}
\label{subsec:s6_simulation}

Table~\ref{tab:s6_absolute} reports absolute performance measures for the uncontaminated baseline and for selected one- and two-sided contamination settings. We fix $n=100$ and $\tau=0.5$ to separate the effect of outlier magnitude from that of the quantile level. The two contaminated scenarios use nominal contamination proportion $\varphi=0.10$ and magnitudes $\omega\in\{5,100\}$. In the two-sided case, 6\% of observations are contaminated on each side, giving a total contamination proportion of 12\%. For a fitted coefficient vector $\widehat{\bm\beta}$, the integrated squared conditional-quantile error is $\{\widehat\alpha-\alpha_\tau^*\}^2+(\widehat{\bm b}-\bm b^*)^\top\bm\Sigma(\widehat{\bm b}-\bm b^*)$ under the centered Gaussian design. We report $\mathrm{RMSE}_f$ as the square root of its Monte Carlo average. The coverage probability (CP), average interval length (AvL), and interval score (IS) are averages across the intercept and three slope coefficients; $\mathrm{CP}_\alpha$ records intercept coverage separately. The AL--LPAL intervals in this table are the GP--BB-calibrated intervals used in the main contamination-grid experiment.

\begin{table}[!htb]
\centering
\setlength{\tabcolsep}{5pt}
\renewcommand{\arraystretch}{1.06}
\begin{tabular}{lccccc}
\toprule
Method & $\mathrm{RMSE}_f$ & CP & $\mathrm{CP}_\alpha$ & AvL & IS \\
\midrule
\addlinespace[4pt]
\multicolumn{6}{l}{\textit{Panel A. Uncontaminated baseline}} \\
\addlinespace[2pt]
AL--LPAL & 0.132 & 0.950 & 0.964 & 0.312 & 0.375 \\
GAH & 0.133 & 0.912 & 0.930 & 0.271 & 0.380 \\
GAL & 0.129 & 0.946 & 0.948 & 0.293 & \textbf{0.356} \\
AL & 0.129 & 0.943 & 0.940 & 0.288 & 0.357 \\
Skew-$t$ & \textbf{0.126} & 1.000 & 1.000 & 0.560 & 0.560 \\
cGAL & 0.130 & 0.934 & 0.938 & 0.284 & 0.360 \\
\midrule
\addlinespace[4pt]
\multicolumn{6}{l}{\textit{Panel B. One-sided contamination}} \\
\midrule
\addlinespace[2pt]
\multicolumn{6}{l}{\hspace{1em}$\omega=5$} \\
AL--LPAL & 0.160 & 0.933 & 0.962 & 0.349 & 0.454 \\
GAH & 0.163 & 0.922 & 0.922 & 0.333 & 0.443 \\
GAL & 0.223 & 0.893 & 0.650 & 0.427 & 0.616 \\
AL & 0.180 & 0.961 & 0.920 & 0.431 & 0.489 \\
Skew-$t$ & \textbf{0.150} & 1.000 & 1.000 & 0.734 & 0.734 \\
cGAL & 0.159 & 0.919 & 0.910 & 0.314 & \textbf{0.434} \\
\addlinespace[2pt]
\multicolumn{6}{l}{\hspace{1em}$\omega=100$} \\
AL--LPAL & 0.158 & 0.943 & 0.960 & 0.349 & \textbf{0.440} \\
GAH & 0.193 & 0.968 & 0.920 & 0.482 & 0.525 \\
GAL & 2.152 & 0.750 & 0.000 & 2.719 & 12.791 \\
AL & 0.228 & 1.000 & 1.000 & 2.006 & 2.006 \\
Skew-$t$ & \textbf{0.142} & 1.000 & 1.000 & 0.754 & 0.754 \\
cGAL & 0.264 & 0.899 & 0.598 & 0.544 & 0.735 \\
\midrule
\addlinespace[4pt]
\multicolumn{6}{l}{\textit{Panel C. Two-sided contamination}} \\
\midrule
\addlinespace[2pt]
\multicolumn{6}{l}{\hspace{1em}$\omega=5$} \\
AL--LPAL & 0.151 & 0.949 & 0.968 & 0.351 & 0.431 \\
GAH & 0.149 & 0.945 & 0.974 & 0.341 & 0.415 \\
GAL & 0.154 & 0.996 & 0.998 & 0.500 & 0.506 \\
AL & 0.149 & 0.994 & 0.996 & 0.488 & 0.494 \\
Skew-$t$ & \textbf{0.139} & 1.000 & 1.000 & 0.801 & 0.801 \\
cGAL & 0.143 & 0.934 & 0.940 & 0.303 & \textbf{0.384} \\
\addlinespace[2pt]
\multicolumn{6}{l}{\hspace{1em}$\omega=100$} \\
AL--LPAL & 0.151 & 0.948 & 0.974 & 0.350 & \textbf{0.422} \\
GAH & 0.154 & 0.996 & 1.000 & 0.543 & 0.547 \\
GAL & 0.469 & 1.000 & 1.000 & 3.187 & 3.187 \\
AL & 0.199 & 1.000 & 1.000 & 2.779 & 2.779 \\
Skew-$t$ & \textbf{0.136} & 1.000 & 1.000 & 0.826 & 0.826 \\
cGAL & 0.159 & 1.000 & 1.000 & 0.683 & 0.683 \\
\bottomrule
\end{tabular}
\caption{Point and interval performance under selected contamination settings at $n=100$ and $\tau=0.5$. Values are Monte Carlo averages over 500 requested replications, with successful replications used for each method. The two-sided design has a total contamination rate of $1.2\varphi=0.12$. The intervals have nominal coverage 95\%. The smallest $\mathrm{RMSE}_f$ and IS within each setting are shown in bold.}
\label{tab:s6_absolute}
\end{table}

The absolute scores reveal a clear efficiency--robustness trade-off. Under clean sampling, GAL has the lowest IS (0.356), compared with 0.375 for calibrated AL--LPAL. At $\omega=5$, cGAL yields the lowest IS under both one-sided (0.434) and two-sided (0.384) contamination. At $\omega=100$, however, AL--LPAL achieves the lowest IS in both regimes: 0.440 under one-sided contamination and 0.422 under two-sided contamination, compared with 2.006 and 2.779, respectively, for AL. The corresponding AL--LPAL overall coverage rates are 0.943 and 0.948, with intercept coverage 0.960 and 0.974. Thus, the principal gain under extreme contamination is more reliable and economical interval inference; the smallest point-estimation RMSE is not necessarily attained by AL--LPAL in every setting.

Table~\ref{tab:s6_intercept} isolates the effect of draw-wise residual-quantile correction on intercept point estimation under 10\% positive one-sided contamination. For each replication, the signed errors of the original posterior-median intercept and the corrected intercept are evaluated against the same clean target $\alpha_\tau^*$. Their Monte Carlo means give signed bias, while the square roots of the mean squared errors give RMSE. The two estimators are compared on the same set of successful replications. Because Figure 3 in the main manuscript already displays bias over the full outlier-magnitude grid, we focus on $\omega=5$ and $100$ to contrast moderate and extreme contamination.

\begin{table}[!htb]
\centering
\setlength{\tabcolsep}{8pt}
\renewcommand{\arraystretch}{1.07}
\begin{tabular}{rrrrrr}
\toprule
& & \multicolumn{2}{c}{Bias} & \multicolumn{2}{c}{RMSE} \\
\cmidrule(lr){3-4}\cmidrule(lr){5-6}
$n$ & $\tau$ & Raw & Corrected & Raw & Corrected \\
\midrule
\addlinespace[4pt]
\multicolumn{6}{l}{\textit{Panel A. Moderate contamination ($\omega=5$)}} \\
\midrule
\addlinespace[2pt]
100 & 0.1 & 0.080 & 0.009 & 0.167 & 0.147 \\
100 & 0.5 & 0.037 & 0.007 & 0.086 & 0.071 \\
100 & 0.9 & -0.120 & -0.002 & 0.233 & 0.221 \\
\addlinespace[2pt]
250 & 0.1 & 0.088 & 0.021 & 0.125 & 0.092 \\
250 & 0.5 & 0.032 & 0.011 & 0.061 & 0.046 \\
250 & 0.9 & -0.146 & 0.060 & 0.180 & 0.161 \\
\addlinespace[2pt]
500 & 0.1 & 0.093 & 0.027 & 0.116 & 0.073 \\
500 & 0.5 & 0.027 & 0.016 & 0.047 & 0.035 \\
500 & 0.9 & -0.148 & 0.138 & 0.166 & 0.192 \\
\addlinespace[4pt]
\midrule
\multicolumn{6}{l}{\textit{Panel B. Extreme contamination ($\omega=100$)}} \\
\midrule
\addlinespace[2pt]
100 & 0.1 & 0.169 & 0.026 & 0.227 & 0.154 \\
100 & 0.5 & -0.001 & -0.003 & 0.077 & 0.072 \\
100 & 0.9 & -0.157 & -0.045 & 0.213 & 0.148 \\
\addlinespace[2pt]
250 & 0.1 & 0.172 & 0.034 & 0.196 & 0.097 \\
250 & 0.5 & 0.001 & -0.001 & 0.050 & 0.044 \\
250 & 0.9 & -0.170 & -0.029 & 0.197 & 0.097 \\
\addlinespace[2pt]
500 & 0.1 & 0.175 & 0.018 & 0.189 & 0.073 \\
500 & 0.5 & -0.001 & -0.001 & 0.037 & 0.031 \\
500 & 0.9 & -0.171 & -0.013 & 0.186 & 0.068 \\
\bottomrule
\end{tabular}
\caption{Signed bias and RMSE of the original and residual-quantile-corrected intercept estimators under 10\% positive one-sided contamination. Each setting has 500 requested Monte Carlo replications; both summaries use the same successful replications. Biases and RMSEs are measured relative to the clean intercept target $\alpha_\tau^*$.}
\label{tab:s6_intercept}
\end{table}

For extreme contamination, correction substantially reduces the intercept RMSE at both tail quantiles. At $n=500$ and $\omega=100$, RMSE falls from 0.189 to 0.073 at $\tau=0.1$ and from 0.186 to 0.068 at $\tau=0.9$, corresponding to reductions of approximately 61\% and 63\%. The original signed biases remain close to 0.175 and $-0.171$, whereas the corrected biases are 0.018 and $-0.013$. Under moderate contamination, the gain is not universal: at $n=500$, $\tau=0.9$, and $\omega=5$, correction changes the bias from $-0.148$ to 0.138 and increases RMSE from 0.166 to 0.192. Together, these results establish a substantial targeting improvement in the extreme-contamination regime while identifying a concrete finite-sample cost of correction in a less extreme setting.

\subsection{Additional assessment of interval calibration}
\label{subsec:s6_ablation}

Table~\ref{tab:s6_ablation} isolates the contributions of residual-quantile correction, Gaussian-process-based uncertainty calibration, and Bayesian-bootstrap weighting. We use the same simulation setting as Table 1 of the main manuscript: 10\% positive one-sided contamination with $\omega=100$, $\tau=0.5$, and $n\in\{100,250,500\}$. The four intercept interval constructions differ only in their post-processing of a common AL--LPAL fit. The original posterior intervals are compared with the posterior intervals after draw-wise residual-quantile correction, intervals obtained by propagating Gaussian-process slope fluctuations through an unweighted residual quantile, and the fully calibrated intervals that additionally incorporate Bayesian-bootstrap weighting. Since correction does not alter the slope posterior and bootstrap weighting is confined to the intercept, the slope comparison requires only the original and Gaussian-process-calibrated intervals.

\begin{table}[!htb]
\centering
\setlength{\tabcolsep}{4.5pt}
\renewcommand{\arraystretch}{1.10}
\begin{tabular}{clccc}
\toprule
$n$ & Method & CP & AvL & IS \\
\midrule
\multicolumn{5}{l}{\textit{Panel A. Intercept}} \\
\midrule
\addlinespace[2pt]
\multirow{4}{*}{$100$}
& Original & 0.892 (0.0139) & 0.252 (0.0022) & 0.383 (0.0240) \\
& Corrected & 0.586 (0.0220) & 0.113 (0.0015) & 0.813 (0.0541) \\
& GP calibrated & 0.612 (0.0218) & 0.124 (0.0014) & 0.738 (0.0518) \\
& Fully calibrated & 0.974 (0.0071) & 0.304 (0.0022) & 0.340 (0.0127) \\
\addlinespace[3pt]
\multirow{4}{*}{$250$}
& Original & 0.884 (0.0143) & 0.152 (0.0013) & 0.278 (0.0229) \\
& Corrected & 0.478 (0.0224) & 0.057 (0.0007) & 0.602 (0.0389) \\
& GP calibrated & 0.574 (0.0221) & 0.065 (0.0007) & 0.514 (0.0359) \\
& Fully calibrated & 0.964 (0.0083) & 0.184 (0.0011) & 0.219 (0.0104) \\
\addlinespace[3pt]
\multirow{4}{*}{$500$}
& Original & 0.854 (0.0158) & 0.106 (0.0010) & 0.201 (0.0137) \\
& Corrected & 0.390 (0.0218) & 0.034 (0.0004) & 0.483 (0.0265) \\
& GP calibrated & 0.468 (0.0223) & 0.040 (0.0004) & 0.409 (0.0241) \\
& Fully calibrated & 0.966 (0.0081) & 0.128 (0.0006) & 0.143 (0.0046) \\
\addlinespace[3pt]
\midrule
\multicolumn{5}{l}{\textit{Panel B. Average of three slopes}} \\
\midrule
\addlinespace[2pt]
\multirow{2}{*}{$100$}
& Original & 0.913 (0.0083) & 0.309 (0.0026) & 0.444 (0.0171) \\
& Calibrated & 0.952 (0.0063) & 0.367 (0.0024) & 0.443 (0.0134) \\
\addlinespace[3pt]
\multirow{2}{*}{$250$}
& Original & 0.884 (0.0091) & 0.184 (0.0014) & 0.308 (0.0153) \\
& Calibrated & 0.961 (0.0055) & 0.254 (0.0011) & 0.289 (0.0084) \\
\addlinespace[3pt]
\multirow{2}{*}{$500$}
& Original & 0.848 (0.0102) & 0.128 (0.0010) & 0.251 (0.0125) \\
& Calibrated & 0.955 (0.0061) & 0.190 (0.0007) & 0.220 (0.0058) \\
\bottomrule
\end{tabular}
\caption{Ablation analysis of nominal 95\% interval performance under 10\% positive one-sided contamination with $\omega=100$ and $\tau=0.5$. Empirical coverage probability (CP), average interval length (AvL), and interval score (IS) are followed by their Monte Carlo standard errors in parentheses, based on 500 matched replications at each sample size. All methods share the same AL--LPAL fit. The Gaussian-process-calibrated intercept intervals use unweighted residual quantiles; full calibration additionally incorporates Bayesian-bootstrap weighting. Slope summaries average the three coefficient-specific measures within each replication.}
\label{tab:s6_ablation}
\end{table}

The results demonstrate that correcting posterior centering does not by itself ensure accurate interval coverage. For the intercept, the original posterior intervals show persistent undercoverage, with CP decreasing from 0.892 at $n=100$ to 0.854 at $n=500$. Residual-quantile correction alone produces much shorter intervals but more pronounced undercoverage: CP falls from 0.586 to 0.390 as $n$ increases. Propagating the Gaussian-process slope reference through an unweighted residual quantile improves coverage only modestly, reaching 0.468 at $n=500$. By contrast, the fully calibrated intervals attain CPs of 0.974, 0.964, and 0.966 at $n=100$, 250, and 500, respectively. They also achieve lower interval scores than the original posterior at every sample size: 0.340 versus 0.383, 0.219 versus 0.278, and 0.143 versus 0.201. The comparison establishes that Bayesian-bootstrap variation in the residual quantile is essential for the substantial empirical coverage improvement observed here; quantile correction and Gaussian-process slope calibration alone do not recover nominal intercept coverage.

For the slopes, Gaussian-process calibration improves empirical coverage from 0.913, 0.884, and 0.848 under the original posterior to 0.952, 0.961, and 0.955, respectively. The corresponding interval scores are 0.443, 0.289, and 0.220, compared with 0.444, 0.308, and 0.251 for the unadjusted posterior intervals. Thus, the calibration procedure provides near-nominal slope coverage throughout and improves interval scores most clearly at larger sample sizes. Together, the intercept and slope results identify the distinct empirical roles of the Gaussian-process approximation and Bayesian-bootstrap variation in recovering well-calibrated uncertainty under severe contamination.

\subsection{Additional results for the airline delay data}
\label{subsec:s6_airline}

\begin{table}[!htb]
\centering
\setlength{\tabcolsep}{11pt}
\begin{tabular}{crrrr}
\toprule
$\tau$ & Max. $\widehat R$ & Min. bulk ESS & Min. tail ESS \\
\midrule
0.10 & 1.002 & 3,624 & 4,649 \\
0.25 & 1.002 & 1,871 & 2,822 \\
0.50 & 1.002 & 2,064 & 4,269 \\
0.75 & 1.005 & 719 & 1,809 \\
0.90 & 1.013 & 489 & 1,377 \\
\bottomrule
\end{tabular}
\caption{MCMC convergence diagnostics for the airline delay application. Maxima and minima are taken over the seven regression coefficients, the scale $\sigma$, and the mixing weight $s$ (nine monitored parameters per quantile). The largest $\widehat R$, 1.0131 at $\tau=0.90$, arises for the log-distance coefficient; its bulk and tail ESS are 654 and 1,772.}
\label{tab:s6_mcmc}
\end{table}

\begin{figure}[!htbp]
\centering
\begin{minipage}[t]{0.49\textwidth}
\centering
\includegraphics[width=\linewidth]{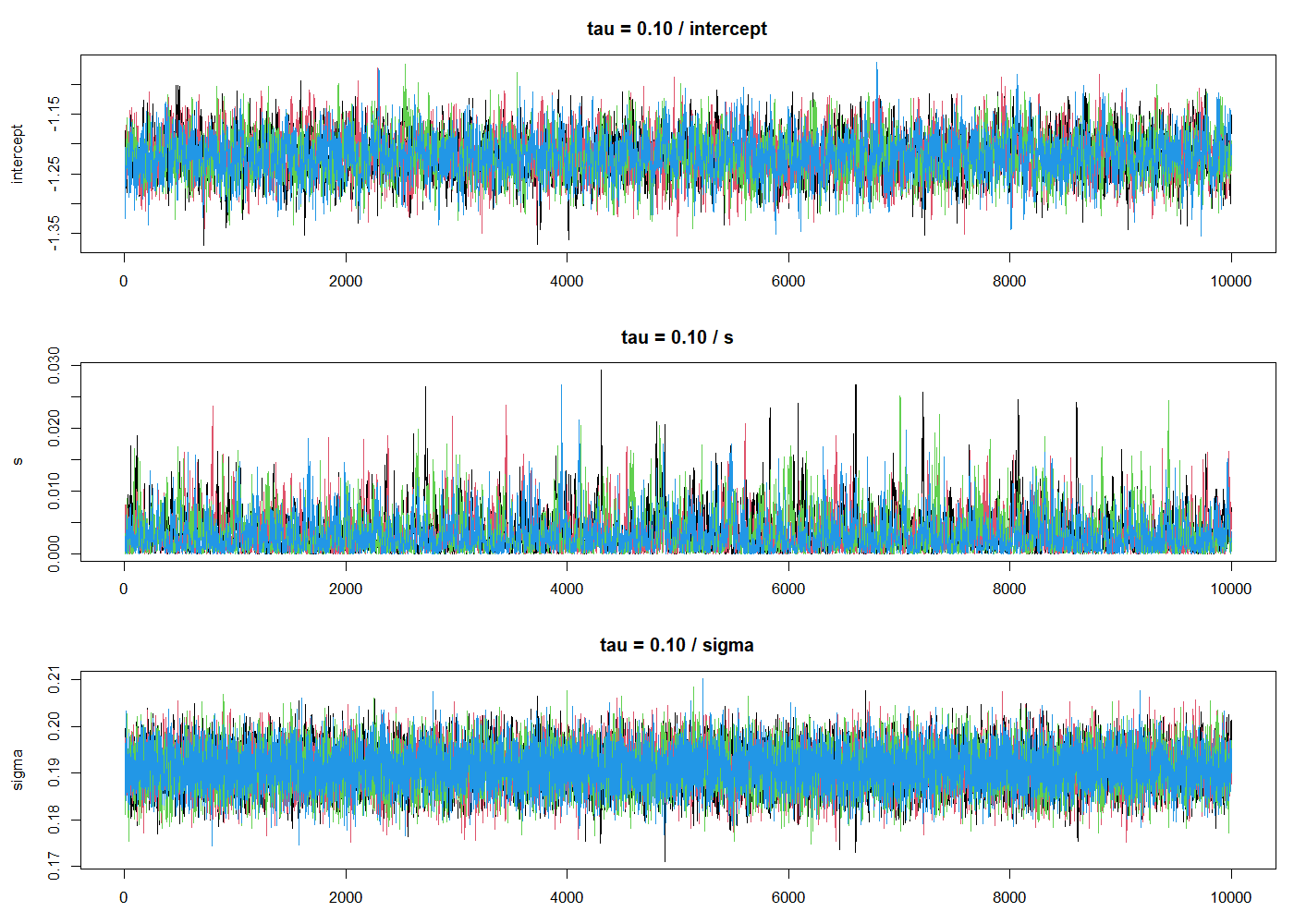}
\vspace{-0.5em}
\par\small (a) $\tau=0.10$
\end{minipage}\hfill
\begin{minipage}[t]{0.49\textwidth}
\centering
\includegraphics[width=\linewidth]{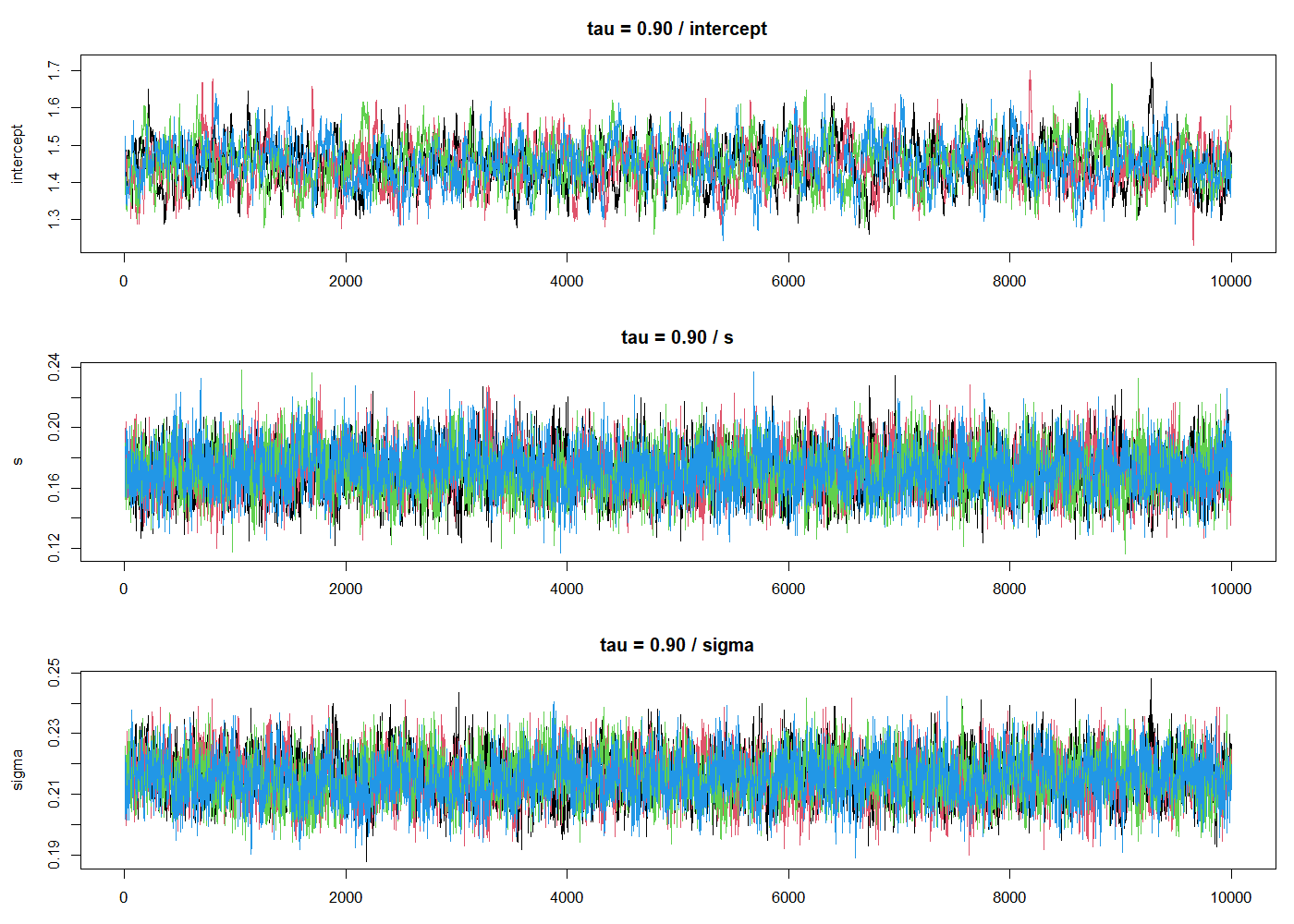}
\vspace{-0.5em}
\par\small (b) $\tau=0.90$
\end{minipage}

\caption{Posterior trace plots for the intercept, mixture weight $s$, and scale parameter $\sigma$ in the airline delay application. Four independently initialized MCMC chains are displayed at $\tau=0.10$ and $\tau=0.90$. The traces show substantial overlap across chains without persistent trends.}
\label{fig:s6_traces}
\end{figure}

Table~\ref{tab:s6_mcmc} summarizes rank-normalized split-$\widehat R$ and bulk and tail effective sample sizes (ESS) \citep{Vehtari2021} for the AL--LPAL fits in the airline delay application. The summaries are deliberately conservative: at each quantile level we report the maximum $\widehat R$ and the minimum bulk and tail ESS over all seven regression coefficients together with $\sigma$ and $s$. All five quantile levels show substantial effective sampling. Across all combinations, the minimum bulk and tail ESS are 489 and 1,377, respectively. The representative chain traces in Figure~\ref{fig:s6_traces} show substantial overlap and no persistent drift for the intercept, mixture weight, and scale under the two selected quantile levels.

\begin{table}[!htb]
\centering
\setlength{\tabcolsep}{8pt}
\begin{tabular}{lrrrrr}
\toprule
& \multicolumn{5}{c}{Quantile level $\tau$} \\
\cmidrule(lr){2-6}
Method & $0.10$ & $0.25$ & $0.50$ & $0.75$ & $0.90$ \\
\midrule
QR & 4.126 & 8.344 & 12.736 & 13.382 & 9.949 \\
AL & 4.120 & 8.345 & \textbf{12.733} & \textbf{13.374} & \textbf{9.949} \\
GAL & 4.099 & 8.371 & 12.859 & 13.588 & 10.376 \\
GAH & 4.112 & 8.364 & 12.799 & 13.645 & 10.716 \\
cGAL & \textbf{4.099} & 8.362 & 12.832 & 13.711 & 11.280 \\
Skew-$t$ & 4.105 & 8.343 & 12.760 & 13.581 & 10.648 \\
Original AL--LPAL & 4.120 & 8.343 & 12.744 & 13.588 & 10.800 \\
Corrected AL--LPAL & 4.117 & \textbf{8.340} & 12.744 & 13.504 & 10.186 \\
\bottomrule
\end{tabular}
\caption{Mean holdout check loss for the airline delay application, evaluated on the same 10,000 test observations for each method. Values are expressed in the original arrival-delay units, with smaller values indicating better predictive performance. The smallest value in each column is highlighted in bold based on the unrounded scores. Original AL--LPAL uses the original posterior coefficient estimates, whereas Corrected AL--LPAL uses the estimates obtained after residual-quantile correction. Interval calibration is not involved in computing point predictions.}
\label{tab:s6_airline_loss}
\end{table}

Table~\ref{tab:s6_airline_loss} reports the mean holdout check loss for all competing methods across five quantile levels. Corrected AL--LPAL delivers competitive predictive performance throughout, with check losses within 2.4\% of the best-performing method at every quantile level and the lowest loss at $\tau=0.25$. The benefit of residual-quantile correction is particularly pronounced at $\tau=0.90$, where the mean check loss decreases from 10.800 for Original AL--LPAL to 10.186 for Corrected AL--LPAL, corresponding to a reduction of approximately 5.7\%. These results demonstrate that residual-quantile correction substantially improves upper-tail prediction while preserving competitive predictive accuracy across the remaining quantile levels.

\bibliographystyle{chicago}
\bibliography{bibliography}